%% file: neurips_2026.tex
\documentclass{article}

\PassOptionsToPackage{numbers, compress}{natbib}

\usepackage[main, final]{neurips_2026}

\usepackage[utf8]{inputenc} 
\usepackage[T1]{fontenc}    
\usepackage{url}            
\usepackage{booktabs}       
\usepackage{amsfonts}       
\usepackage{nicefrac}       
\usepackage{microtype}      
\usepackage{xcolor}         

\usepackage{amsmath}
\usepackage{amssymb}
\usepackage{mathtools}
\usepackage{amsthm}
\usepackage{tikz}
\usetikzlibrary{calc, arrows.meta, decorations.markings, shapes.geometric}
\usepackage{subcaption}
\usepackage{float}
\usepackage{physics}
\usepackage{multirow}
\usepackage[colorlinks=true, citecolor=blue, linkcolor=blue, urlcolor=blue]{hyperref}

\usepackage[ruled,vlined]{algorithm2e}

\usepackage[capitalize,noabbrev,nameinlink]{cleveref}
\crefname{algorithm}{algorithm}{algorithms}
\Crefname{algorithm}{Algorithm}{Algorithms}
\crefname{equation}{Eq.}{Eqs.}
\Crefname{equation}{Eq.}{Eqs.}

\theoremstyle{plain}
\newtheorem{theorem}{Theorem}[section]
\newtheorem{proposition}[theorem]{Proposition}
\newtheorem{lemma}[theorem]{Lemma}
\newtheorem{corollary}[theorem]{Corollary}
\theoremstyle{definition}
\newtheorem{definition}[theorem]{Definition}

\theoremstyle{remark}
\newtheorem{remark}[theorem]{Remark}

\AddToHook{env/definition/begin}{%
  \crefalias{theorem}{definition}}
\AddToHook{env/proposition/begin}{%
  \crefalias{theorem}{proposition}}
\AddToHook{env/lemma/begin}{%
  \crefalias{theorem}{lemma}}
\AddToHook{env/corollary/begin}{%
  \crefalias{theorem}{corollary}}
\AddToHook{env/assumption/begin}{%
  \crefalias{theorem}{assumption}}
\AddToHook{env/remark/begin}{%
  \crefalias{theorem}{remark}}

\crefname{definition}{Definition}{Definitions}
\Crefname{definition}{Definition}{Definitions}
\crefname{proposition}{Proposition}{Propositions}
\Crefname{proposition}{Proposition}{Propositions}
\crefname{lemma}{Lemma}{Lemmas}
\Crefname{lemma}{Lemma}{Lemmas}
\crefname{corollary}{Corollary}{Corollaries}
\Crefname{corollary}{Corollary}{Corollaries}
\crefname{assumption}{Assumption}{Assumptions}
\Crefname{assumption}{Assumption}{Assumptions}
\crefname{remark}{Remark}{Remarks}
\Crefname{remark}{Remark}{Remarks}

\title{Drift-React: One-step Generation of Reaction Pathways via SE(3) Drifting Fields}

\author{
  Rémi Schlama$^{1,2}$ \quad Philippe Schwaller$^{1,2}$ \\[0.2cm]
  $^{1}$Laboratory of Artificial Chemical Intelligence (LIAC), EPFL, Lausanne, Switzerland \\
  $^{2}$National Centre of Competence in Research (NCCR) Catalysis, EPFL, Lausanne, Switzerland \\[0.15cm]
  \texttt{\{remi.schlama,philippe.schwaller\}@epfl.ch}
}

\begin{document}

\maketitle

\begin{abstract}
  Mapping reaction pathways and transition states (TS) is fundamental to chemistry but computationally expensive at scale. The minimum energy pathway (MEP) dictates reaction rates and mechanisms, yet recovering it via electronic-structure methods requires thousands of costly force evaluations. Recent generative models accelerate TS identification but rely on iterative inference and only predict isolated saddle-point snapshots, missing the continuous reaction trajectory. We introduce Drift-React, an $\mathrm{SE}(3)$-equivariant generative framework that predicts complete reaction pathways in a single forward pass from only reactant and product geometries. By shifting distribution evolution to training via a Sinkhorn-weighted drifting field, Drift-React eliminates both the iterative force evaluations of NEB-style methods and the sequential ODE/SDE integration of diffusion and flow matching models. Evaluated on the Transition1x and Halo8 datasets, our one-step model generates geometrically accurate reaction-pathway candidates that closely track DFT reference pathways. For isolated TS prediction, Drift-React matches the sub-Ångström accuracy of state-of-the-art iterative models while delivering orders-of-magnitude acceleration, clearing a major computational bottleneck for large-scale reaction network exploration.
\end{abstract}

\section{Introduction}

The rate and selectivity of a chemical reaction are governed by the geometry and energy of a single, fleeting nuclear configuration: the transition state (TS). Formally, the TS is the first-order saddle point on the Born--Oppenheimer potential energy surface (PES), located at the apex of the minimum energy pathway (MEP) connecting reactant and product basins. Within transition state theory~\cite{eyringActivatedComplexChemical1935} and its variational extensions~\cite{truhlarVariationalTransitionState,fukuiVariationalPrinciplesChemical1981}, the activation free energy $\Delta G^\ddagger$, derived from the electronic structure and vibrational partition functions of the TS, dictates the reaction rate via the Eyring equation. Accurate TS prediction is therefore indispensable to the rational design of catalysts~\cite{kozuchHowConceptualizeCatalytic2011}, enzymes~\cite{kissComputationalEnzymeDesign2013}, and pharmaceutical compounds~\cite{schrammEnzymaticTransitionStates2018}. Classical computational approaches, chief among them the Nudged Elastic Band (NEB) method~\cite{henkelmanClimbingImageNudged2000,lindgrenScaledDynamicOptimizations2019} and the String Method~\cite{eStringMethodStudy2002}, locate the TS by iteratively relaxing a discretized chain of nuclear configurations along the MEP under \textit{ab initio} forces. Each force evaluation requires solving the Kohn--Sham equations to self-consistency at $\mathcal{O}(N^3)$ cost with standard implementations~\cite{goedeckerLinearScalingElectronic1999}, rendering these methods expensive at scale for high-throughput screening~\cite{gomez-bombarelliDesignEfficientMolecular2016} of even moderately sized organic molecules~\cite{unkeMachineLearningForce2021}.

Machine learning offers a promising alternative~\cite{beagleholeMachineLearningTransition2025,xuECTSUltrafastDiffusion2025,darouichTrainingDomainRobust2026,wuMachineLearnedLeftmostHessian2026,nikitinRightSaddleStereochemistryAware2026}. Diffusion- and flow-based models~\cite{duanAccurateTransitionState2023,kimDiffusionbasedGenerativeAI2024,duanOptimalTransportGenerating2025,tuoFlowMatchingReaction2025} generate accurate TS candidates without electronic-structure calculations during the model forward pass, but require iterative SDE/ODE integration and typically predict an isolated TS without the surrounding reaction-coordinate context. Methods that recover full-pathway information, such as NeuralNEB~\cite{schreinerNeuralNEBNeuralNetworks2022} and MEPIN~\cite{namTransferableLearningReaction2025}, still rely on iterative relaxation, energy-based training, or handcrafted geometric priors.

We introduce Drift-React, an $\mathrm{SE}(3)$-equivariant pathway generator based on Drifting Models~\cite{dengGenerativeModelingDrifting2026}. It maps simple priors such as linear interpolations to geometrically accurate pathway candidates in one forward pass, avoiding both sequential generative inference and iterative NEB-style force evaluations. As with other generative TS methods, these outputs are candidate geometries requiring subsequent quantum-chemical verification. Nevertheless, Drift-React provides strong initialization: on a 50-reaction Halo8~\cite{leeDatasetChemicalReaction2025} subset, 98\% of candidates converge under TS optimization, with 96\% of converged structures verified as first-order saddle points and 100\% recovering the correct IRC connectivity.

Our main contributions are:
\begin{enumerate}
    \item \textbf{One-step pathway generator.} We generate a complete reaction-pathway candidate approximating a DFT reference MEP from a reactant--product pair in a single network evaluation.
    \item \textbf{$\mathrm{SE}(3)$-equivariant drifting field for reaction pathways.} We extend the Drifting Models paradigm~\cite{dengGenerativeModelingDrifting2026} via (i) per-sample Kabsch alignment with a LEFTNet~\cite{duNewPerspectiveBuilding2023} backbone, and (ii) a specialization of the joint-affinity construction to the single-positive regime characteristic of reaction-pathway data. Ablations against direct supervised regression and naive coupling validate the joint-affinity construction.
    \item \textbf{Accurate pathway generation and effective TS initialization.} Drift-React achieves the best geometric fidelity among the evaluated full-pathway baselines on Transition1x~\cite{schreinerTransition1xDatasetBuilding2022} and Halo8~\cite{leeDatasetChemicalReaction2025}, while providing high-quality initializations for downstream quantum-chemical TS refinement.
    \item \textbf{Unified pathways and TS treatment.} Reducing the pathway to three images recovers the standard reactant--TS--product setting without modifying architecture or training. Drift-React retains competitive TS-RMSD while requiring a single network evaluation per reaction, compared with iterative inference in existing diffusion- and flow-based generators.
\end{enumerate}

\begin{figure}[h]
    \centering
    \includegraphics[width=\linewidth]{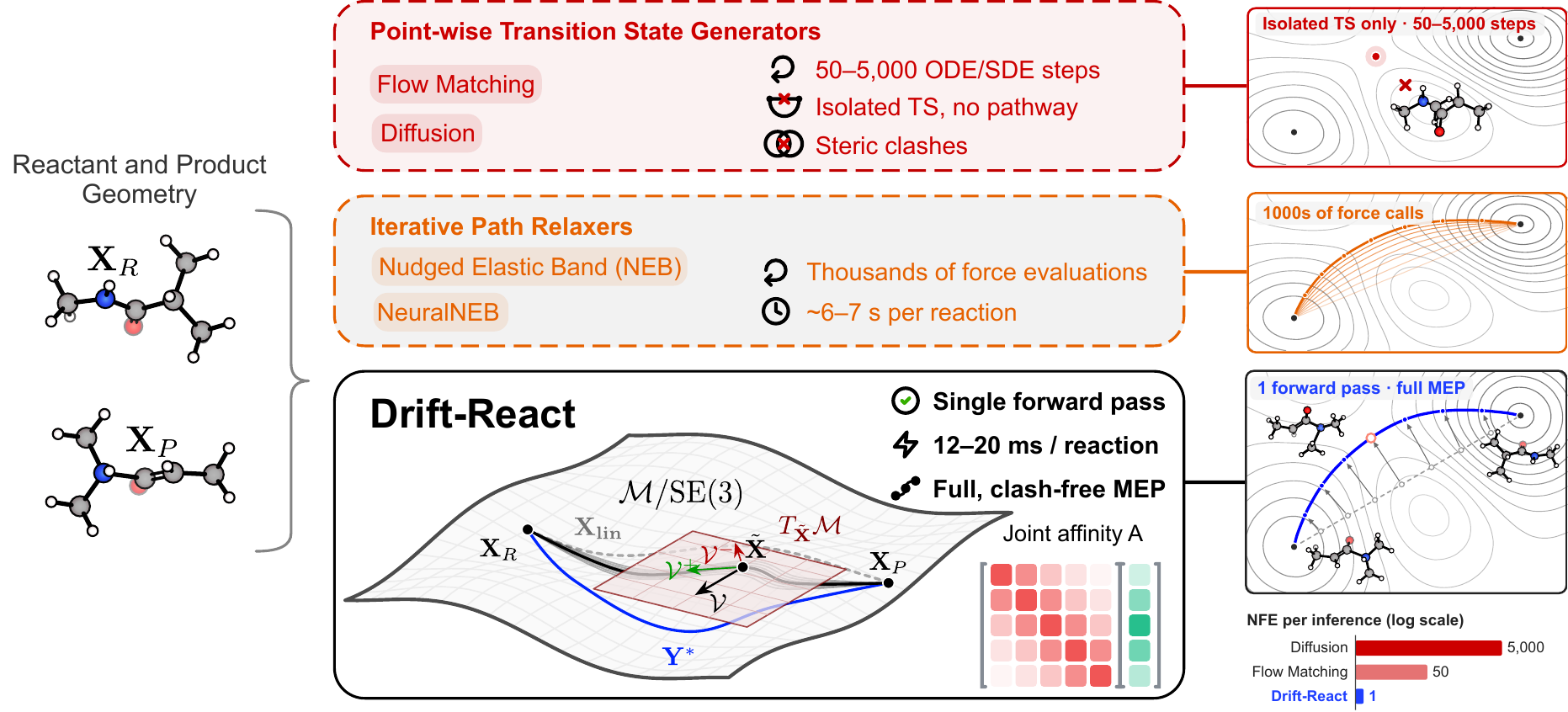}
    \caption{\textbf{Three paradigms for reaction pathway prediction.} \textbf{Top:} Point-wise generators (e.g., diffusion) predict isolated TS geometries at fast inference but ignore the surrounding PES valley, often producing structures with steric clashes. \textbf{Middle:} Iterative chain-of-states methods (e.g., NEB) recover the full MEP but require thousands of force evaluations per reaction. \textbf{Bottom:} Drift-React generates the complete reaction pathway in a single forward pass ($\sim$12~ms) by learning an $\mathrm{SE}(3)$-equivariant drifting field that balances target attraction with swarm repulsion.}
\label{fig:overview}
\end{figure}

\section{Reaction Pathway Generation via Drifting}

\subsection{Potential Energy Surface and Reaction Pathways}

Under the Born--Oppenheimer approximation, a molecular transformation occurs on a high-dimensional potential energy surface (PES) $V: \mathbb{R}^{3N} \to \mathbb{R}$, where $\mathbf{X} \in \mathbb{R}^{3N}$ denotes the Cartesian coordinates of $N$ atoms. The reactant and product geometries, $\mathbf{X}_R$ and $\mathbf{X}_P$, occupy stable local minima of this surface, characterized by a vanishing gradient $\nabla V = 0$ and a Hessian that is positive-definite in the vibrational subspace. Connecting these minima requires navigating the full $3N-6$ dimensional space of internal degrees of freedom, a dimensionality that renders exhaustive grid-based searches intractable and necessitates targeted pathway algorithms. The transition state $\mathbf{X}^\ddagger$ is the first-order saddle point of $V$, with a vanishing gradient and a single negative Hessian eigenvalue along the reaction coordinate~\cite{schlegelGeometryOptimization2011}.

Fukui's intrinsic reaction coordinate (IRC)~\cite{fukuiPathChemicalReactions1981,tachibanaDifferentialGeometryChemically1978} formalizes the pathway connecting $\mathbf{X}^\ddagger$ to both minima as the steepest-descent path on the PES, defining the minimum energy pathway (MEP) that we denote $\mathbf{Y}^\star$. A continuous MEP is a curve $\gamma: [0,1] \to \mathbb{R}^{3N}$ with $\gamma(0) = \mathbf{X}_R$, $\gamma(1) = \mathbf{X}_P$, along which the potential gradient has no component perpendicular to the path tangent $\hat{\boldsymbol{\tau}}(\alpha) = \dot\gamma(\alpha)/\lVert\dot\gamma(\alpha)\rVert$~\cite{eStringMethodStudy2002}:
\begin{equation}
\bigl[\nabla V(\gamma(\alpha))\bigr]^{\perp} = 0, \quad \text{where}\; [\nabla V]^\perp := \nabla V - (\nabla V \cdot \hat{\boldsymbol{\tau}}) \hat{\boldsymbol{\tau}}.
\label{eq:mep}
\end{equation}

Classical chain-of-states methods discretize $\gamma$ into $F$ images $\{\mathbf{Y}^\star_1, \ldots, \mathbf{Y}^\star_F\}$ and iteratively relax them under \textit{ab initio} forces. The most widely used are the nudged elastic band (NEB)~\cite{henkelmanClimbingImageNudged2000} and the string method~\cite{eStringMethodStudy2002}. This procedure routinely demands thousands of quantum-mechanical gradient evaluations per reaction and constitutes the primary computational bottleneck in transition state search. Our goal is to replace this iterative inner loop with a one-step generator that amortizes the cost across a training set of precomputed reference MEPs.

\subsection{The Drifting Generative Framework}

We train a non-iterative neural network $f_\theta: (\mathbb{R}^{3N})^F \to (\mathbb{R}^{3N})^F$ that maps a prior pathway to a reaction pathway candidate in a single forward pass, where $F$ is the number of images discretizing the MEP. We define $p_{\text{prior}}$ as a (possibly noised) interpolation between $\mathbf{X}_R$ and $\mathbf{X}_P$, which encodes boundary conditions but lacks barrier geometry. The generator induces a pushforward distribution $q_\theta = [f_\theta]_{\#}\,p_{\text{prior}}$, and the objective is to learn $f_\theta$ such that $q_\theta \approx p_{\text{data}}$, where $p_{\text{data}}$ is the distribution of reference MEP trajectories (e.g., from Transition1x~\cite{schreinerTransition1xDatasetBuilding2022} or Halo8~\cite{leeDatasetChemicalReaction2025}).

Rather than enforcing this match through iterative sampling, we follow the Drifting Models paradigm of~\citet{dengGenerativeModelingDrifting2026}. At training step $s$, each generated sample $\mathbf{X}^{(s)} = f_{\theta^{(s)}}(\mathbf{Z})$ with $\mathbf{Z} \sim p_{\text{prior}}$ undergoes an implicit evolution governed by an anti-symmetric drifting field $\mathcal{V}_{p,q}$:
\begin{equation}
    \mathbf{X}^{(s+1)} = \mathbf{X}^{(s)} + \mathcal{V}_{p, q_{\theta^{(s)}}}(\mathbf{X}^{(s)}),
\end{equation}
where the residual displacement is realized through the parameter update $\theta^{(s)} \to \theta^{(s+1)}$. The anti-symmetry of the continuous drifting ($\mathcal{V}_{p,q} = -\mathcal{V}_{q,p}$) ensures the generator stops drifting precisely when $q_\theta$ has recovered the target data distribution. \Cref{fig:drifting_method} illustrates this evolution: at each training step, the empirical drifting field transports a swarm of $K$ generated pathways toward the reference MEP, with the swarm acting as both the evaluation points of $\mathcal{V}$ and the empirical samples of $q_\theta$.

\begin{figure*}[h]
    \centering
    \resizebox{\textwidth}{!}{%
    \input{figures/tikz/drifting}
    }
    \caption{%
    \textbf{Drifting model for reaction pathways.}
    Panels show a configuration space with iso-energy contours. Endpoints $\mathbf{X}_R,\mathbf{X}_P$ are pinned by a sinusoidal envelope on the displacement output.
    \textbf{(a)} Prior $\mathbf{Z}$: linear interpolation plus Gaussian noise $\boldsymbol{\epsilon}\sim\mathcal{N}(\mathbf{0},\sigma^{2}\mathbf{I})$.
    \textbf{(b)} The drifting field \cref{eq:drift_final} moves each swarm member toward the reference MEP. $f_\theta$ is regressed onto $\mathrm{sg}(\mathbf{X}+\mathcal{V})$, and the step is repeated.
    \textbf{(c)} At the ideal equilibrium, $q_\theta=p_{\text{data}}$ and $\mathcal{V}\equiv\mathbf{0}$.
    \textbf{(d)} At inference, $f_{\theta^\star}$ maps a fresh prior sample to a complete pathway candidate in a single forward pass.
}
    \label{fig:drifting_method}
\end{figure*}

\subsection{Empirical Drifting Field and \texorpdfstring{$\mathrm{SE}(3)$}{SE(3)}-Equivariance}

Following the kernelized mean-shift construction of~\citet{dengGenerativeModelingDrifting2026}, we decompose the drifting field into attractive (toward $p_{\text{data}}$) and repulsive (away from $q_\theta$) contributions. With kernel $k(\mathbf{X}, \mathbf{Y}) = \exp(-\lVert\mathbf{X} - \mathbf{Y}\rVert/\tau)$ and partition functions $Z_p(\mathbf{X}) := \mathbb{E}_{\mathbf{Y}^+ \sim p_{\text{data}}}[k(\mathbf{X}, \mathbf{Y}^+)]$, $Z_q(\mathbf{X}) := \mathbb{E}_{\mathbf{Y}^- \sim q_\theta}[k(\mathbf{X}, \mathbf{Y}^-)]$, the continuous drifting field of~\citet{dengGenerativeModelingDrifting2026} admits the product-kernel form
\begin{equation}\label{eq:deng_continuous}
    \mathcal{V}_{p, q}(\mathbf{X}) = \frac{\mathbb{E}\bigl[k(\mathbf{X}, \mathbf{Y}^+)\, k(\mathbf{X}, \mathbf{Y}^-)\, (\mathbf{Y}^+ - \mathbf{Y}^-)\bigr]}{Z_p(\mathbf{X})\, Z_q(\mathbf{X})},
\end{equation}
which is anti-symmetric under $(p, q) \leftrightarrow (q, p)$ and vanishes when the two distributions match. Our reaction-pathway setting specializes Eq.~\eqref{eq:deng_continuous} to a single positive sample $\mathbf{Y}^\star$ paired with a swarm of $K$ negatives drawn from $q_\theta$. Naive Monte Carlo estimation in this regime is degenerate, because the kernel weight on the single positive cancels its own empirical partition function. The joint normalization introduced below resolves this by treating $\mathbf{Y}^\star$ and the swarm members on equal footing. We defer the full analysis to~\Cref{app:empirical}.

We evaluate this field empirically over a swarm of $K$ generated paths $\{\mathbf{X}_1, \ldots, \mathbf{X}_K\} \sim q_\theta$ and a reference MEP $\mathbf{Y}^\star \sim p_{\text{data}}$. To stabilize swarm dynamics and approximate optimal transport, we employ a jointly normalized Sinkhorn-style affinity matrix $\mathbf{A}$ (following~\citet[Algorithm 2]{dengGenerativeModelingDrifting2026}). Defining the augmented target set $\mathcal{T} = [\mathbf{Y}^\star, \mathbf{X}_1, \ldots, \mathbf{X}_K]$, we compute pairwise logits $L_{ij} = -\lVert \tilde{\mathbf{X}}_i - \mathcal{T}_j \rVert / \tau$ in a Kabsch-aligned frame (denoted by $\tilde{\cdot}$) and obtain the symmetric affinity matrix
\begin{equation}
    \mathbf{A}_{ij} = \sqrt{\sigma_{\mathrm{row}}(L)_{ij}\;\sigma_{\mathrm{col}}(L)_{ij}}.
\end{equation}
Partitioning $\mathbf{A}$ into an attractive affinity $A^{+}_i := \mathbf{A}_{i,0}$ and a repulsive affinity $A^{-}_{ij} := \mathbf{A}_{ij}$ for $j \ge 1$ yields the empirical drifting field in the aligned frame:
\begin{equation}\label{eq:drift_final}
    \tilde{\mathcal{V}}(\tilde{\mathbf{X}}_i) = A^{+}_i \sum_{j=1}^{K} A^{-}_{ij}\,(\mathbf{Y}^\star - \tilde{\mathbf{X}}_j),
\end{equation}
which is then mapped back to the original frame by the inverse Kabsch rotation~\cite{kabschSolutionBestRotation1976}, $\mathcal{V}(\mathbf{X}_i) = R_i^\top\,\tilde{\mathcal{V}}(\tilde{\mathbf{X}}_i)$. See \cref{app:empirical} for the full alignment-and-inverse-rotation protocol.

\textbf{Equivariance.} Because physical reaction pathways are invariant under rigid-body motions, we parameterize $f_\theta$ using LEFTNet~\cite{duNewPerspectiveBuilding2023}, guaranteeing $\mathrm{SE}(3)$-equivariant vector updates. The drifting field $\mathcal{V}$ is $\mathrm{SE}(3)$-equivariant by construction (see \Cref{app:equivariance}). The Sinkhorn affinities $A^{+}_i, A^{-}_{ij}$ are scalar invariants because they are computed from Kabsch-aligned pairwise distances. The displacement vectors $(\mathbf{Y}^\star - \tilde{\mathbf{X}}_j)$ transform covariantly under rotations applied jointly to all configurations. The inverse rotation $R_i^\top$ then maps the drift from the aligned frame back to the generator's output frame, preserving equivariance under the rigid-body action that aligned the inputs.

\textbf{Endpoint pinning.} To enforce the boundary conditions $f_\theta(\mathbf{Z})_1 = \mathbf{X}_R$ and $f_\theta(\mathbf{Z})_F = \mathbf{X}_P$ exactly, we apply a sinusoidal envelope $w(\alpha_k) = \sin(\pi \alpha_k)$ with $\alpha_k = (k-1)/(F-1)$ to the per-image displacement output of $f_\theta$. Since $w(\alpha_1) = w(\alpha_F) = 0$, the network predicts only the $F-2$ intermediate images. The reactant and product geometries are preserved regardless of network weights.

\textbf{Training.} The equilibrium condition $\mathcal{V} = \mathbf{0}$ motivates fixed-point loss enforced via the stop-gradient operator ($\mathrm{sg}$), following~\citet{dengGenerativeModelingDrifting2026}
\begin{equation}\label{eq:loss}
    \mathcal{L}(\theta) = \frac{1}{K} \sum_{i=1}^K \bigl\lVert \mathbf{X}_i - \mathrm{sg}\bigl(\mathbf{X}_i + \mathcal{V}(\mathbf{X}_i)\bigr)\bigr\rVert^2.
\end{equation}
Both $\mathbf{X}_i$ and $\mathcal{V}(\mathbf{X}_i)$ live in the original (unaligned) frame, so the regression target lies in the same frame as the generator output. Minimizing $\mathcal{L}(\theta)$ drives the swarm toward equilibrium without requiring back-propagation through the complex distribution dependence of $\mathcal{V}$.

\textbf{Inference.} At test time, $f_\theta$ is evaluated once on the linear interpolation between $\mathbf{X}_R$ and $\mathbf{X}_P$ to produce a complete $F$-image pathway candidate in a single forward pass. The drifting field $\mathcal{V}$, swarm sampling, and Sinkhorn affinities are training-time constructs only, eliminating the iterative inference required by diffusion and flow-based models.

\subsection{Algorithm}

\Cref{alg:drifting} summarizes the full training procedure. At each iteration, we sample a batch of $K$ linear interpolations from the prior, push them through the generator to produce a swarm $\{\mathbf{X}_i\}_{i=1}^K$. We then assemble the augmented target set with the reference MEP, compute the Sinkhorn affinities, evaluate the drifting field defined in Eq.~\eqref{eq:drift_final}, and take a gradient step on the fixed-point loss of Eq.~\eqref{eq:loss}.

\begin{algorithm}[H]
\DontPrintSemicolon
\caption{Training the Drifting Generative Model}
\label{alg:drifting}

\KwIn{Reference pathway $\mathbf{Y}^\star$, prior $p_{\text{prior}}$, generator $f_\theta$, swarm size $K$, temperature $\tau$, learning rate $\eta$}
\KwOut{Updated generator weights $\theta$}
\While{not converged}{
    \For{$i = 1, \ldots, K$}{
        $\mathbf{Z}_i \sim p_{\text{prior}}(\mathbf{X}_R, \mathbf{X}_P)$ \tcp*{Sample prior}
        $\mathbf{X}_i \gets f_\theta(\mathbf{Z}_i)$ \tcp*{Generate pathway}
    }
    $\mathcal{T} \gets \{\mathbf{Y}^\star\} \cup \{\mathbf{X}_1, \dots, \mathbf{X}_K\}$ \tcp*{Augmented target set}

    \BlankLine
    Kabsch-align $\{\mathbf{X}_i\} \to \{\tilde{\mathbf{X}}_i\}$ with rotations $\{R_i\}$
    
    $L_{ij} \gets -\|\tilde{\mathbf{X}}_i - \mathcal{T}_j\| / \tau$
    
    $\mathbf{A} \gets \sqrt{\sigma_{\text{row}}(L) \odot \sigma_{\text{col}}(L)}$ \tcp*{Joint softmax}
    $A^+_i \gets \mathbf{A}_{i,0}$ \tcp*{Attractive affinity}
    $A^-_{ij} \gets \mathbf{A}_{ij} \; \forall j \in \{1 \dots K\}$ \tcp*{Repulsive affinity}

    \BlankLine
    \tcp{Drifting field}
    $\tilde{\mathcal{V}}(\tilde{\mathbf{X}}_i) \gets A^+_i \sum_{j=1}^K A^-_{ij} (\mathbf{Y}^\star - \tilde{\mathbf{X}}_j)$
    
    $\mathcal{V}(\mathbf{X}_i) \gets R_i^\top\,\tilde{\mathcal{V}}(\tilde{\mathbf{X}}_i)$
    
    \BlankLine
    $\mathcal{L}(\theta) \gets \frac{1}{K} \sum_{i=1}^K \big\| \mathbf{X}_i - \mathrm{sg}\big(\mathbf{X}_i + \mathcal{V}(\mathbf{X}_i)\big) \big\|^2$
    
    $\theta \gets \theta - \eta \nabla_\theta \mathcal{L}(\theta)$
}
\end{algorithm}

\textbf{Physical Interpretation.} The drifting field should be viewed as a data-driven geometric transport mechanism rather than an approximation to the physical force field on the PES. Each displacement $(\mathbf{Y}^\star-\tilde{\mathbf{X}}_j)$ points from a generated swarm member toward the reference pathway, while the joint affinities determine how strongly these displacements contribute to the update. The attractive term $A_i^+$ promotes agreement with the reference-pathway distribution, whereas the repulsive coupling $A^-_{ij}$ coordinates the swarm and discourages collapse onto poorly distributed configurations. In this sense, the field provides a geometric analogue of the attractive and regularizing effects present in iterative pathway optimization, while remaining learned entirely from pathway geometry rather than energies or forces.

\section{Results}
\label{sec:results}

We evaluate Drift-React on two complementary tasks: full reaction pathway generation, which is the primary contribution of our method, and point-wise transition state (TS) prediction, which we include for compatibility with existing single-structure generative baselines. Both tasks draw on Transition1x~\cite{schreinerTransition1xDatasetBuilding2022}, a dataset of $\sim$10k organic reaction pathways computed at the $\omega$B97x/6-31G(d) level of theory. The full-pathway evaluation additionally uses the recently released Halo8~\cite{leeDatasetChemicalReaction2025}, which extends Transition1x to halogenated chemistry: $\sim$19k reaction pathways at the $\omega$B97X-3c level, combining recalculated Transition1x reactions with new molecules containing fluorine, chlorine, and bromine.

\subsection{Continuous Reaction Pathway Generation}

We assess Drift-React's ability to predict the complete reaction pathway in a single forward pass given only the reactant ($\mathbf{X}_R$) and product ($\mathbf{X}_P$) geometries. Transition1x provides 10 images per reaction, Halo8 provides 8. For Transition1x we follow the test split convention of recent generative TS methods~\cite{duanOptimalTransportGenerating2025,duanAccurateTransitionState2023}, evaluating on the full pathway rather than the reactant-TS-product triplet alone. Since Halo8 lacks a canonical split, we construct one using a reactant-grouped 80/10/10 partition that places all reactions sharing a reactant molecule into the same partition. This prevents trivial leakage between train and test: reactions differing only in the leaving group, the attacking nucleophile, or minor halogen substitutions would otherwise appear on both sides under a naive random split. The exact grouping procedure is detailed in \Cref{app:splits}.

We benchmark against linear interpolation (a non-physical geometric baseline corresponding to our prior $p_{\mathrm{prior}}$), the Image-Dependent Pair Potential (IDPP) method~\cite{smidstrupImprovedInitialGuess2014}, geodesic interpolation~\cite{zhuGeodesicInterpolationReaction2019}, and NeuralNEB~\cite{schreinerNeuralNEBNeuralNetworks2022}, a learned surrogate-force-field NEB method. Geometric fidelity is measured via TS-RMSD, IRC-RMSD, and discrete Fréchet distance over the full pathway. Energetic fidelity is measured via the absolute activation barrier error $|\Delta E_{\mathrm{TS}}^\ddagger|$ and the per-image energy MAE $E_{\mathrm{image}}$. All energy metrics use single-point DFT on the generated geometries, evaluated at each dataset's reference level of theory. Further implementation details are in \Cref{app:baselines}.

\begin{figure}[p]
    \centering
    \makeatletter\def\@captype{table}\makeatother
\centering
\caption{%
    \textbf{Quantitative comparison of continuous pathway generation methods.} 
    Metrics include Transition State Root Mean Square Deviation (TS-RMSD), Intrinsic Reaction Coordinate RMSD (IRC-RMSD), Fréchet Distance, absolute activation barrier error ($|\Delta E^\ddagger|$), Per-image Energy Mean Absolute Error ($E_{\text{image}}$ MAE), and total generation time per pathway. Lower is better ($\downarrow$) for all metrics.
}
\label{tab:benchmark_results}
\resizebox{\textwidth}{!}{%
\begin{tabular}{llcccccc}
\toprule
\textbf{Dataset} & \textbf{Method} & \textbf{TS-RMSD (\AA) $\downarrow$} & \textbf{IRC-RMSD (\AA) $\downarrow$} & \textbf{Fr\'echet Dist. $\downarrow$} & \textbf{$\abs{\Delta E_{\text{TS}}^\ddagger}$ (eV) $\downarrow$} & \textbf{$E_{\text{image}}$ MAE (eV) $\downarrow$} & \textbf{Time (ms) $\downarrow$} \\ 
\midrule
\multirow{5}{*}{\textbf{Transition1x}}
& Linear Interpolation & 0.432 $\pm$ 0.196 & 0.364 $\pm$ 0.115 & 0.748 $\pm$ 0.230 & 1.881 $\pm$ 1.203 & 0.917 $\pm$ 0.279 & 0.2 $\pm$ 0.1 \\
& IDPP & 0.432 $\pm$ 0.196 & 0.364 $\pm$ 0.115 & 0.748 $\pm$ 0.230 & 1.896 $\pm$ 1.205 & \textbf{0.915 $\pm$ 0.278} & 8.4 $\pm$ 6.6 \\
& Geodesic Interpolation & 0.432 $\pm$ 0.193 & 0.364 $\pm$ 0.111 & 0.747 $\pm$ 0.230 & 1.910 $\pm$ 1.210 & 0.916 $\pm$ 0.278 & 338.2 $\pm$ 236.5 \\
& NeuralNEB (300 steps, CI) & 0.581 $\pm$ 0.367 & 0.949 $\pm$ 0.625 & 2.775 $\pm$ 2.738 & 5.090 $\pm$ 30.403 & 12.281 $\pm$ 28.298 & 8902.9 $\pm$ 1979.6 \\
& \textbf{Drift-React (Ours)} & \textbf{0.151 $\pm$ 0.093} & \textbf{0.223 $\pm$ 0.123} & \textbf{0.460 $\pm$ 0.258} & \textbf{1.631 $\pm$ 1.380} & 3.214 $\pm$ 2.069 & 20.5 $\pm$ 36.5 \\
\midrule
\multirow{5}{*}{\textbf{Halo8}}
& Linear Interpolation & 0.334 $\pm$ 0.190 & 0.205 $\pm$ 0.108 & 0.349 $\pm$ 0.201 & 8.826 $\pm$ 51.490 & 3.937 $\pm$ 16.814 & 0.1 $\pm$ 0.0 \\
& IDPP & 0.337 $\pm$ 0.184 & 0.207 $\pm$ 0.108 & 0.351 $\pm$ 0.196 & 1.651 $\pm$ 1.246 & 1.306 $\pm$ 0.701 & 26.5 $\pm$ 6.8 \\
& Geodesic Interpolation & 0.316 $\pm$ 0.206 & 0.205 $\pm$ 0.116 & 0.338 $\pm$ 0.202 & 1.247 $\pm$ 1.014 & 1.055 $\pm$ 0.558 & 289.0 $\pm$ 521.8 \\
& NeuralNEB (200 steps, no CI) & 0.538 $\pm$ 0.295 & 0.472 $\pm$ 0.316 & 0.786 $\pm$ 0.583 & 4.663 $\pm$ 18.212 & 4.262 $\pm$ 16.126 & 7560.9 $\pm$ 1871.0 \\
& \textbf{Drift-React (Ours)} & \textbf{0.107 $\pm$ 0.086} & \textbf{0.123 $\pm$ 0.109} & \textbf{0.212 $\pm$ 0.189} & \textbf{0.914 $\pm$ 2.980} & \textbf{0.621 $\pm$ 1.145} & 12.7 $\pm$ 15.6 \\
\bottomrule
\end{tabular}%
}
\makeatletter\def\@captype{figure}\makeatother
    {
    \centering{
    \includegraphics[width=\textwidth]{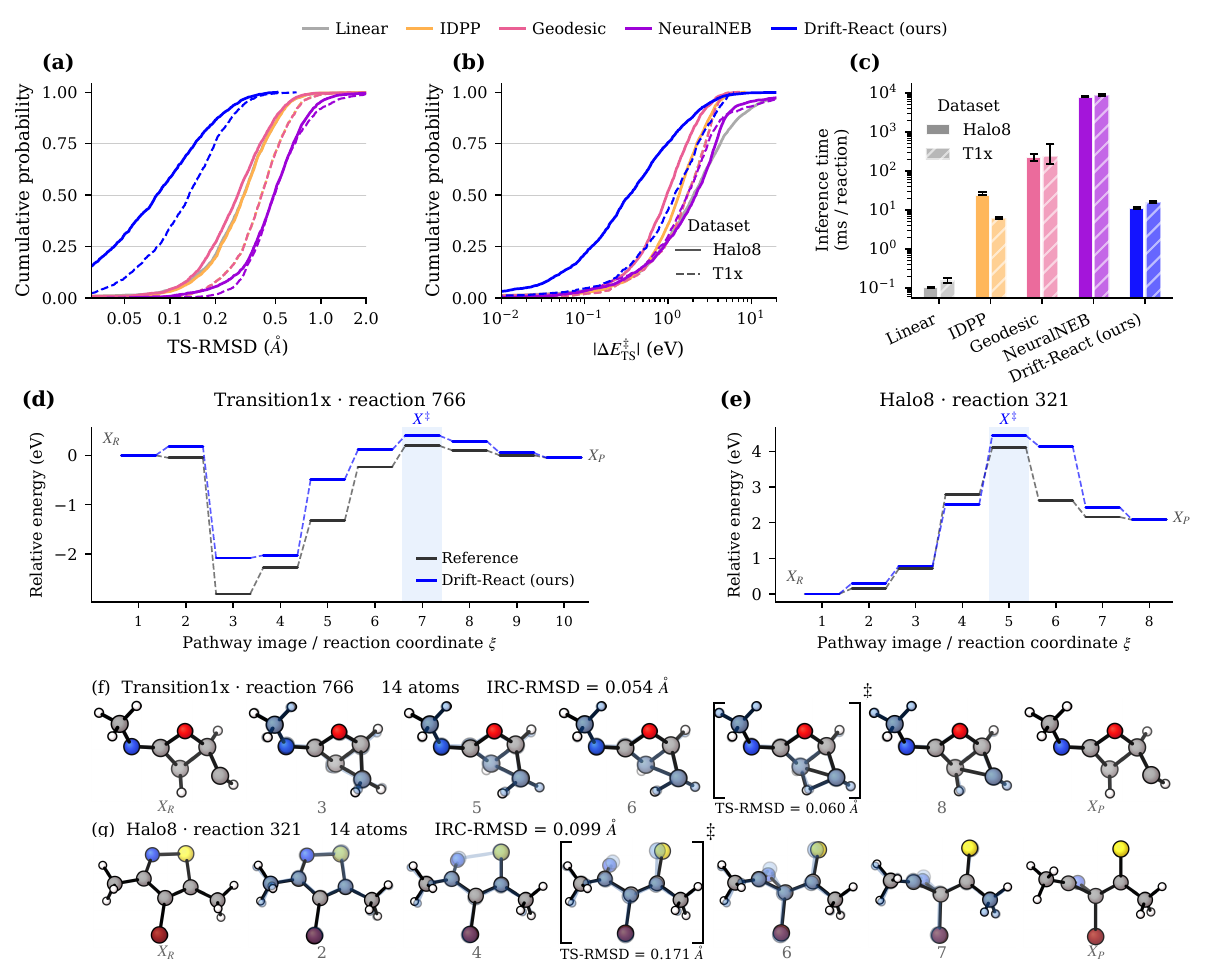}
    }
    \caption{%
    \textbf{Accuracy and inference cost of Drift-React and baseline pathway-generation methods.} Cumulative distributions over test reactions on Halo8 (solid lines) and Transition1x (dashed lines):
\textbf{(a)} TS-RMSD. \textbf{(b)} Absolute activation barrier error $|\Delta E^\ddagger_{\mathrm{TS}}|$. \textbf{(c)} Median inference time per reaction (log scale), with error bars spanning the interquartile range.
\textbf{(d, e)} Illustrative Transition1x and Halo8 energy profiles, comparing Drift-React predictions (blue) with reference energies (black), each relative to its reactant energy. \textbf{(f, g)} Corresponding pathway images, with reference structures in element colors and predictions overlaid in translucent blue.
}
}
    \label{fig:baseline_comparison}
\end{figure}

\Cref{tab:benchmark_results} summarizes the quantitative comparison. Drift-React achieves the best geometric fidelity across both datasets and all three geometric metrics. On Halo8 the TS-RMSD of 0.107 \AA~is more than three times lower than the best classical baseline (geodesic interpolation, 0.316 \AA) and five times lower than the retrained NeuralNEB (0.538 \AA). On Transition1x, Drift-React achieves a TS-RMSD of 0.151 \AA, again surpassing all baselines. The IRC-RMSD and Fréchet distance show comparable margins, indicating that the geometric improvement extends along the full pathway rather than concentrating at the saddle point alone.

The energetic comparison is more nuanced. Heuristic interpolation baselines achieve deceptively competitive geometric metrics on Halo8 because coordinate-space interpolation bounds the maximum geometric deviation, but the resulting paths frequently traverse regions of steric clash where atoms pass directly through one another. This is reflected in extreme energetic outliers visible in the standard deviations: linear interpolation has a barrier MAE of $8.826 \pm 51.5$~eV on Halo8, with the long tail driven by reactions where the linear path produces unphysical atomic clashes. NeuralNEB, in our retrained setup, relaxes IDPP-initialized pathways under learned forces using a maximum of 200 optimization steps without climbing image on Halo8 and 300 steps with climbing image on Transition1x. These configurations yield barrier MAEs of 4.663~eV and 5.090~eV, respectively, but retain high-variance pathway errors and substantially higher computational costs than Drift-React. NEB-style methods are well known to be sensitive to spring-constant tuning, climbing-image scheduling, and convergence criteria~\cite{asgeirssonNudgedElasticBand2021,lindgrenScaledDynamicOptimizations2019}. Trained to generate physically consistent pathways directly, Drift-React achieves the best barrier error on Halo8 ($0.914 \pm 2.98$~eV, compared to $1.247 \pm 1.01$~eV for geodesic interpolation) and the best per-image energy MAE ($0.621 \pm 1.15$~eV).

On Transition1x, the energetic story is similar on the barrier metric ($1.631$~eV vs.\ $1.881$~eV for linear, $5.090$~eV for NeuralNEB), but Drift-React's per-image energy MAE ($3.214$~eV) is higher than the classical interpolation baselines ($\sim 0.92$~eV). This reflects a structural property of the Transition1x split, which evaluates against CI-NEB-relaxed pathways that contain only the reactant–TS–product triplet plus interpolated intermediate images. The barrier is correctly recovered, but absolute image energies are sensitive to small geometric perturbations along the rest of the curve. We expand on this and discuss the inherent sensitivity of single-point DFT energies to sub-Ångström coordinate variations near the TS in \Cref{app:ts_energetics}.

In terms of efficiency, Drift-React generates a complete pathway in $12$--$20$~ms per reaction, 400-500 times faster than NeuralNEB ($\sim 7{,}000$--$8{,}000$~ms) and an order of magnitude faster than geodesic interpolation ($\sim 290$--$340$~ms). This speedup is achieved without iterative optimization or force evaluation: the trained generator amortizes the computational cost of pathway construction into a single forward pass.

We note that the concurrent MEPIN model~\cite{namTransferableLearningReaction2025} addresses a related but distinct problem formulation: it operates at the GFN1-xTB level of theory rather than DFT, evaluates against a different test split, and trains via energy-based objectives requiring PES queries. A direct head-to-head comparison would require retraining MEPIN at the DFT level on our splits. We leave this to future work and view the two methods as complementary explorations of the learned-pathway design space.

\subsection{Point-wise Transition State Generation}

To benchmark Drift-React against existing single-structure TS generators~\cite{kimDiffusionbasedGenerativeAI2024,duanAccurateTransitionState2023,duanOptimalTransportGenerating2025}, we evaluate it on the reactant-TS-product subset of Transition1x using the standardized data split released with ReactOT~\cite{duanOptimalTransportGenerating2025}. We restrict the pathway to $F = 3$ images corresponding to reactant, TS, and product, leaving the architecture, drifting field, and training procedure unchanged from the full-pathway model: the only change is the pathway resolution. Energy calculations use single-point DFT at the $\omega$B97x/6-31G(d) level, matching Transition1x's reference level of theory.

\begin{figure}[h]
    \centering{
    \includegraphics[width=\textwidth]{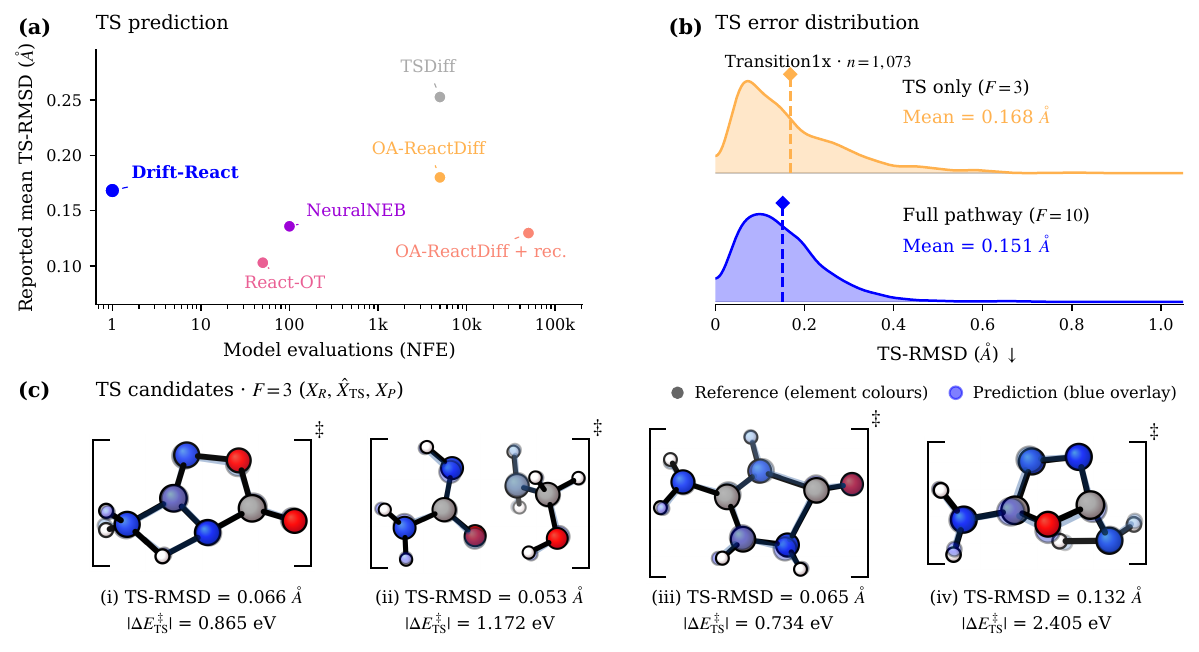}
    }
    \caption{
    \textbf{Drift-React generates transition-state candidates in a single network evaluation.}
    \textbf{(a)} Reported mean TS-RMSD versus the number of network function evaluations (NFE, log scale) for Transition1x TS prediction. Drift-React requires one evaluation, baseline values are taken from prior work.
    \textbf{(b)} TS-RMSD distributions for TS-only ($F=3$, orange) and full-pathway ($F=10$, blue) Drift-React models on 1,073 Transition1x test reactions. Both models are evaluated against the same reference TS, with the $F=10$ frame selected using the reference TS index.
    \textbf{(c)} Selected $F=3$ TS candidates, with reference structures in element colors and predictions overlaid in translucent blue.
    }
    \label{fig:ts-gen}
\end{figure}

\Cref{fig:ts-gen} summarizes the comparison. Drift-React achieves a mean TS-RMSD of 0.168~\AA, comparable to OA-ReactDiff (0.18~\AA) and behind ReactOT (0.103~\AA), at substantially lower inference cost. A single forward pass suffices to generate the TS, compared to 50 ODE-integration steps for ReactOT, $5{,}000$ denoising steps for TSDiff and OA-ReactDiff, and up to $50{,}000$ evaluations for OA-ReactDiff with recommender-based selection over 40 samples. The resulting inference time is roughly $10\times$ lower than ReactOT and two to three orders of magnitude lower than the diffusion-based methods.

\Cref{fig:ts-gen}a shows the trade-off between reported TS-RMSD and the number of network function evaluations (NFE). Drift-React and ReactOT define the accuracy-NFE Pareto front among the compared methods: Drift-React requires a single evaluation, whereas ReactOT achieves the lowest reported TS-RMSD with 50 evaluations~\cite{duanOptimalTransportGenerating2025}. NFE measures the number of model evaluations required at inference and is therefore related to computational cost, although wall-clock time also depends on model architecture, hardware, and implementation. \Cref{fig:ts-gen}b further compares the TS-error distributions of the TS-only ($F=3$) and full-pathway ($F=10$) models on the same 1,073 test reactions: the full-pathway model achieves a lower mean TS-RMSD (0.151~\AA{} versus 0.168~\AA{}) and a reduced upper error tail, despite substantial overlap between the distributions. Selected predictions in \Cref{fig:ts-gen}c illustrate close agreement with the reference TS geometries, with sub-ångström RMSD across the examples shown.

This trade-off reflects the different training objectives: ReactOT and OA-ReactDiff specialize in TS geometry by training and evaluating exclusively on single saddle-point structures, whereas Drift-React optimizes geometric continuity along the full reaction coordinate, with the TS emerging as the highest-energy image of the predicted pathway. The same architecture and training procedure handle both regimes, only the pathway resolution changes. The same architecture and training procedure that produces full pathways at $F = 10$ produces point-wise TS structures at $F = 3$. We adopt TS-RMSD as the primary axis of comparison because it is the metric on which the speed--accuracy trade-off is cleanest: saddle-point energies on raw single-point DFT are heavily sensitive to sub-ångström coordinate variations along the unstable mode of the reaction coordinate. Median energy statistics on raw geometries, following the convention of~\citet{duanOptimalTransportGenerating2025}, are reported in \Cref{app:ts_energetics}. We note that Drift-React's $F = 3$, TS-only energy performance is degraded relative to dedicated single-structure generators (\Cref{app:ts_energetics}). We view this as the cost of the unified architecture and recommend running Drift-React in full-pathway mode (where the activation-barrier MAE is 1.631~eV) for applications requiring accurate point-wise TS energetics. This accuracy–efficiency trade-off is suited to high-throughput reaction-network exploration, where the same architectural framework supports both pathway and TS prediction in a single forward pass.

\section{Conclusion}

We introduced Drift-React, a generative framework that predicts complete chemical reaction pathways in a single forward pass, eliminating both the iterative force-evaluation loops of classical chain-of-states methods such as NEB and the sequential SDE/ODE integration that constrains modern diffusion and flow-based models. By amortizing the evolution of the pushforward distribution into training via an $\mathrm{SE}(3)$-equivariant drifting field with Sinkhorn-weighted swarm repulsion, Drift-React removes the central inference-time bottleneck that has limited learned pathway models in high-throughput settings.

On full reaction pathway generation, Drift-React achieves the best geometric fidelity across both Transition1x and Halo8 and all three geometric metrics (TS-RMSD, IRC-RMSD, Fréchet distance), surpassing classical interpolation baselines and the iterative-ML NeuralNEB approach. Energetic accuracy is dataset-dependent: on Halo8, Drift-React achieves the best activation barrier error and per-image energy MAE. On Transition1x, the activation barrier is correctly recovered while per-image energies are sensitive to the structure of the reference data. These results are obtained starting only from a trivial linear interpolation between reactant and product, with no privileged geometric prior at inference.

As a special case, restricting the pathway to three images (reactant, transition state, product) recovers the standard point-wise TS prediction setting addressed by recent diffusion and flow-matching models. The same training procedure transfers without modification: only the pathway resolution changes. In this regime, Drift-React jointly defines the speed-accuracy Pareto front with ReactOT, occupying the speed-optimal corner with TS-RMSD competitive with diffusion and flow-based generators at a single network evaluation per reaction, an order of magnitude faster than ReactOT and two to three orders faster than diffusion-based methods.

Together, these results show that one-step amortized generation can produce geometrically accurate pathway and TS candidates from DFT-level reference data, while downstream quantum refinement confirms that these candidates frequently lie within the convergence basin of the corresponding first-order saddle point. We see Drift-React as a foundation for ultra-fast exploration of complex reaction networks, and more broadly as evidence that the Drifting Models paradigm of~\citet{dengGenerativeModelingDrifting2026} extends naturally to physical scientific domains where iterative inference is the binding constraint on practical deployment.

\section*{Data and Code Availability}
The official implementation of Drift-React is available on a public GitHub repository under the MIT license at \url{https://github.com/schwallergroup/drift-react}.

\section*{Acknowledgments}
This publication was created as part of NCCR Catalysis (grant number 225147), a National Centre of Competence in Research funded by the Swiss National Science Foundation.

\bibliographystyle{unsrtnat}
\bibliography{references}


\newpage
\appendix

\crefalias{section}{appendix}
\crefalias{subsection}{appendix}
\crefalias{subsubsection}{appendix}
\renewcommand{\figurename}{Supplementary Figure}
\renewcommand{\tablename}{Supplementary Table}

\setcounter{figure}{0}
\setcounter{table}{0}
\renewcommand{\thefigure}{S\arabic{figure}}
\renewcommand{\thetable}{S\arabic{table}}

\input{appendix}

\newpage
\input{checklist.tex}

\end{document}

%% file: figures/tikz/drifting.tex
\begin{tikzpicture}[
    >=Stealth,
    line width=1.0pt,
    dot/.style       ={circle, fill, draw=white, line width=0.5pt, inner sep=1.45pt},
    dataDot/.style   ={circle, draw=pathBlue, fill=white, line width=0.7pt, inner sep=1.25pt},
    anchorDot/.style ={circle, fill=neutralDark, draw=white, line width=0.7pt, inner sep=1.9pt},
    tsRing/.style    ={circle, draw=tsAccent, line width=1.1pt, fill=white, inner sep=1.8pt},
    panelSub/.style  ={font=\large, text=neutralDark},
    panelSub2/.style ={font=\normalsize, text=neutralMid},
    panelTag/.style  ={font=\large\bfseries, text=neutralDark, anchor=north west, inner sep=2pt},
    endpointLab/.style={font=\large, text=neutralDark, inner sep=1.5pt},
    annotLab/.style  ={font=\normalsize, text=neutralDark, inner sep=1.5pt},
    annotBlue/.style ={font=\normalsize, text=pathBlue, inner sep=1.5pt},
    driftArrow/.style={-{Stealth[length=3.2pt, width=2.8pt]}, line width=0.75pt, draw=neutralDark},
    frame/.style     ={draw=neutralLine, line width=0.5pt},
    pesStroke/.style ={line width=0.4pt},
]

\input{figures/tikz/drift_figdata.tex}

\definecolor{pathBlue}{HTML}{0000FF}
\definecolor{pathBlueSoft}{HTML}{9A9AF0}
\definecolor{neutralLine}{HTML}{9A9A9A}
\definecolor{neutralMid}{HTML}{6A6A6A}
\definecolor{neutralDark}{HTML}{2B2B2B}
\definecolor{tsAccent}{HTML}{FA8775}
\colorlet{pesLine}{neutralDark}          

\def\plotW{5.2}\def\plotH{3.8}\def\panelGap{0.7}
\pgfmathsetmacro{\panelStride}{\plotW+\panelGap}
\def\titleY{-0.3}
\def\molW{1.55cm}\def\tsW{1.45cm}

\newcommand{\configBackground}[1]{%
    \fill[white] (0,0) rectangle (\plotW,\plotH);
    \begin{scope}
        \clip (0,0) rectangle (\plotW,\plotH);
        \pesContours
    \end{scope}
    \draw[frame] (0,0) rectangle (\plotW,\plotH);
    \node[panelTag] at (0,\plotH) {(#1)};
}
\newcommand{\molThumb}[4][0.95cm]{
    \node[anchor=#4, inner sep=1.2pt, fill=white, fill opacity=0.92, text opacity=1,
          draw=neutralLine, line width=0.35pt, rounded corners=1.5pt] at (#3)
        {\includegraphics[width=#1]{figures/tikz/#2.pdf}};
}
\newcommand{\rpThumbs}{%
    \molThumb{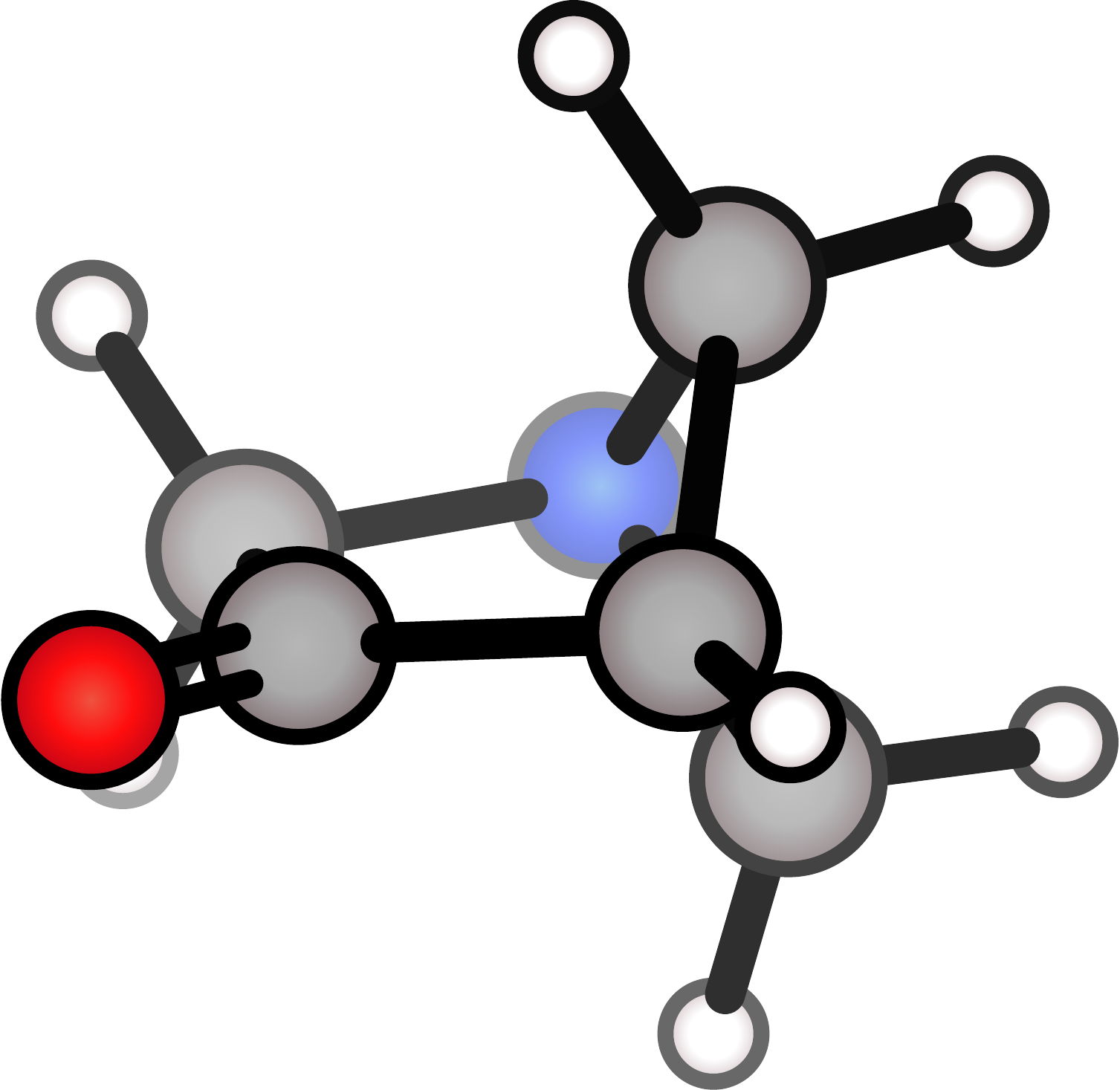}{1.55,0.05}{south west}
    \draw[neutralLine, line width=0.4pt] (\Rpt) -- (1.55,0.3);
    \molThumb{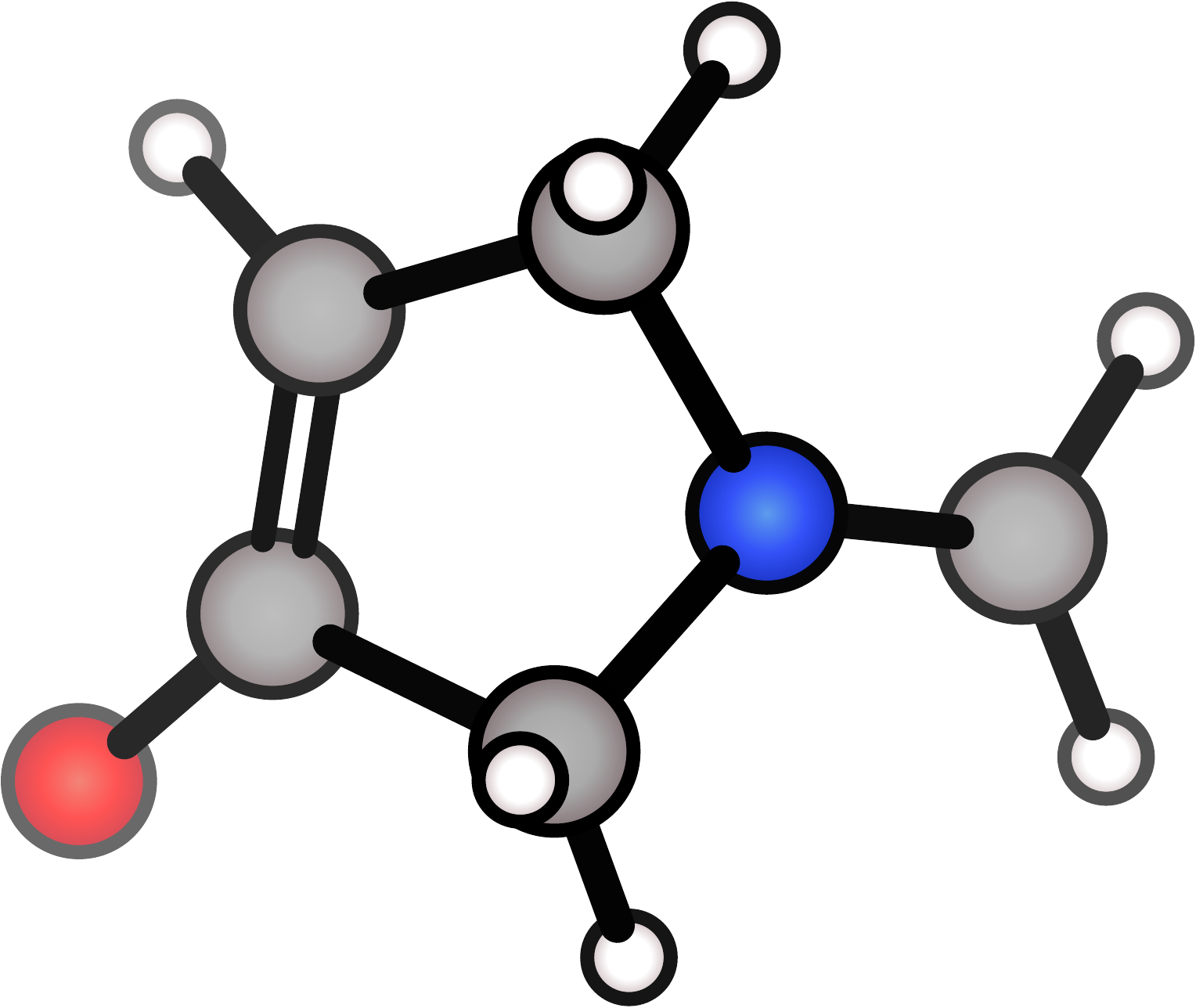}{\plotW-0.06,\plotH-1.05}{north east}
    \draw[neutralLine, line width=0.4pt] (\Ppt) -- (\plotW-0.5,\plotH-1.05);
}
\newcommand{\anchors}{%
    \rpThumbs
    \node[anchorDot] at (\Rpt) {};
    \node[anchorDot] at (\Ppt) {};
    \node[endpointLab, anchor=north east] at (\Rpt) {$\mathbf{X}_R$};
    \node[endpointLab, anchor=south] at ([yshift=2pt]\Ppt) {$\mathbf{X}_P$};
}
\newcommand{\trainStage}[4]{%
    \draw[pathBlue!#4!white, line width=1.4pt] plot[smooth] coordinates {#1};
    \foreach \x/\y in #3 { \node[dot, fill=pathBlue!#4!white, inner sep=1.1pt, line width=0.4pt] at (\x,\y) {}; }
}
\newcommand{\stageArrows}[1]{%
    \foreach \x/\y/\u/\v in #1 { \draw[driftArrow, shorten <=1.8pt] (\x,\y) -- (\u,\v); }
}
\newcommand{\panelTitle}[2]{%
    \node[panelSub,  anchor=north] (pt) at (\plotW/2,\titleY) {#1};
    \node[panelSub2, anchor=north, align=center] at (pt.south) {#2};
}
\newcommand{\priorSample}[2]{
    \draw[neutralLine, dashed, line width=0.8pt] (\Rpt) -- (\Ppt);
    \foreach \x/\y in \priorMeans {
        \fill[neutralLine!15] (\x,\y) circle (\sigmaCm);
        \draw[neutralLine!55, line width=0.3pt] (\x,\y) circle (\sigmaCm);
        \fill[neutralLine] (\x,\y) circle (0.55pt);
    }
    \draw[neutralMid, line width=1.1pt] plot[smooth] coordinates {#2};
    \foreach \x/\y in #1 { \node[dot, fill=neutralMid] at (\x,\y) {}; }
}
\def\trainW{\panelStride*3-\panelGap}      
\def\infGap{1.6}                            
\pgfmathsetmacro{\xE}{3*\panelStride-\panelGap+\infGap}

    \fill[neutralLine!8, rounded corners=6pt] (-0.3,-2.0) rectangle ({\trainW+0.3},\plotH+1.35);
    \fill[pathBlue!4, rounded corners=6pt] ({\xE-0.3},-2.0) rectangle ({\xE+\plotW+0.3},\plotH+1.35);
\node[font=\Large\bfseries, text=neutralDark, anchor=north west] (hT) at (-0.12,\plotH+1.27) {Training};
\node[font=\normalsize, text=neutralMid, anchor=base west] at ([xshift=6pt]hT.base east)
    {iterate $s=1,\dots,S$ \ until $\mathcal{V}\to\mathbf{0}$};
\node[font=\Large\bfseries, text=pathBlue, anchor=north west] (hI) at ({\xE-0.12},\plotH+1.27) {Inference};
\node[font=\normalsize, text=neutralMid, anchor=base west] at ([xshift=6pt]hI.base east) {NFE $=1$};

\begin{scope}[shift={(0,0)}]
    \configBackground{a}
    \priorSample{\priorImgs}{\priorPath}
    \anchors
    \draw[neutralMid, line width=0.4pt] (2.513,1.860) -- ++(-0.5,0.55)
        node[anchor=south, font=\small, text=neutralMid, inner sep=1pt] {$\pm\sigma$};
    \begin{scope}[shift={(\plotW-0.95,0.28)}]
        \draw[-{Stealth[length=3pt]}, line width=0.7pt, neutralDark] (0,0) -- (0.55,0)
            node[right, annotLab, inner sep=1pt] {$\xi_1$};
        \draw[-{Stealth[length=3pt]}, line width=0.7pt, neutralDark] (0,0) -- (0,0.55)
            node[above, annotLab, inner sep=1pt, yshift=-2pt] {$\xi_2$};
    \end{scope}
    \panelTitle{Prior \ $\mathbf{Z}\sim p_{\mathrm{prior}}$}
      {\small $\mathbf{Z}=(1-t)\,\mathbf{X}_R+t\,\mathbf{X}_P+\boldsymbol{\epsilon}$\\[-1pt]
       \small $\boldsymbol{\epsilon}\sim\mathcal{N}(\mathbf{0},\sigma^{2}\mathbf{I})$}
\end{scope}

\draw[-{Stealth[length=5pt]}, line width=1.2pt, neutralDark]
    (\plotW+0.08,\plotH/2) -- (\panelStride-0.08,\plotH/2)
    node[midway, above, font=\normalsize] {$f_{\theta}$};

\begin{scope}[shift={(1*\panelStride,0)}]
    \configBackground{b}
    \draw[pathBlueSoft, dashed, line width=0.9pt] plot[smooth] coordinates {\mepPath};
    \foreach \x/\y in \dataImgs { \node[dataDot] at (\x,\y) {}; }
    \draw[neutralLine!75, line width=0.9pt] plot[smooth] coordinates {\pathBtwo};
    \draw[neutralLine!75, line width=0.9pt] plot[smooth] coordinates {\pathBthree};
    \foreach \x/\y in \imgsBtwo   { \node[dot, fill=neutralLine!75, inner sep=1.05pt] at (\x,\y) {}; }
    \foreach \x/\y in \imgsBthree { \node[dot, fill=neutralLine!75, inner sep=1.05pt] at (\x,\y) {}; }
    \draw[pathBlue!45!white, line width=1.6pt] plot[smooth] coordinates {\pathB};
    \foreach \x/\y in \imgsB { \node[dot, fill=pathBlue!45!white] at (\x,\y) {}; }
    \foreach \x/\y/\u/\v in \arrowsB { \draw[driftArrow, shorten <=1.8pt] (\x,\y) -- (\u,\v); }
    \node[font=\small, text=pathBlue, anchor=south east, inner sep=1pt] at (1.95,2.35) {$\mathbf{Y}^{\star}$};
    \anchors
    \node[font=\scriptsize, anchor=south east, align=left, fill=white, fill opacity=0.9,
          text opacity=1, draw=neutralLine, line width=0.3pt, rounded corners=1pt,
          inner sep=1.8pt] at (\plotW-0.05,0.05) {%
        \textcolor{pathBlue!45!white}{\rule[0.5ex]{7pt}{1.2pt}}\,$\mathbf{X}_i$\\[-0.6pt]
        \textcolor{neutralLine}{\rule[0.5ex]{7pt}{0.9pt}}\,$\mathbf{X}_j\ (j\neq i)$\\[-0.6pt]
        \textcolor{pathBlue!60!white}{\rule[0.5ex]{1.4pt}{0.6pt}\,\rule[0.5ex]{1.4pt}{0.6pt}\,\rule[0.5ex]{1.4pt}{0.6pt}}\,$\mathbf{Y}^{\star}$ (MEP)\\[-0.6pt]
        \textcolor{neutralDark}{\rule[0.5ex]{7pt}{0.8pt}}\,$\mathcal{V}(\mathbf{X}_i)$, scaled};
    \panelTitle{Drifting field at step $s$}
      {\small $\mathcal{V}(\mathbf{X}_i)=A_i^{+}\sum_{j=1}^{K}A_{ij}^{-}\,(\mathbf{Y}^{\star}-\mathbf{X}_j)$}
\end{scope}

\draw[-{Stealth[length=5pt, width=4pt]}, line width=1.0pt, neutralDark, rounded corners=7pt]
    ({1*\panelStride+\plotW-0.9},\plotH+0.06) -- ++(0,0.42)
    -| ({1*\panelStride+0.9},\plotH+0.06);
\node[font=\small, text=neutralDark, fill=neutralLine!8, inner xsep=4pt, inner ysep=1.5pt]
    at ({1*\panelStride+\plotW/2},\plotH+0.48) {update $\theta$, \ $s\leftarrow s+1$};

\begin{scope}[shift={(2*\panelStride,0)}]
    \configBackground{c}
    \draw[pathBlue, line width=1.8pt] plot[smooth] coordinates {\mepPath};
    \foreach \x/\y in \mepImgs {
        \node[circle, draw=pathBlue!55, line width=0.4pt, inner sep=2.0pt] at (\x,\y) {};
        \node[dot, fill=pathBlue, inner sep=1.1pt, line width=0.4pt] at (\x,\y) {};
    }
    \node[tsRing] at (\TSpt) {};
    \anchors
    \node[font=\small, text=neutralDark, anchor=north, fill=white, fill opacity=0.85,
          text opacity=1, inner sep=2pt] at (2.95,\plotH-0.06)
        {$\mathcal{V}(\mathbf{X}_i)=\mathbf{0}\ \ \forall i$};
    \begin{scope}[shift={(2.84,0.42)}]
        \fill[white, fill opacity=0.92, rounded corners=1.5pt] (-0.18,-0.34) rectangle (2.32,1.42);
        \draw[neutralLine, line width=0.3pt, rounded corners=1.5pt] (-0.18,-0.34) rectangle (2.32,1.42);
        \draw[-{Stealth[length=2.4pt]}, line width=0.5pt, neutralDark] (0,0) -- (2.2,0);
        \draw[-{Stealth[length=2.4pt]}, line width=0.5pt, neutralDark] (0,0) -- (0,1.12);
        \node[font=\scriptsize, anchor=north east, inner sep=1pt] at (2.2,-0.02) {iteration $s$};
        \node[font=\scriptsize, anchor=west, inner sep=1pt] at (0.04,1.12) {$\overline{\|\mathcal{V}\|}$};
        \draw[neutralDark, line width=0.9pt] plot[smooth] coordinates {\vCurve};
        \node[circle, fill=pathBlue, draw=white, line width=0.3pt, inner sep=1.3pt,
              label={[font=\scriptsize, inner sep=1pt]above:$\mathbf{0}$}] at (\vMarkF) {};
    \end{scope}
    \panelTitle{Converged \ $q_{\theta}=p_{\mathrm{data}}$}
               {\small fixed point of training: \ $\mathcal{V}\equiv\mathbf{0}$}
\end{scope}

\draw[-{Stealth[length=6pt, width=5pt]}, line width=1.8pt, pathBlue]
    ({\trainW+0.12},\plotH/2) -- ({\xE-0.12},\plotH/2)
    node[midway, above, font=\normalsize] {$f_{\theta^\star}$};

\begin{scope}[shift={(\xE,0)}]
    \configBackground{d}
    \draw[neutralLine, dashed, line width=0.6pt] (\Rpt) -- (\Ppt);
    \draw[neutralMid!80, line width=1.0pt] plot[smooth] coordinates {\infPath};
    \draw[pathBlue, line width=1.8pt] plot[smooth] coordinates {\mepPath};
    \foreach \x/\y/\u/\v in \oneStep {
        \draw[-{Stealth[length=2.8pt]}, line width=0.5pt, neutralMid, shorten <=1.5pt, shorten >=2.5pt] (\x,\y) -- (\u,\v);
        \node[dot, fill=neutralMid, inner sep=1.1pt] at (\x,\y) {};
    }
    \foreach \x/\y in \mepImgs { \node[dot, fill=pathBlue] at (\x,\y) {}; }
    \node[tsRing] (tse) at (\TSpt) {};
    \node[anchor=north west, inner sep=1.2pt, fill=white, fill opacity=0.92, text opacity=1]
        (tsThumb) at (0.75,\plotH-0.1) {\includegraphics[width=0.95cm]{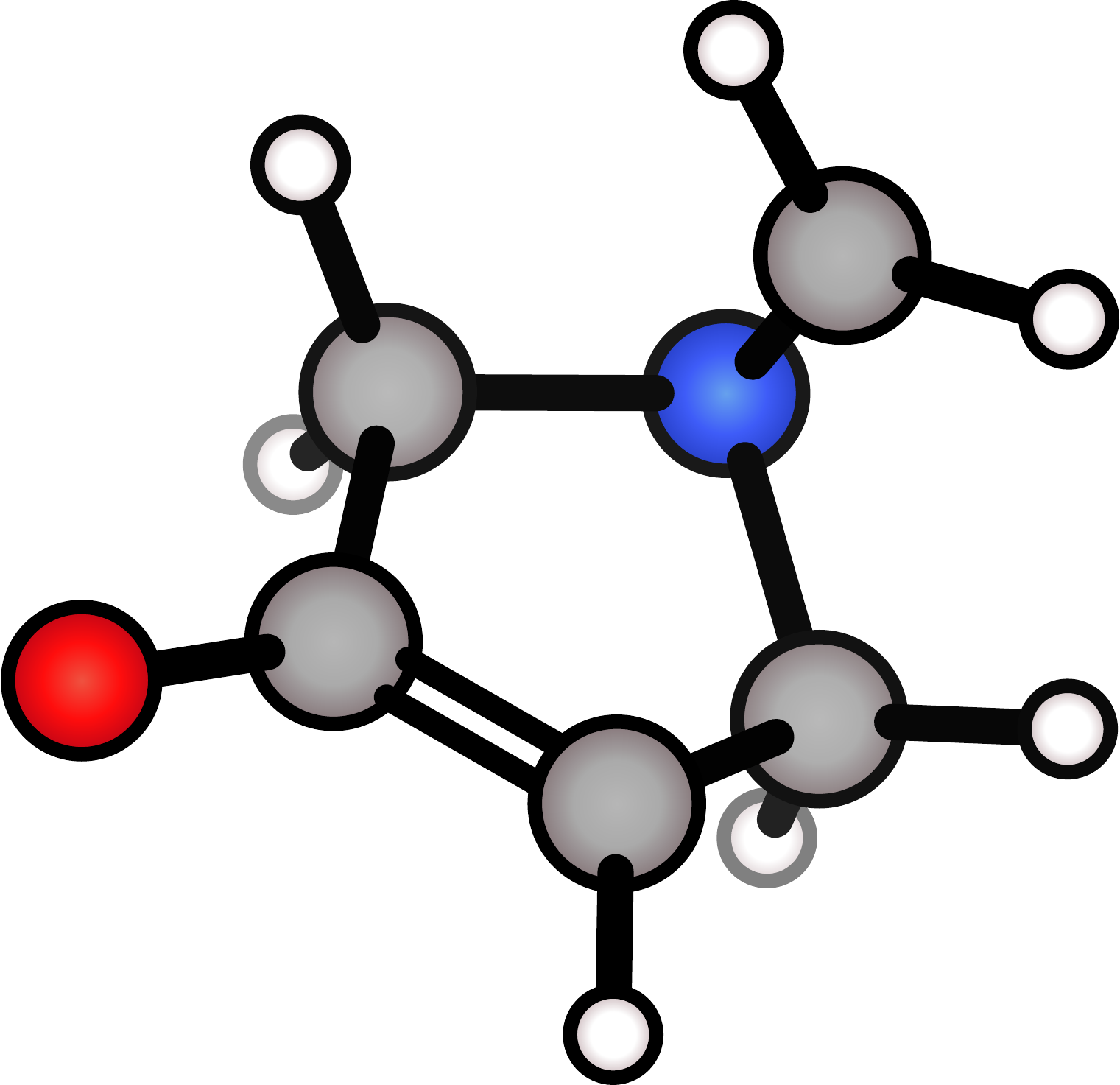}};
    \draw[black, line width=0.7pt]
        ([xshift=2pt]tsThumb.south west) -| (tsThumb.north west) -- ++(2pt,0);
    \draw[black, line width=0.7pt]
        ([xshift=-2pt]tsThumb.south east) -| (tsThumb.north east) -- ++(-2pt,0);
    \node[font=\small\bfseries, anchor=north west, inner sep=0.5pt] at ([xshift=1pt]tsThumb.north east) {$\ddagger$};
    \draw[tsAccent, line width=0.5pt] (tsThumb.south east) -- (tse);
    \anchors
    \panelTitle{$\hat{\mathbf{Y}}=f_{\theta^\star}(\mathbf{Z})$ \ in one step}
               {\small fresh $\boldsymbol{\epsilon}\sim\mathcal{N}(\mathbf{0},\sigma^{2}\mathbf{I})$ \ $\to$ \ MEP}
\end{scope}
\end{tikzpicture}

%% file: figures/tikz/drift_figdata.tex
\newcommand{\pesContours}{%
\draw[pesStroke, pesLine!55!white] plot[smooth] coordinates {(4.389,2.791) (4.333,2.813) (4.280,2.848) (4.235,2.894) (4.202,2.946) (4.179,3.017) (4.176,3.087) (4.193,3.155) (4.228,3.213) (4.271,3.256) (4.333,3.290) (4.404,3.306) (4.474,3.303) (4.539,3.284) (4.600,3.247) (4.646,3.199) (4.680,3.143) (4.701,3.073) (4.702,3.003) (4.684,2.934) (4.651,2.876) (4.600,2.828) (4.544,2.798) (4.474,2.783) (4.404,2.788) (4.389,2.791)};%
\draw[pesStroke, pesLine!52!white] plot[smooth] coordinates {(4.375,2.593) (4.305,2.615) (4.249,2.641) (4.193,2.676) (4.137,2.721) (4.095,2.765) (4.053,2.820) (4.020,2.876) (3.997,2.932) (3.978,3.003) (3.971,3.073) (3.976,3.143) (3.992,3.213) (4.016,3.270) (4.051,3.326) (4.095,3.377) (4.151,3.424) (4.207,3.458) (4.263,3.482) (4.333,3.501) (4.404,3.508) (4.474,3.505) (4.536,3.494) (4.600,3.474) (4.656,3.447) (4.712,3.410) (4.761,3.368) (4.810,3.312) (4.846,3.256) (4.873,3.199) (4.894,3.129) (4.904,3.059) (4.904,2.989) (4.895,2.925) (4.876,2.862) (4.851,2.806) (4.815,2.750) (4.768,2.698) (4.712,2.653) (4.656,2.621) (4.600,2.599) (4.530,2.585) (4.460,2.582) (4.389,2.590) (4.375,2.593)};%
\draw[pesStroke, pesLine!50!white] plot[smooth] coordinates {(4.404,2.440) (4.336,2.455) (4.277,2.475) (4.221,2.499) (4.165,2.531) (4.110,2.567) (4.058,2.609) (4.014,2.652) (3.968,2.703) (3.926,2.760) (3.891,2.820) (3.864,2.876) (3.842,2.939) (3.828,2.998) (3.820,3.059) (3.820,3.129) (3.828,3.198) (3.847,3.270) (3.870,3.326) (3.900,3.382) (3.940,3.438) (3.982,3.484) (4.039,3.532) (4.095,3.570) (4.151,3.600) (4.207,3.622) (4.277,3.642) (4.347,3.654) (4.418,3.657) (4.488,3.653) (4.558,3.641) (4.626,3.621) (4.684,3.597) (4.740,3.568) (4.796,3.530) (4.853,3.484) (4.895,3.440) (4.937,3.387) (4.976,3.326) (5.003,3.270) (5.024,3.213) (5.043,3.143) (5.053,3.073) (5.055,3.003) (5.049,2.935) (5.038,2.876) (5.021,2.817) (4.993,2.750) (4.962,2.694) (4.923,2.639) (4.881,2.592) (4.825,2.543) (4.768,2.505) (4.712,2.477) (4.652,2.455) (4.593,2.441) (4.530,2.433) (4.460,2.433) (4.404,2.440)};%
\draw[pesStroke, pesLine!48!white] plot[smooth] coordinates {(4.389,2.313) (4.326,2.328) (4.263,2.349) (4.207,2.373) (4.151,2.402) (4.095,2.436) (4.039,2.476) (3.982,2.522) (3.935,2.567) (3.884,2.622) (3.842,2.675) (3.800,2.736) (3.768,2.792) (3.741,2.848) (3.720,2.904) (3.702,2.967) (3.690,3.031) (3.684,3.101) (3.686,3.171) (3.696,3.241) (3.715,3.312) (3.736,3.368) (3.764,3.424) (3.800,3.480) (3.842,3.535) (3.884,3.579) (3.931,3.621) (3.982,3.659) (4.035,3.691) (4.091,3.719) (4.151,3.743) (4.208,3.761) (4.271,3.775) (4.333,3.784) (4.404,3.787) (4.474,3.785) (4.544,3.775) (4.608,3.761) (4.670,3.742) (4.726,3.719) (4.782,3.692) (4.839,3.658) (4.890,3.621) (4.937,3.581) (4.981,3.536) (5.021,3.489) (5.058,3.438) (5.091,3.383) (5.119,3.326) (5.142,3.270) (5.161,3.205) (5.175,3.143) (5.183,3.073) (5.186,3.003) (5.181,2.932) (5.170,2.862) (5.152,2.792) (5.127,2.722) (5.101,2.666) (5.068,2.609) (5.028,2.553) (4.979,2.498) (4.933,2.455) (4.881,2.415) (4.825,2.380) (4.768,2.352) (4.704,2.328) (4.642,2.313) (4.572,2.303) (4.502,2.301) (4.432,2.306) (4.389,2.313)};%
\draw[pesStroke, pesLine!45!white] plot[smooth] coordinates {(4.418,2.186) (4.347,2.201) (4.277,2.223) (4.221,2.245) (4.165,2.271) (4.111,2.300) (4.053,2.337) (3.996,2.376) (3.940,2.420) (3.886,2.469) (3.843,2.511) (3.800,2.557) (3.758,2.608) (3.716,2.665) (3.679,2.722) (3.648,2.778) (3.621,2.834) (3.599,2.890) (3.577,2.961) (3.563,3.031) (3.556,3.101) (3.556,3.171) (3.561,3.234) (3.573,3.298) (3.594,3.368) (3.616,3.424) (3.644,3.480) (3.678,3.536) (3.716,3.588) (3.756,3.635) (3.800,3.678) (3.848,3.719) (3.898,3.756) (3.950,3.789) (4.011,3.821) (4.066,3.845) (4.123,3.866) (4.193,3.886) (4.263,3.899) (4.333,3.907) (4.404,3.910) (4.474,3.907) (4.544,3.899) (4.614,3.885) (4.684,3.865) (4.754,3.839) (4.811,3.813) (4.867,3.782) (4.921,3.747) (4.977,3.705) (5.025,3.663) (5.077,3.610) (5.119,3.559) (5.156,3.508) (5.189,3.455) (5.220,3.396) (5.245,3.340) (5.270,3.270) (5.288,3.201) (5.301,3.129) (5.307,3.059) (5.309,2.989) (5.304,2.918) (5.294,2.848) (5.279,2.778) (5.260,2.717) (5.234,2.652) (5.207,2.595) (5.175,2.541) (5.135,2.483) (5.091,2.429) (5.049,2.385) (5.001,2.342) (4.951,2.305) (4.899,2.272) (4.846,2.244) (4.782,2.218) (4.712,2.196) (4.642,2.183) (4.572,2.176) (4.502,2.177) (4.432,2.184) (4.418,2.186)};%
\draw[pesStroke, pesLine!42!white] plot[smooth] coordinates {(5.400,2.757) (5.382,2.694) (5.358,2.626) (5.332,2.567) (5.304,2.511) (5.270,2.455) (5.232,2.400) (5.189,2.348) (5.147,2.302) (5.101,2.258) (5.049,2.216) (4.993,2.177) (4.940,2.146) (4.883,2.118) (4.825,2.095) (4.762,2.076) (4.698,2.062) (4.628,2.053) (4.558,2.051) (4.488,2.055) (4.418,2.066) (4.347,2.082) (4.280,2.104) (4.221,2.127) (4.165,2.153) (4.109,2.182) (4.053,2.216) (3.996,2.253) (3.940,2.294) (3.884,2.339) (3.832,2.385) (3.788,2.427) (3.744,2.471) (3.702,2.518) (3.661,2.567) (3.618,2.623) (3.580,2.680) (3.547,2.734) (3.516,2.792) (3.491,2.847) (3.470,2.904) (3.449,2.974) (3.435,3.043) (3.427,3.115) (3.426,3.185) (3.431,3.256) (3.444,3.326) (3.463,3.396) (3.489,3.466) (3.516,3.522) (3.547,3.577) (3.587,3.635) (3.632,3.691) (3.674,3.736) (3.730,3.788) (3.784,3.831) (3.842,3.871) (3.898,3.904) (3.954,3.933) (4.011,3.957) (4.081,3.982) (4.144,4.000) (4.144,4.000)};%
\draw[pesStroke, pesLine!42!white] plot[smooth] coordinates {(4.642,4.000) (4.712,3.980) (4.775,3.958) (4.839,3.930) (4.895,3.902) (4.951,3.869) (5.006,3.831) (5.060,3.789) (5.105,3.749) (5.149,3.705) (5.189,3.660) (5.231,3.607) (5.270,3.551) (5.302,3.496) (5.330,3.442) (5.356,3.382) (5.381,3.312) (5.400,3.246) (5.400,3.246)};%
\draw[pesStroke, pesLine!42!white] plot[smooth] coordinates {(0.684,0.319) (0.614,0.341) (0.558,0.367) (0.502,0.403) (0.452,0.446) (0.404,0.500) (0.365,0.559) (0.338,0.615) (0.316,0.685) (0.304,0.755) (0.302,0.825) (0.311,0.896) (0.330,0.966) (0.354,1.022) (0.387,1.078) (0.432,1.134) (0.477,1.177) (0.530,1.214) (0.586,1.243) (0.656,1.266) (0.726,1.277) (0.796,1.275) (0.867,1.262) (0.937,1.235) (0.993,1.201) (1.041,1.163) (1.091,1.109) (1.133,1.050) (1.161,0.996) (1.185,0.938) (1.204,0.872) (1.213,0.811) (1.217,0.741) (1.210,0.671) (1.194,0.601) (1.172,0.544) (1.139,0.488) (1.093,0.432) (1.045,0.390) (0.993,0.358) (0.937,0.333) (0.867,0.314) (0.796,0.308) (0.726,0.312) (0.684,0.319)};%
\draw[pesStroke, pesLine!40!white] plot[smooth] coordinates {(5.400,2.420) (5.360,2.357) (5.320,2.300) (5.274,2.244) (5.232,2.199) (5.175,2.147) (5.122,2.104) (5.063,2.062) (5.007,2.029) (4.951,2.000) (4.895,1.976) (4.825,1.953) (4.754,1.936) (4.684,1.926) (4.614,1.922) (4.544,1.924) (4.474,1.932) (4.404,1.945) (4.339,1.963) (4.277,1.985) (4.221,2.008) (4.167,2.033) (4.109,2.064) (4.053,2.097) (4.000,2.132) (3.940,2.173) (3.884,2.216) (3.833,2.258) (3.786,2.299) (3.738,2.342) (3.695,2.385) (3.646,2.436) (3.604,2.482) (3.561,2.532) (3.523,2.581) (3.482,2.637) (3.445,2.694) (3.412,2.750) (3.383,2.806) (3.358,2.862) (3.337,2.918) (3.316,2.989) (3.301,3.059) (3.293,3.129) (3.291,3.199) (3.295,3.270) (3.305,3.340) (3.322,3.410) (3.345,3.480) (3.369,3.536) (3.397,3.593) (3.430,3.649) (3.469,3.705) (3.514,3.761) (3.561,3.813) (3.609,3.860) (3.659,3.902) (3.716,3.944) (3.772,3.981) (3.804,4.000)};%
\draw[pesStroke, pesLine!40!white] plot[smooth] coordinates {(4.963,4.000) (5.021,3.965) (5.072,3.930) (5.127,3.888) (5.175,3.846) (5.220,3.803) (5.271,3.747) (5.316,3.692) (5.357,3.635) (5.392,3.579) (5.400,3.566)};%
\draw[pesStroke, pesLine!40!white] plot[smooth] coordinates {(0.726,0.080) (0.656,0.091) (0.587,0.109) (0.530,0.130) (0.474,0.156) (0.418,0.188) (0.361,0.227) (0.305,0.274) (0.263,0.317) (0.221,0.368) (0.186,0.418) (0.154,0.474) (0.128,0.530) (0.103,0.601) (0.086,0.671) (0.076,0.741) (0.072,0.811) (0.076,0.882) (0.086,0.952) (0.103,1.022) (0.123,1.079) (0.147,1.134) (0.178,1.191) (0.215,1.247) (0.261,1.303) (0.305,1.347) (0.352,1.387) (0.404,1.423) (0.460,1.455) (0.516,1.481) (0.586,1.505) (0.656,1.520) (0.726,1.527) (0.796,1.526) (0.867,1.516) (0.930,1.500) (0.993,1.475) (1.049,1.446) (1.105,1.409) (1.161,1.361) (1.204,1.317) (1.246,1.265) (1.286,1.205) (1.318,1.148) (1.344,1.094) (1.367,1.036) (1.390,0.966) (1.408,0.896) (1.420,0.825) (1.426,0.755) (1.427,0.685) (1.422,0.615) (1.410,0.544) (1.391,0.474) (1.368,0.418) (1.338,0.362) (1.302,0.309) (1.260,0.260) (1.204,0.209) (1.147,0.170) (1.091,0.140) (1.035,0.116) (0.966,0.095) (0.895,0.082) (0.825,0.076) (0.754,0.077) (0.726,0.080)};%
\draw[pesStroke, pesLine!38!white] plot[smooth] coordinates {(5.400,2.177) (5.358,2.128) (5.316,2.083) (5.263,2.033) (5.212,1.991) (5.161,1.954) (5.105,1.917) (5.049,1.886) (4.993,1.858) (4.937,1.836) (4.867,1.813) (4.796,1.796) (4.726,1.786) (4.656,1.782) (4.586,1.783) (4.516,1.791) (4.446,1.804) (4.375,1.823) (4.305,1.847) (4.249,1.870) (4.193,1.897) (4.137,1.927) (4.081,1.959) (4.025,1.995) (3.968,2.033) (3.912,2.074) (3.856,2.117) (3.804,2.160) (3.758,2.199) (3.707,2.244) (3.662,2.286) (3.619,2.328) (3.575,2.372) (3.533,2.416) (3.491,2.462) (3.449,2.511) (3.407,2.563) (3.365,2.619) (3.323,2.679) (3.287,2.736) (3.256,2.792) (3.228,2.848) (3.204,2.904) (3.182,2.965) (3.164,3.031) (3.151,3.101) (3.144,3.171) (3.143,3.241) (3.148,3.312) (3.160,3.382) (3.177,3.452) (3.196,3.514) (3.223,3.579) (3.250,3.635) (3.281,3.690) (3.317,3.747) (3.359,3.803) (3.406,3.860) (3.449,3.905) (3.505,3.958) (3.556,4.000) (3.556,4.000)};%
\draw[pesStroke, pesLine!38!white] plot[smooth] coordinates {(5.184,4.000) (5.232,3.961) (5.281,3.916) (5.330,3.867) (5.372,3.820) (5.400,3.786)};%
\draw[pesStroke, pesLine!38!white] plot[smooth] coordinates {(0.712,-0.117) (0.642,-0.106) (0.572,-0.090) (0.502,-0.067) (0.446,-0.045) (0.389,-0.019) (0.337,0.011) (0.277,0.049) (0.221,0.092) (0.170,0.137) (0.123,0.185) (0.081,0.234) (0.039,0.291) (0.003,0.348) (-0.027,0.404) (-0.053,0.460) (-0.074,0.516) (-0.096,0.587) (-0.111,0.657) (-0.121,0.727) (-0.126,0.797) (-0.125,0.868) (-0.119,0.938) (-0.108,1.008) (-0.091,1.078) (-0.069,1.148) (-0.047,1.205) (-0.020,1.261) (0.011,1.315) (0.049,1.373) (0.092,1.429) (0.137,1.479) (0.187,1.528) (0.235,1.567) (0.291,1.608) (0.345,1.640) (0.399,1.668) (0.460,1.694) (0.530,1.718) (0.600,1.734) (0.670,1.745) (0.740,1.749) (0.811,1.747) (0.881,1.739) (0.945,1.724) (1.007,1.705) (1.063,1.681) (1.119,1.651) (1.175,1.613) (1.228,1.570) (1.272,1.528) (1.316,1.478) (1.353,1.429) (1.390,1.373) (1.423,1.317) (1.453,1.261) (1.478,1.205) (1.502,1.148) (1.526,1.080) (1.544,1.022) (1.563,0.952) (1.579,0.882) (1.591,0.811) (1.599,0.741) (1.603,0.671) (1.603,0.601) (1.597,0.530) (1.587,0.460) (1.570,0.390) (1.545,0.320) (1.519,0.264) (1.486,0.207) (1.444,0.151) (1.400,0.102) (1.345,0.053) (1.288,0.011) (1.232,-0.023) (1.175,-0.050) (1.119,-0.073) (1.049,-0.095) (0.979,-0.110) (0.909,-0.119) (0.839,-0.123) (0.768,-0.122) (0.712,-0.117)};%
\draw[pesStroke, pesLine!35!white] plot[smooth] coordinates {(0.336,-0.200) (0.282,-0.172) (0.221,-0.136) (0.170,-0.102) (0.114,-0.060) (0.067,-0.020) (0.020,0.025) (-0.032,0.079) (-0.074,0.130) (-0.110,0.179) (-0.144,0.230) (-0.180,0.292) (-0.200,0.330)};%
\draw[pesStroke, pesLine!35!white] plot[smooth] coordinates {(-0.200,1.357) (-0.172,1.411) (-0.135,1.472) (-0.097,1.528) (-0.060,1.576) (-0.018,1.625) (0.024,1.668) (0.070,1.710) (0.121,1.753) (0.179,1.794) (0.235,1.830) (0.291,1.860) (0.347,1.887) (0.404,1.909) (0.474,1.932) (0.543,1.949) (0.614,1.962) (0.684,1.969) (0.754,1.972) (0.825,1.969) (0.885,1.963) (0.951,1.952) (1.016,1.935) (1.077,1.914) (1.133,1.889) (1.189,1.859) (1.243,1.823) (1.295,1.781) (1.339,1.738) (1.386,1.686) (1.428,1.632) (1.470,1.570) (1.503,1.514) (1.533,1.458) (1.560,1.401) (1.585,1.345) (1.608,1.289) (1.629,1.233) (1.653,1.166) (1.672,1.106) (1.692,1.036) (1.709,0.974) (1.723,0.915) (1.736,0.854) (1.749,0.783) (1.759,0.713) (1.766,0.643) (1.770,0.573) (1.770,0.502) (1.766,0.432) (1.757,0.362) (1.742,0.292) (1.723,0.227) (1.698,0.165) (1.670,0.109) (1.635,0.053) (1.596,0.003) (1.554,-0.044) (1.507,-0.088) (1.456,-0.128) (1.400,-0.166) (1.344,-0.198) (1.341,-0.200)};%
\draw[pesStroke, pesLine!35!white] plot[smooth] coordinates {(5.400,1.944) (5.345,1.893) (5.295,1.851) (5.246,1.814) (5.189,1.776) (5.133,1.743) (5.077,1.714) (5.021,1.689) (4.964,1.668) (4.895,1.648) (4.825,1.633) (4.754,1.624) (4.684,1.621) (4.614,1.624) (4.544,1.633) (4.474,1.648) (4.405,1.668) (4.347,1.689) (4.291,1.713) (4.235,1.740) (4.179,1.770) (4.123,1.803) (4.070,1.837) (4.011,1.878) (3.954,1.918) (3.898,1.961) (3.842,2.005) (3.791,2.047) (3.744,2.088) (3.693,2.132) (3.646,2.174) (3.601,2.216) (3.556,2.258) (3.512,2.300) (3.463,2.349) (3.421,2.393) (3.379,2.437) (3.337,2.483) (3.295,2.531) (3.253,2.581) (3.211,2.636) (3.168,2.693) (3.131,2.750) (3.098,2.804) (3.067,2.862) (3.042,2.917) (3.020,2.975) (3.000,3.041) (2.987,3.101) (2.977,3.171) (2.974,3.241) (2.976,3.312) (2.985,3.382) (2.999,3.452) (3.018,3.522) (3.042,3.591) (3.067,3.649) (3.095,3.705) (3.126,3.761) (3.163,3.817) (3.203,3.874) (3.249,3.930) (3.295,3.980) (3.315,4.000)};%
\draw[pesStroke, pesLine!35!white] plot[smooth] coordinates {(5.384,4.000) (5.400,3.984)};%
\draw[pesStroke, pesLine!32!white] plot[smooth] coordinates {(-0.017,-0.200) (-0.066,-0.158) (-0.111,-0.116) (-0.158,-0.067) (-0.200,-0.020) (-0.200,-0.020)};%
\draw[pesStroke, pesLine!32!white] plot[smooth] coordinates {(5.400,1.687) (5.344,1.643) (5.288,1.604) (5.234,1.570) (5.175,1.537) (5.119,1.510) (5.063,1.486) (4.993,1.462) (4.924,1.443) (4.853,1.431) (4.782,1.424) (4.712,1.423) (4.642,1.429) (4.572,1.441) (4.507,1.458) (4.446,1.478) (4.389,1.501) (4.333,1.527) (4.279,1.556) (4.221,1.590) (4.165,1.626) (4.109,1.665) (4.053,1.706) (3.996,1.749) (3.940,1.793) (3.887,1.837) (3.837,1.879) (3.788,1.921) (3.740,1.963) (3.692,2.005) (3.646,2.047) (3.599,2.090) (3.552,2.132) (3.506,2.174) (3.461,2.216) (3.416,2.258) (3.371,2.300) (3.327,2.342) (3.284,2.385) (3.241,2.427) (3.196,2.471) (3.154,2.515) (3.112,2.559) (3.070,2.606) (3.030,2.652) (2.986,2.706) (2.944,2.762) (2.904,2.820) (2.871,2.876) (2.842,2.932) (2.819,2.989) (2.796,3.059) (2.780,3.129) (2.771,3.199) (2.768,3.270) (2.771,3.340) (2.780,3.410) (2.793,3.480) (2.811,3.551) (2.832,3.613) (2.856,3.677) (2.881,3.733) (2.909,3.789) (2.941,3.845) (2.977,3.902) (3.014,3.954) (3.050,4.000)};%
\draw[pesStroke, pesLine!32!white] plot[smooth] coordinates {(-0.200,1.771) (-0.158,1.819) (-0.116,1.862) (-0.068,1.907) (-0.018,1.949) (0.038,1.991) (0.095,2.029) (0.151,2.062) (0.206,2.090) (0.263,2.115) (0.319,2.137) (0.389,2.159) (0.460,2.177) (0.530,2.190) (0.600,2.200) (0.670,2.205) (0.740,2.208) (0.811,2.206) (0.877,2.202) (0.937,2.195) (1.007,2.183) (1.077,2.167) (1.143,2.146) (1.204,2.121) (1.260,2.093) (1.311,2.062) (1.367,2.019) (1.413,1.977) (1.456,1.930) (1.495,1.879) (1.533,1.823) (1.566,1.767) (1.595,1.710) (1.622,1.654) (1.646,1.598) (1.669,1.542) (1.691,1.486) (1.717,1.415) (1.737,1.357) (1.759,1.289) (1.779,1.226) (1.798,1.163) (1.819,1.092) (1.835,1.034) (1.853,0.966) (1.871,0.896) (1.888,0.825) (1.903,0.755) (1.916,0.685) (1.928,0.615) (1.937,0.544) (1.944,0.474) (1.948,0.404) (1.948,0.334) (1.944,0.264) (1.934,0.193) (1.919,0.125) (1.901,0.067) (1.879,0.011) (1.850,-0.045) (1.815,-0.102) (1.771,-0.158) (1.733,-0.200)};%
\draw[pesStroke, pesLine!30!white] plot[smooth] coordinates {(5.400,1.358) (5.344,1.319) (5.288,1.284) (5.232,1.253) (5.175,1.225) (5.119,1.201) (5.053,1.177) (4.993,1.159) (4.923,1.145) (4.853,1.137) (4.782,1.136) (4.712,1.142) (4.642,1.155) (4.572,1.176) (4.516,1.198) (4.460,1.225) (4.404,1.257) (4.347,1.293) (4.294,1.331) (4.239,1.373) (4.187,1.415) (4.137,1.458) (4.089,1.500) (4.043,1.542) (3.997,1.584) (3.951,1.626) (3.906,1.668) (3.861,1.710) (3.817,1.753) (3.772,1.795) (3.727,1.837) (3.682,1.879) (3.637,1.921) (3.591,1.963) (3.546,2.005) (3.499,2.047) (3.453,2.090) (3.407,2.130) (3.358,2.174) (3.310,2.216) (3.261,2.258) (3.213,2.300) (3.163,2.342) (3.114,2.385) (3.064,2.427) (3.014,2.469) (2.964,2.511) (2.916,2.552) (2.866,2.595) (2.818,2.637) (2.773,2.680) (2.729,2.722) (2.677,2.777) (2.635,2.827) (2.600,2.876) (2.567,2.932) (2.542,2.989) (2.523,3.047) (2.509,3.115) (2.502,3.185) (2.502,3.256) (2.506,3.326) (2.515,3.396) (2.528,3.466) (2.544,3.536) (2.563,3.607) (2.585,3.677) (2.607,3.738) (2.633,3.803) (2.658,3.860) (2.686,3.916) (2.717,3.972) (2.733,4.000)};%
\draw[pesStroke, pesLine!30!white] plot[smooth] coordinates {(-0.200,2.189) (-0.149,2.230) (-0.089,2.272) (-0.032,2.307) (0.025,2.337) (0.081,2.363) (0.137,2.384) (0.207,2.407) (0.277,2.425) (0.347,2.438) (0.418,2.448) (0.487,2.455) (0.558,2.459) (0.628,2.461) (0.698,2.461) (0.768,2.460) (0.839,2.456) (0.909,2.452) (0.979,2.445) (1.049,2.436) (1.113,2.427) (1.175,2.415) (1.246,2.399) (1.316,2.378) (1.372,2.356) (1.428,2.329) (1.484,2.294) (1.540,2.247) (1.582,2.202) (1.622,2.146) (1.653,2.092) (1.680,2.033) (1.703,1.977) (1.728,1.907) (1.750,1.837) (1.771,1.767) (1.791,1.696) (1.807,1.637) (1.825,1.570) (1.844,1.500) (1.863,1.431) (1.883,1.359) (1.902,1.289) (1.919,1.227) (1.938,1.163) (1.958,1.092) (1.975,1.030) (1.994,0.966) (2.014,0.896) (2.032,0.833) (2.050,0.769) (2.069,0.699) (2.088,0.630) (2.106,0.559) (2.123,0.488) (2.139,0.418) (2.152,0.348) (2.163,0.278) (2.171,0.207) (2.176,0.137) (2.175,0.067) (2.169,-0.003) (2.155,-0.074) (2.134,-0.144) (2.110,-0.200) (2.110,-0.200)};%
\draw[pesStroke, pesLine!28!white] plot[smooth] coordinates {(5.400,0.776) (5.344,0.733) (5.288,0.693) (5.233,0.657) (5.175,0.621) (5.119,0.589) (5.063,0.561) (5.007,0.535) (4.951,0.513) (4.881,0.492) (4.811,0.478) (4.740,0.473) (4.670,0.481) (4.606,0.502) (4.544,0.541) (4.496,0.587) (4.450,0.643) (4.412,0.699) (4.378,0.755) (4.347,0.808) (4.314,0.868) (4.281,0.924) (4.249,0.979) (4.215,1.036) (4.180,1.092) (4.143,1.148) (4.109,1.200) (4.067,1.261) (4.026,1.317) (3.984,1.373) (3.940,1.429) (3.898,1.481) (3.859,1.528) (3.814,1.580) (3.774,1.626) (3.730,1.674) (3.688,1.719) (3.646,1.763) (3.600,1.809) (3.558,1.851) (3.514,1.893) (3.469,1.935) (3.422,1.977) (3.374,2.019) (3.325,2.062) (3.274,2.104) (3.225,2.144) (3.168,2.187) (3.112,2.228) (3.056,2.268) (3.000,2.306) (2.944,2.342) (2.888,2.376) (2.832,2.408) (2.775,2.438) (2.719,2.466) (2.663,2.492) (2.607,2.515) (2.551,2.537) (2.481,2.560) (2.411,2.580) (2.341,2.595) (2.270,2.607) (2.200,2.613) (2.130,2.610) (2.061,2.581) (2.021,2.525) (1.996,2.455) (1.980,2.385) (1.968,2.314) (1.958,2.244) (1.952,2.174) (1.948,2.104) (1.946,2.033) (1.946,1.963) (1.949,1.893) (1.953,1.823) (1.960,1.753) (1.968,1.682) (1.978,1.612) (1.989,1.542) (2.002,1.472) (2.017,1.401) (2.032,1.337) (2.047,1.275) (2.065,1.205) (2.085,1.134) (2.106,1.064) (2.128,0.994) (2.152,0.924) (2.172,0.866) (2.197,0.797) (2.218,0.741) (2.240,0.685) (2.263,0.629) (2.287,0.573) (2.311,0.516) (2.336,0.460) (2.362,0.404) (2.388,0.348) (2.415,0.292) (2.442,0.235) (2.470,0.179) (2.497,0.123) (2.525,0.067) (2.551,0.012) (2.577,-0.045) (2.600,-0.102) (2.621,-0.160) (2.633,-0.200)};%
\draw[pesStroke, pesLine!28!white] plot[smooth] coordinates {(-0.200,2.735) (-0.130,2.755) (-0.060,2.769) (0.011,2.776) (0.081,2.779) (0.151,2.779) (0.221,2.775) (0.291,2.770) (0.355,2.764) (0.418,2.757) (0.486,2.750) (0.558,2.742) (0.628,2.734) (0.698,2.727) (0.768,2.720) (0.839,2.713) (0.909,2.708) (0.979,2.703) (1.049,2.699) (1.119,2.695) (1.189,2.693) (1.260,2.691) (1.330,2.690) (1.400,2.689) (1.470,2.689) (1.540,2.689) (1.611,2.690) (1.681,2.691) (1.751,2.693) (1.821,2.698) (1.891,2.707) (1.959,2.736) (2.004,2.790) (2.031,2.848) (2.057,2.918) (2.080,2.989) (2.101,3.059) (2.123,3.129) (2.144,3.199) (2.164,3.270) (2.184,3.340) (2.204,3.410) (2.224,3.480) (2.242,3.546) (2.259,3.607) (2.278,3.677) (2.297,3.747) (2.316,3.817) (2.336,3.888) (2.354,3.950) (2.370,4.000)};%
\draw[pesStroke, pesLine!25!white] plot[smooth] coordinates {(-0.200,3.512) (-0.144,3.467) (-0.092,3.424) (-0.041,3.382) (0.011,3.341) (0.067,3.298) (0.123,3.257) (0.179,3.219) (0.234,3.185) (0.291,3.153) (0.347,3.123) (0.404,3.097) (0.459,3.073) (0.516,3.051) (0.576,3.031) (0.642,3.011) (0.712,2.993) (0.782,2.979) (0.853,2.968) (0.914,2.961) (0.979,2.955) (1.049,2.953) (1.119,2.954) (1.189,2.959) (1.260,2.967) (1.330,2.980) (1.400,2.997) (1.458,3.017) (1.524,3.045) (1.577,3.073) (1.639,3.113) (1.693,3.157) (1.737,3.201) (1.779,3.251) (1.821,3.310) (1.856,3.368) (1.885,3.424) (1.910,3.480) (1.932,3.536) (1.955,3.607) (1.975,3.677) (1.989,3.737) (2.003,3.803) (2.014,3.874) (2.023,3.944) (2.028,4.000)};%
\draw[pesStroke, pesLine!25!white] plot[smooth] coordinates {(3.056,0.206) (3.126,0.198) (3.196,0.197) (3.267,0.203) (3.337,0.216) (3.406,0.235) (3.463,0.257) (3.519,0.283) (3.575,0.315) (3.632,0.354) (3.688,0.400) (3.734,0.446) (3.782,0.502) (3.820,0.559) (3.851,0.615) (3.876,0.671) (3.898,0.741) (3.912,0.811) (3.918,0.882) (3.918,0.952) (3.910,1.022) (3.895,1.092) (3.875,1.163) (3.849,1.233) (3.824,1.289) (3.796,1.345) (3.764,1.401) (3.730,1.456) (3.691,1.514) (3.649,1.570) (3.605,1.626) (3.561,1.677) (3.519,1.724) (3.477,1.768) (3.423,1.823) (3.378,1.865) (3.330,1.907) (3.281,1.949) (3.227,1.991) (3.170,2.033) (3.112,2.073) (3.056,2.108) (3.000,2.141) (2.944,2.170) (2.888,2.196) (2.832,2.218) (2.761,2.240) (2.691,2.257) (2.621,2.265) (2.551,2.266) (2.481,2.256) (2.411,2.233) (2.354,2.203) (2.301,2.160) (2.256,2.108) (2.219,2.047) (2.194,1.991) (2.172,1.925) (2.158,1.865) (2.148,1.795) (2.142,1.724) (2.141,1.654) (2.144,1.584) (2.150,1.514) (2.158,1.450) (2.168,1.387) (2.182,1.317) (2.199,1.247) (2.218,1.177) (2.240,1.106) (2.265,1.036) (2.286,0.980) (2.310,0.924) (2.335,0.868) (2.362,0.811) (2.392,0.755) (2.424,0.699) (2.460,0.643) (2.495,0.592) (2.537,0.537) (2.578,0.488) (2.621,0.442) (2.677,0.390) (2.729,0.348) (2.789,0.307) (2.846,0.275) (2.901,0.249) (2.958,0.229) (3.028,0.211) (3.056,0.206)};%
\draw[pesStroke, pesLine!23!white] plot[smooth] coordinates {(0.249,4.000) (0.253,3.930) (0.263,3.860) (0.277,3.796) (0.297,3.733) (0.319,3.677) (0.347,3.621) (0.381,3.565) (0.418,3.515) (0.459,3.466) (0.502,3.424) (0.552,3.382) (0.612,3.340) (0.670,3.306) (0.726,3.280) (0.782,3.259) (0.853,3.238) (0.923,3.224) (0.993,3.217) (1.063,3.215) (1.133,3.220) (1.204,3.232) (1.274,3.252) (1.330,3.273) (1.386,3.302) (1.442,3.338) (1.495,3.382) (1.540,3.428) (1.581,3.480) (1.617,3.536) (1.646,3.593) (1.668,3.649) (1.690,3.719) (1.704,3.789) (1.713,3.860) (1.715,3.930) (1.711,4.000) (1.711,4.000)};%
\draw[pesStroke, pesLine!23!white] plot[smooth] coordinates {(3.014,0.544) (3.084,0.539) (3.154,0.543) (3.225,0.555) (3.284,0.573) (3.350,0.601) (3.407,0.633) (3.458,0.671) (3.504,0.713) (3.547,0.763) (3.588,0.825) (3.615,0.882) (3.638,0.952) (3.650,1.022) (3.654,1.092) (3.648,1.163) (3.635,1.233) (3.618,1.293) (3.592,1.359) (3.566,1.415) (3.535,1.472) (3.499,1.528) (3.463,1.579) (3.426,1.626) (3.379,1.680) (3.337,1.724) (3.292,1.767) (3.244,1.809) (3.196,1.847) (3.140,1.887) (3.088,1.921) (3.028,1.955) (2.972,1.981) (2.916,2.004) (2.846,2.024) (2.775,2.037) (2.705,2.040) (2.640,2.033) (2.579,2.017) (2.523,1.990) (2.467,1.949) (2.425,1.904) (2.389,1.851) (2.362,1.795) (2.340,1.728) (2.328,1.668) (2.319,1.598) (2.317,1.528) (2.321,1.458) (2.329,1.387) (2.340,1.321) (2.354,1.261) (2.375,1.191) (2.396,1.128) (2.423,1.064) (2.449,1.008) (2.479,0.952) (2.513,0.896) (2.551,0.842) (2.593,0.789) (2.635,0.743) (2.683,0.699) (2.733,0.659) (2.789,0.623) (2.846,0.593) (2.902,0.571) (2.972,0.552) (3.014,0.544)};%
\draw[pesStroke, pesLine!20!white] plot[smooth] coordinates {(0.746,4.000) (0.726,3.930) (0.724,3.860) (0.737,3.789) (0.760,3.733) (0.796,3.679) (0.842,3.635) (0.895,3.600) (0.951,3.577) (1.017,3.565) (1.077,3.565) (1.138,3.579) (1.199,3.607) (1.253,3.649) (1.298,3.705) (1.325,3.761) (1.341,3.831) (1.341,3.902) (1.326,3.972) (1.315,4.000)};%
\draw[pesStroke, pesLine!20!white] plot[smooth] coordinates {(2.972,0.838) (3.042,0.833) (3.111,0.839) (3.182,0.861) (3.239,0.890) (3.295,0.934) (3.335,0.980) (3.369,1.036) (3.390,1.092) (3.403,1.163) (3.403,1.233) (3.393,1.299) (3.376,1.359) (3.353,1.415) (3.323,1.472) (3.286,1.528) (3.240,1.584) (3.196,1.630) (3.140,1.678) (3.084,1.719) (3.028,1.751) (2.972,1.776) (2.913,1.795) (2.846,1.806) (2.775,1.804) (2.705,1.786) (2.649,1.755) (2.602,1.710) (2.565,1.655) (2.541,1.598) (2.525,1.528) (2.519,1.458) (2.522,1.387) (2.534,1.317) (2.551,1.253) (2.574,1.191) (2.600,1.134) (2.633,1.078) (2.674,1.022) (2.719,0.973) (2.775,0.925) (2.832,0.888) (2.888,0.862) (2.958,0.841) (2.972,0.838)};%
}%
\def\mepPath{(0.766,0.786) (0.776,0.840) (0.789,0.892) (0.803,0.945) (0.819,0.997) (0.836,1.049) (0.853,1.100) (0.872,1.151) (0.891,1.202) (0.911,1.253) (0.932,1.303) (0.954,1.353) (0.976,1.403) (0.999,1.452) (1.022,1.501) (1.047,1.550) (1.072,1.598) (1.098,1.646) (1.125,1.693) (1.152,1.740) (1.181,1.786) (1.210,1.832) (1.240,1.877) (1.271,1.922) (1.303,1.966) (1.336,2.009) (1.370,2.052) (1.404,2.094) (1.439,2.136) (1.475,2.177) (1.512,2.217) (1.549,2.256) (1.587,2.295) (1.626,2.333) (1.665,2.371) (1.705,2.408) (1.746,2.444) (1.787,2.480) (1.828,2.515) (1.870,2.550) (1.913,2.584) (1.956,2.618) (1.999,2.650) (2.043,2.682) (2.087,2.714) (2.132,2.744) (2.178,2.774) (2.224,2.803) (2.271,2.831) (2.318,2.858) (2.366,2.884) (2.414,2.909) (2.463,2.933) (2.513,2.955) (2.563,2.976) (2.613,2.996) (2.665,3.015) (2.716,3.032) (2.768,3.047) (2.821,3.062) (2.874,3.074) (2.927,3.086) (2.980,3.096) (3.034,3.105) (3.088,3.113) (3.142,3.120) (3.196,3.126) (3.250,3.130) (3.304,3.134) (3.359,3.137) (3.413,3.139) (3.467,3.140) (3.522,3.141) (3.576,3.140) (3.631,3.139) (3.685,3.138) (3.739,3.135) (3.794,3.133) (3.848,3.129) (3.902,3.125) (3.956,3.120) (4.011,3.115) (4.065,3.109) (4.119,3.103) (4.173,3.096) (4.227,3.088) (4.280,3.079) (4.334,3.070) (4.387,3.059) (4.440,3.045)}%
\def\Rpt{0.766,0.786}%
\def\Ppt{4.440,3.045}%
\def\TSpt{2.013,2.661}%
\def\priorImgs{1.181/0.983, 1.603/1.244, 2.081/1.699, 2.474/1.811, 3.034/2.185, 3.485/2.378, 3.956/2.804}%
\def\priorMeans{1.203/1.055, 1.640/1.323, 2.076/1.592, 2.513/1.860, 2.995/2.157, 3.476/2.453, 3.958/2.749}%
\def\priorPath{(0.766,0.786) (0.807,0.810) (0.847,0.833) (0.888,0.854) (0.927,0.874) (0.966,0.892) (1.005,0.909) (1.044,0.926) (1.082,0.942) (1.121,0.958) (1.159,0.974) (1.197,0.990) (1.235,1.007) (1.273,1.024) (1.311,1.042) (1.349,1.062) (1.389,1.084) (1.428,1.109) (1.469,1.136) (1.510,1.166) (1.552,1.200) (1.596,1.237) (1.641,1.279) (1.686,1.324) (1.733,1.372) (1.780,1.420) (1.827,1.469) (1.873,1.517) (1.919,1.563) (1.964,1.606) (2.008,1.645) (2.051,1.678) (2.091,1.705) (2.130,1.726) (2.167,1.741) (2.203,1.751) (2.238,1.758) (2.273,1.764) (2.309,1.769) (2.344,1.774) (2.381,1.781) (2.420,1.791) (2.460,1.805) (2.503,1.824) (2.548,1.849) (2.594,1.878) (2.642,1.910) (2.692,1.945) (2.742,1.982) (2.792,2.019) (2.842,2.056) (2.892,2.092) (2.941,2.126) (2.988,2.158) (3.034,2.185) (3.078,2.208) (3.121,2.228) (3.161,2.244) (3.201,2.258) (3.239,2.271) (3.277,2.283) (3.314,2.295) (3.350,2.308) (3.387,2.323) (3.423,2.341) (3.460,2.362) (3.498,2.387) (3.536,2.417) (3.574,2.450) (3.614,2.487) (3.654,2.525) (3.694,2.565) (3.735,2.606) (3.776,2.647) (3.817,2.687) (3.859,2.725) (3.900,2.761) (3.942,2.794) (3.984,2.824) (4.025,2.849) (4.067,2.872) (4.108,2.893) (4.150,2.911) (4.191,2.929) (4.233,2.946) (4.274,2.963) (4.316,2.981) (4.357,3.000) (4.398,3.021) (4.440,3.045)}%
\def\infImgs{1.198/0.869, 1.674/1.153, 2.149/1.640, 2.580/1.927, 3.019/2.114, 3.452/2.573, 3.911/2.731}%
\def\infPath{(0.766,0.786) (0.807,0.808) (0.848,0.824) (0.888,0.836) (0.928,0.843) (0.969,0.848) (1.009,0.851) (1.050,0.853) (1.090,0.855) (1.132,0.859) (1.173,0.864) (1.216,0.873) (1.259,0.886) (1.303,0.902) (1.347,0.922) (1.392,0.945) (1.437,0.971) (1.483,1.000) (1.529,1.032) (1.574,1.068) (1.620,1.106) (1.667,1.146) (1.713,1.189) (1.758,1.234) (1.804,1.281) (1.850,1.329) (1.895,1.377) (1.940,1.425) (1.985,1.473) (2.029,1.520) (2.073,1.566) (2.117,1.609) (2.160,1.650) (2.203,1.689) (2.245,1.725) (2.286,1.758) (2.328,1.789) (2.368,1.817) (2.409,1.842) (2.449,1.866) (2.489,1.886) (2.528,1.905) (2.567,1.922) (2.606,1.936) (2.644,1.949) (2.682,1.960) (2.720,1.971) (2.758,1.983) (2.795,1.994) (2.833,2.008) (2.870,2.023) (2.907,2.040) (2.945,2.061) (2.982,2.085) (3.019,2.114) (3.056,2.147) (3.093,2.184) (3.130,2.225) (3.167,2.268) (3.204,2.312) (3.242,2.356) (3.279,2.400) (3.316,2.443) (3.353,2.483) (3.390,2.520) (3.427,2.554) (3.464,2.582) (3.502,2.606) (3.539,2.625) (3.576,2.640) (3.614,2.652) (3.653,2.663) (3.691,2.672) (3.731,2.680) (3.771,2.689) (3.812,2.699) (3.854,2.711) (3.897,2.725) (3.941,2.743) (3.986,2.764) (4.032,2.788) (4.078,2.814) (4.125,2.842) (4.172,2.871) (4.218,2.901) (4.264,2.931) (4.310,2.961) (4.354,2.990) (4.398,3.019) (4.440,3.045)}%
\def\sigmaCm{0.080}%
\def\oneStep{1.198/0.869/0.944/1.332, 1.674/1.153/1.215/1.839, 2.149/1.640/1.578/2.285, 2.580/1.927/2.013/2.661, 3.019/2.114/2.562/2.976, 3.452/2.573/3.178/3.124, 3.911/2.731/3.812/3.132}%
\def\dataImgs{0.944/1.332, 1.215/1.839, 1.578/2.285, 2.013/2.661, 2.562/2.976, 3.178/3.124, 3.812/3.132}%
\def\mepImgs{0.944/1.332, 1.215/1.839, 1.578/2.285, 2.013/2.661, 2.562/2.976, 3.178/3.124, 3.812/3.132}%
\def\pathB{(0.766,0.786) (0.803,0.814) (0.840,0.840) (0.877,0.865) (0.914,0.888) (0.951,0.911) (0.987,0.932) (1.023,0.953) (1.060,0.973) (1.096,0.994) (1.132,1.014) (1.168,1.034) (1.204,1.054) (1.240,1.075) (1.276,1.097) (1.313,1.121) (1.351,1.146) (1.389,1.173) (1.427,1.203) (1.467,1.235) (1.508,1.270) (1.550,1.309) (1.593,1.351) (1.637,1.396) (1.681,1.443) (1.726,1.491) (1.772,1.539) (1.817,1.586) (1.862,1.632) (1.906,1.674) (1.949,1.713) (1.991,1.747) (2.031,1.776) (2.069,1.799) (2.107,1.816) (2.143,1.830) (2.179,1.841) (2.215,1.850) (2.251,1.858) (2.288,1.867) (2.325,1.877) (2.364,1.890) (2.405,1.906) (2.448,1.927) (2.492,1.953) (2.539,1.982) (2.587,2.014) (2.636,2.048) (2.685,2.084) (2.735,2.120) (2.785,2.156) (2.834,2.190) (2.883,2.223) (2.931,2.253) (2.978,2.280) (3.022,2.303) (3.066,2.322) (3.108,2.339) (3.149,2.353) (3.189,2.366) (3.228,2.378) (3.267,2.390) (3.306,2.403) (3.344,2.417) (3.383,2.434) (3.422,2.453) (3.461,2.476) (3.501,2.502) (3.542,2.532) (3.583,2.565) (3.625,2.599) (3.667,2.634) (3.709,2.670) (3.752,2.706) (3.795,2.741) (3.838,2.775) (3.881,2.806) (3.924,2.835) (3.967,2.860) (4.010,2.882) (4.054,2.902) (4.097,2.919) (4.140,2.935) (4.183,2.950) (4.226,2.964) (4.269,2.978) (4.311,2.992) (4.354,3.008) (4.397,3.026) (4.440,3.045)}%
\def\imgsB{1.152/1.025, 1.557/1.315, 2.021/1.769, 2.419/1.913, 2.977/2.280, 3.448/2.468, 3.938/2.844}%
\def\pathBtwo{(0.766,0.786) (0.806,0.814) (0.850,0.841) (0.896,0.867) (0.946,0.893) (0.996,0.918) (1.047,0.942) (1.097,0.966) (1.145,0.990) (1.192,1.014) (1.235,1.038) (1.274,1.062) (1.308,1.087) (1.340,1.112) (1.368,1.138) (1.395,1.165) (1.420,1.192) (1.446,1.220) (1.472,1.250) (1.499,1.280) (1.529,1.311) (1.562,1.343) (1.599,1.376) (1.639,1.410) (1.682,1.444) (1.728,1.479) (1.775,1.512) (1.824,1.545) (1.874,1.576) (1.923,1.605) (1.973,1.632) (2.021,1.655) (2.068,1.676) (2.114,1.693) (2.157,1.707) (2.200,1.720) (2.242,1.731) (2.283,1.742) (2.324,1.753) (2.365,1.766) (2.406,1.780) (2.448,1.798) (2.491,1.819) (2.536,1.844) (2.582,1.873) (2.628,1.906) (2.676,1.941) (2.724,1.979) (2.773,2.018) (2.821,2.058) (2.869,2.099) (2.917,2.138) (2.963,2.177) (3.009,2.214) (3.053,2.248) (3.096,2.280) (3.138,2.308) (3.178,2.335) (3.217,2.359) (3.254,2.382) (3.291,2.403) (3.327,2.424) (3.363,2.445) (3.397,2.465) (3.432,2.486) (3.466,2.508) (3.500,2.532) (3.534,2.556) (3.568,2.582) (3.602,2.608) (3.637,2.636) (3.672,2.663) (3.707,2.691) (3.744,2.718) (3.781,2.745) (3.820,2.771) (3.859,2.796) (3.900,2.820) (3.942,2.842) (3.986,2.863) (4.030,2.882) (4.076,2.901) (4.122,2.918) (4.168,2.936) (4.214,2.953) (4.261,2.970) (4.307,2.988) (4.352,3.006) (4.396,3.025) (4.440,3.045)}%
\def\imgsBtwo{1.258/1.052, 1.568/1.348, 2.057/1.671, 2.506/1.826, 3.053/2.248, 3.488/2.524, 3.914/2.827}%
\def\pathBthree{(0.766,0.786) (0.802,0.815) (0.838,0.845) (0.872,0.876) (0.906,0.907) (0.940,0.939) (0.975,0.971) (1.010,1.003) (1.048,1.036) (1.087,1.068) (1.129,1.099) (1.173,1.131) (1.221,1.161) (1.271,1.191) (1.323,1.220) (1.376,1.248) (1.428,1.276) (1.479,1.301) (1.529,1.326) (1.577,1.349) (1.622,1.370) (1.662,1.390) (1.699,1.408) (1.731,1.424) (1.760,1.440) (1.786,1.456) (1.812,1.473) (1.837,1.491) (1.862,1.512) (1.889,1.535) (1.919,1.562) (1.951,1.593) (1.987,1.630) (2.028,1.671) (2.073,1.716) (2.121,1.764) (2.171,1.813) (2.224,1.862) (2.277,1.910) (2.331,1.956) (2.385,1.997) (2.438,2.033) (2.490,2.063) (2.539,2.086) (2.586,2.101) (2.630,2.110) (2.673,2.115) (2.714,2.117) (2.753,2.117) (2.791,2.117) (2.827,2.119) (2.862,2.124) (2.896,2.132) (2.930,2.147) (2.962,2.168) (2.994,2.198) (3.026,2.234) (3.058,2.276) (3.090,2.322) (3.123,2.371) (3.156,2.420) (3.191,2.469) (3.227,2.516) (3.264,2.559) (3.303,2.597) (3.344,2.629) (3.388,2.653) (3.434,2.668) (3.482,2.676) (3.531,2.679) (3.582,2.677) (3.634,2.673) (3.686,2.667) (3.738,2.662) (3.790,2.657) (3.841,2.656) (3.891,2.659) (3.939,2.668) (3.986,2.684) (4.030,2.707) (4.074,2.736) (4.116,2.769) (4.156,2.805) (4.197,2.844) (4.237,2.883) (4.276,2.921) (4.316,2.958) (4.357,2.992) (4.398,3.021) (4.440,3.045)}%
\def\imgsBthree{1.154/1.118, 1.669/1.393, 1.978/1.620, 2.506/2.072, 2.962/2.168, 3.373/2.646, 3.955/2.673}%
\def\arrowsB{1.152/1.025/0.927/1.237, 1.557/1.315/1.217/1.715, 2.021/1.769/1.645/2.312, 2.419/1.913/2.000/2.524, 2.977/2.280/2.597/2.931, 3.448/2.468/3.230/2.929, 3.938/2.844/3.835/3.164}%
\def\dispB{2.057/1.671/1.578/2.285, 1.978/1.620/1.578/2.285}%
\def\hiPt{2.021,1.769}%
\def\pathC{(0.766,0.786) (0.798,0.819) (0.830,0.851) (0.862,0.881) (0.895,0.911) (0.927,0.939) (0.960,0.967) (0.992,0.994) (1.025,1.020) (1.058,1.047) (1.091,1.073) (1.124,1.099) (1.157,1.125) (1.190,1.152) (1.224,1.180) (1.259,1.209) (1.294,1.238) (1.329,1.270) (1.366,1.303) (1.403,1.338) (1.441,1.376) (1.480,1.416) (1.521,1.458) (1.562,1.503) (1.604,1.550) (1.647,1.597) (1.689,1.644) (1.732,1.690) (1.775,1.735) (1.818,1.777) (1.859,1.816) (1.900,1.852) (1.940,1.882) (1.979,1.908) (2.017,1.930) (2.054,1.948) (2.091,1.964) (2.127,1.979) (2.164,1.993) (2.202,2.007) (2.241,2.022) (2.281,2.039) (2.322,2.059) (2.365,2.082) (2.410,2.108) (2.456,2.138) (2.503,2.169) (2.551,2.203) (2.600,2.237) (2.650,2.271) (2.699,2.305) (2.748,2.337) (2.797,2.368) (2.846,2.397) (2.893,2.423) (2.939,2.445) (2.984,2.464) (3.028,2.480) (3.071,2.495) (3.114,2.508) (3.156,2.520) (3.198,2.532) (3.239,2.545) (3.281,2.558) (3.323,2.573) (3.365,2.589) (3.407,2.609) (3.450,2.631) (3.493,2.655) (3.537,2.682) (3.582,2.709) (3.626,2.738) (3.671,2.766) (3.716,2.795) (3.761,2.822) (3.807,2.849) (3.852,2.873) (3.897,2.896) (3.943,2.915) (3.988,2.932) (4.034,2.947) (4.079,2.959) (4.124,2.971) (4.170,2.981) (4.215,2.991) (4.260,3.000) (4.305,3.010) (4.350,3.021) (4.395,3.032) (4.440,3.045)}%
\def\imgsC{1.110/1.088, 1.487/1.423, 1.930/1.875, 2.336/2.066, 2.893/2.422, 3.393/2.602, 3.913/2.903}%
\def\arrowsC{1.110/1.088/0.914/1.273, 2.336/2.066/1.970/2.598, 3.913/2.903/3.822/3.183}%
\def\pathD{(0.766,0.786) (0.788,0.828) (0.811,0.870) (0.835,0.910) (0.860,0.950) (0.885,0.989) (0.911,1.028) (0.937,1.066) (0.964,1.103) (0.991,1.141) (1.018,1.178) (1.046,1.215) (1.074,1.252) (1.103,1.289) (1.132,1.327) (1.162,1.364) (1.192,1.403) (1.223,1.442) (1.255,1.481) (1.288,1.522) (1.322,1.563) (1.357,1.606) (1.392,1.650) (1.429,1.695) (1.466,1.740) (1.505,1.786) (1.543,1.831) (1.582,1.875) (1.622,1.918) (1.661,1.960) (1.701,1.999) (1.740,2.037) (1.779,2.071) (1.818,2.103) (1.856,2.131) (1.894,2.158) (1.933,2.184) (1.972,2.208) (2.011,2.232) (2.050,2.255) (2.091,2.279) (2.132,2.303) (2.174,2.329) (2.218,2.356) (2.262,2.385) (2.308,2.415) (2.355,2.446) (2.402,2.477) (2.450,2.508) (2.498,2.539) (2.547,2.569) (2.596,2.599) (2.645,2.626) (2.693,2.652) (2.742,2.676) (2.790,2.697) (2.838,2.716) (2.885,2.732) (2.933,2.747) (2.980,2.761) (3.027,2.774) (3.074,2.785) (3.121,2.797) (3.168,2.808) (3.215,2.820) (3.263,2.832) (3.310,2.845) (3.359,2.859) (3.407,2.874) (3.456,2.890) (3.504,2.906) (3.554,2.922) (3.603,2.937) (3.652,2.953) (3.701,2.967) (3.751,2.981) (3.800,2.993) (3.850,3.004) (3.900,3.013) (3.949,3.020) (3.998,3.026) (4.048,3.031) (4.097,3.034) (4.146,3.037) (4.196,3.039) (4.245,3.040) (4.294,3.042) (4.343,3.043) (4.391,3.044) (4.440,3.045)}%
\def\imgsD{1.034/1.200, 1.363/1.613, 1.769/2.063, 2.188/2.338, 2.742/2.676, 3.294/2.841, 3.866/3.007}%
\def\arrowsD{1.363/1.613/1.177/1.831, 2.742/2.676/2.535/3.031}%
\def\pathE{(0.766,0.786) (0.780,0.835) (0.797,0.884) (0.815,0.932) (0.834,0.980) (0.854,1.027) (0.875,1.073) (0.896,1.120) (0.918,1.166) (0.941,1.211) (0.964,1.257) (0.988,1.302) (1.012,1.347) (1.037,1.392) (1.063,1.437) (1.089,1.481) (1.116,1.526) (1.144,1.570) (1.173,1.615) (1.202,1.659) (1.233,1.704) (1.264,1.749) (1.296,1.794) (1.329,1.838) (1.363,1.883) (1.398,1.927) (1.434,1.971) (1.470,2.014) (1.506,2.056) (1.544,2.097) (1.581,2.137) (1.619,2.175) (1.658,2.213) (1.697,2.248) (1.736,2.283) (1.775,2.316) (1.815,2.348) (1.855,2.380) (1.895,2.411) (1.937,2.441) (1.978,2.472) (2.021,2.502) (2.064,2.532) (2.107,2.562) (2.152,2.593) (2.197,2.623) (2.243,2.653) (2.290,2.683) (2.337,2.712) (2.384,2.741) (2.432,2.768) (2.481,2.795) (2.530,2.820) (2.579,2.843) (2.629,2.866) (2.678,2.886) (2.728,2.904) (2.779,2.921) (2.829,2.937) (2.879,2.951) (2.930,2.964) (2.981,2.975) (3.032,2.986) (3.083,2.996) (3.135,3.005) (3.186,3.014) (3.238,3.022) (3.290,3.030) (3.342,3.038) (3.394,3.046) (3.447,3.053) (3.499,3.060) (3.552,3.066) (3.604,3.071) (3.657,3.076) (3.709,3.080) (3.762,3.083) (3.814,3.085) (3.867,3.086) (3.919,3.086) (3.972,3.086) (4.024,3.084) (4.077,3.082) (4.129,3.079) (4.181,3.075) (4.233,3.070) (4.285,3.066) (4.337,3.060) (4.389,3.053) (4.440,3.045)}%
\def\imgsE{0.978/1.283, 1.269/1.756, 1.648/2.203, 2.078/2.542, 2.628/2.865, 3.221/3.019, 3.832/3.086}%
\def\arrowsE{1.648/2.203/1.567/2.321, 3.221/3.019/3.174/3.119}%
\def\vCurve{(0.000,1.000) (0.034,0.976) (0.068,0.951) (0.102,0.923) (0.136,0.895) (0.169,0.865) (0.203,0.834) (0.237,0.803) (0.271,0.772) (0.305,0.741) (0.339,0.710) (0.373,0.679) (0.407,0.648) (0.441,0.618) (0.475,0.589) (0.508,0.560) (0.542,0.532) (0.576,0.505) (0.610,0.479) (0.644,0.453) (0.678,0.429) (0.712,0.405) (0.746,0.383) (0.780,0.361) (0.814,0.340) (0.847,0.320) (0.881,0.301) (0.915,0.283) (0.949,0.265) (0.983,0.249) (1.017,0.233) (1.051,0.218) (1.085,0.203) (1.119,0.189) (1.153,0.176) (1.186,0.164) (1.220,0.152) (1.254,0.141) (1.288,0.130) (1.322,0.120) (1.356,0.111) (1.390,0.102) (1.424,0.093) (1.458,0.085) (1.492,0.077) (1.525,0.070) (1.559,0.063) (1.593,0.056) (1.627,0.050) (1.661,0.044) (1.695,0.039) (1.729,0.033) (1.763,0.028) (1.797,0.024) (1.831,0.019) (1.864,0.015) (1.898,0.011) (1.932,0.007) (1.966,0.003) (2.000,0.000)}%
\def\vMarkC{0.205,0.833}%
\def\vMarkD{0.547,0.529}%
\def\vMarkE{1.061,0.213}%
\def\vMarkF{2.000,0.000}%
\def\Unega{2.457,2.986}%
\def\Unegb{1.570,2.336}%
\def\Uposa{2.250,2.338}%
\def\Uposb{1.777,2.983}%
\def\energyCurve{(0.450,1.504) (0.499,1.507) (0.548,1.515) (0.597,1.528) (0.646,1.547) (0.694,1.570) (0.743,1.598) (0.792,1.630) (0.841,1.666) (0.890,1.706) (0.939,1.749) (0.988,1.794) (1.037,1.842) (1.085,1.892) (1.134,1.944) (1.183,1.997) (1.232,2.050) (1.281,2.104) (1.330,2.158) (1.379,2.212) (1.428,2.265) (1.476,2.318) (1.525,2.369) (1.574,2.419) (1.623,2.468) (1.672,2.515) (1.721,2.560) (1.770,2.603) (1.819,2.644) (1.867,2.684) (1.916,2.720) (1.965,2.755) (2.014,2.787) (2.063,2.816) (2.112,2.843) (2.161,2.866) (2.210,2.887) (2.258,2.905) (2.307,2.921) (2.356,2.933) (2.405,2.941) (2.454,2.947) (2.503,2.950) (2.552,2.949) (2.601,2.946) (2.649,2.939) (2.698,2.929) (2.747,2.916) (2.796,2.899) (2.845,2.880) (2.894,2.858) (2.943,2.832) (2.992,2.804) (3.040,2.772) (3.089,2.737) (3.138,2.699) (3.187,2.658) (3.236,2.613) (3.285,2.565) (3.334,2.513) (3.383,2.458) (3.431,2.400) (3.480,2.337) (3.529,2.271) (3.578,2.202) (3.627,2.129) (3.676,2.053) (3.725,1.974) (3.774,1.892) (3.822,1.808) (3.871,1.722) (3.920,1.634) (3.969,1.546) (4.018,1.457) (4.067,1.369) (4.116,1.281) (4.165,1.196) (4.213,1.112) (4.262,1.032) (4.311,0.956) (4.360,0.885) (4.409,0.819) (4.458,0.758) (4.507,0.705) (4.556,0.659) (4.604,0.620) (4.653,0.590) (4.702,0.568) (4.751,0.554) (4.800,0.550)}%
\def\energyImgs{0.967/1.775, 1.484/2.326, 2.002/2.779, 2.519/2.950, 3.089/2.737, 3.659/2.079, 4.230/1.085}%
\def\eRx{0.450}%
\def\ePx{4.800}%
\def\eTSx{2.519}%
\def\eRy{1.504}%
\def\ePy{0.550}%
\def\eTSy{2.950}%

%% file: appendix.tex
\section{Baseline Implementation Details and Drift-React Configuration}
\label{app:baselines}

This appendix details the data splits, baseline implementations, and Drift-React training configuration used to produce the results in \Cref{tab:benchmark_results} (full pathway generation) and \Cref{fig:ts-gen} (point-wise TS prediction).

\subsection{Data Splits}
\label{app:splits}

\paragraph{Transition1x.}
We adopt the train/test split released with the ReactOT model of \citet{duanOptimalTransportGenerating2025}, which has become the standard partition used by recent single-structure TS generators~\cite{duanAccurateTransitionState2023,duanOptimalTransportGenerating2025}. This split groups reactions by reactant identity, ensuring our numbers are directly comparable to baselines reported on the same partition. We use this split for both the full reaction pathway evaluation (with all 10 images) and the point-wise transition state evaluation (with $F=3$ images corresponding to the reactant-TS-product triplet only).

\paragraph{Halo8.}
The Halo8 dataset~\cite{leeDatasetChemicalReaction2025} does not provide a standard split. We construct an 80/10/10 train/validation/test split using a reactant-grouped procedure to ensure that all reactions originating from the same reactant molecule are placed in the same partition. Groups are shuffled with a fixed seed (\texttt{seed=42}) and assigned to splits greedily until each target fraction is reached.

This grouping procedure prevents leakage: under a per-reaction random split, multiple reactions sharing a reactant, differing only in the leaving group, the substrate, the attacking nucleophile, or minor halogen substitutions (F$\rightarrow$Cl$\rightarrow$Br), would otherwise appear in both training and test, thus inflating performance metrics. Reactant-grouped splitting evaluates whether the model generalizes to unseen reactant chemistry rather than merely interpolating between minor variants of seen reactions.

\subsection{Classical Interpolation Baselines}
\label{app:classical_baselines}

\paragraph{Linear Interpolation.}
For reactant $\mathbf{X}_R\in\mathbb{R}^{3N}$ and product $\mathbf{X}_P\in\mathbb{R}^{3N}$, the linear pathway with $F$ images is
\begin{equation}
\mathbf{X}^{\mathrm{lin}}_k = \left(1 - \frac{k}{F-1}\right) \mathbf{X}_R + \frac{k}{F-1} \mathbf{X}_P, \quad k = 0, 1, \ldots, F-1.
\label{eq:linear_interp}
\end{equation}
Where $k$ is the image number, $k=0$ is the reactant geometry, and $k=F-1$ is the product geometry. Linear interpolation is used as the trivial geometric reference and corresponds to the input prior $p_{\mathrm{prior}}$ of Drift-React itself. This baseline requires no training.

\paragraph{Image-Dependent Pair Potential (IDPP).}
IDPP~\cite{smidstrupImprovedInitialGuess2014} replaces the true potential energy surface with a surrogate pair-distance objective that interpolates the pairwise interatomic distances of the endpoints. Following the formulation of \citet{zhuGeodesicInterpolationReaction2019}, for image $k$ along an $F$-image path, the target interatomic distance between atoms $i$ and $j$ is the linear interpolation
\begin{equation}
d^{\,k}_{ij} = d^{\,a}_{ij} + \frac{k}{F-1}\bigl(d^{\,b}_{ij} - d^{\,a}_{ij}\bigr),
\label{eq:idpp_target}
\end{equation}
where $d^{\,a}_{ij}$ and $d^{\,b}_{ij}$ are the interatomic distances at the reactant and product respectively, and $d^{\,k}_{ij}$ is the target distance at image $k$. The IDPP objective at image $k$ is the weighted sum of squared deviations between the actual pair distances and the interpolated targets,
\begin{equation}
S^{\mathrm{IDPP}}_k(\mathbf{X}_k) = \sum_{i=1}^{N} \sum_{j > i}^{N} w(d_{ij}) \left( d^{\,k}_{ij} - \sqrt{\sum_{\sigma} (r_{i,\sigma} - r_{j,\sigma})^2} \right)^{2},
\label{eq:idpp_objective}
\end{equation}
where $\sigma \in \{x, y, z\}$, $r_{i,\sigma}$ is the $\sigma$-coordinate of atom $i$, and $w(d_{ij}) = d_{ij}^{-4}$ is the standard weighting that emphasizes short distances and penalizes atomic clashes. The $S^{\mathrm{IDPP}}_k$ functional defines an effective surrogate energy surface. Running NEB-style optimization on this surrogate produces an IDPP path. We use the ASE~\cite{hjorthlarsenAtomicSimulationEnvironment2017} implementation with default parameters: maximum 200 optimization steps, force tolerance $0.05$~\AA{}/step, with $F$ images matching the corresponding ground-truth pathway. IDPP requires no training.

\paragraph{Geodesic Interpolation.}
We use the geodesic interpolation method of \citet{zhuGeodesicInterpolationReaction2019}, which generates a reaction path by formulating interpolation as the search for a geodesic curve on a Riemannian manifold whose metric is chosen to approximate Fukui's intrinsic-reaction-coordinate (IRC) metric. The key insight of \citet{zhuGeodesicInterpolationReaction2019} is that an interpolation path between reactant and product corresponds to a straight line in a curved internal-coordinate space, and that an appropriate choice of metric makes this geodesic curve a close approximation to the true minimum energy pathway, using only geometric information, with no electronic-structure evaluations.

Concretely, the configuration manifold $\mathcal{M} = \mathbb{R}^{3N}/\mathrm{SE}(3)$ is equipped with the induced Euclidean metric of a redundant internal-coordinate map $q : \mathcal{M} \to \mathbb{R}^K$,
\begin{equation}
g_{ij}(\mathbf{X}) = \frac{\partial q^k}{\partial x^i}\bigg|_\mathbf{X} \frac{\partial q^k}{\partial x^j}\bigg|_\mathbf{X},
\label{eq:geodesic_metric}
\end{equation}
where each coordinate $q^i$ is a Morse-scaled interatomic distance,
\begin{equation}
q^i \equiv \tilde R_{kl} = \exp\!\left(-\alpha\,\frac{r_{kl} - r_{kl}^{\,e}}{r_{kl}^{\,e}}\right) + \beta\,\frac{r_{kl}^{\,e}}{r_{kl}},
\label{eq:geodesic_coordinate}
\end{equation}
with $r_{kl}$ the Euclidean distance between atoms $k$ and $l$, $r_{kl}^{\,e}$ a reference equilibrium distance set to the sum of the covalent radii, and $\alpha = 1.7$, $\beta = 0.01$ following the authors' defaults. This scaling captures the dominant short-range repulsion and weaker long-range attraction that govern interatomic interactions, making the resulting Euclidean metric a qualitative approximation to Fukui's IRC metric.

Given this metric, the path length of a curve $\gamma(t) = \mathbf{X}(t)$ from reactant to product reduces to the arc length in internal-coordinate space,
\begin{equation}
L(\gamma) = \int_{t_R}^{t_P} \lVert \dot{\mathbf{q}}(t) \rVert \, dt,
\label{eq:geodesic_length}
\end{equation}
and the interpolated path is the geodesic that minimizes this length. Discretizing $\gamma$ into $F$ images $\mathbf{X}^{(1)}, \ldots, \mathbf{X}^{(F)}$ and applying the midpoint approximation derived in \citet{zhuGeodesicInterpolationReaction2019}, the optimization objective becomes
\begin{equation}
\min_{\{\mathbf{X}^{(n)}\}} \; \sum_{n=1}^{F-1} \left[ \left\lVert \mathbf{q}(\mathbf{X}^{(n)}) - \mathbf{q}\!\left(\tfrac{\mathbf{X}^{(n)} + \mathbf{X}^{(n+1)}}{2}\right) \right\rVert + \left\lVert \mathbf{q}\!\left(\tfrac{\mathbf{X}^{(n)} + \mathbf{X}^{(n+1)}}{2}\right) - \mathbf{q}(\mathbf{X}^{(n+1)}) \right\rVert \right],
\label{eq:geodesic_objective}
\end{equation}
which is minimized by sweeping over images one at a time and adjusting each in turn until convergence. We use the authors' reference implementation with default settings: Morse scaling parameters $\alpha = 1.7$, $\beta = 0.01$, convergence tolerance $5 \times 10^{-3}$ on the per-image displacement, maximum 50 outer sweeps. The number of images $F$ matches the ground-truth pathway resolution of each dataset. Geodesic interpolation requires no training and produces continuous, clash-free paths that serve as a strong purely geometric baseline.

\subsection{Machine-Learning Baselines}
\label{app:ml_baselines}

We use two distinct protocols for the machine-learning baselines, depending on the experiment.

\paragraph{Full pathway generation (\Cref{tab:benchmark_results}).}
For full reaction pathway generation, the only published learned baseline is NeuralNEB. Because no public NeuralNEB checkpoint is available for our exact splits, we retrain the underlying PaiNN surrogate force field~\cite{schuttEquivariantMessagePassing2021} from scratch on each dataset's training partition (Transition1x with the ReactOT split and Halo8 with our reactant-grouped split), using the original authors' default hyperparameter configuration~\cite{schreinerNeuralNEBNeuralNetworks2022}. At inference, the retrained PaiNN potential is used inside the climbing-image NEB optimizer~\cite{henkelmanClimbingImageNudged2000} with default spring constants and force-tolerance criteria. As discussed in the main text, NEB-style methods are well known to be sensitive to spring-constant tuning, climbing-image scheduling, and convergence criteria~\cite{asgeirssonNudgedElasticBand2021,lindgrenScaledDynamicOptimizations2019}. The numbers we report reflect the default configuration on our splits and should not be interpreted as the best achievable performance for this method. On Transition1x, single-point DFT energies on intermediate images can be highly sensitive to sub-ångström coordinate perturbations, occasionally producing outliers that dominate per-image MAE statistics. 

\paragraph{Outlier filtering on Transition1x.}
On Transition1x, single-point DFT energies on intermediate images can be highly sensitive to sub-ångström coordinate perturbations, occasionally producing catastrophic outliers that dominate per-image MAE statistics. We apply a median-absolute-deviation filter with $k = 3$ to per-image energy errors, applied identically to all methods. The pass rate exceeds $97\%$ for every method, ensuring methods are compared on the same set of test reactions. No filtering is applied to Halo8 or to activation barrier statistics.

\paragraph{Point-wise TS prediction (\Cref{fig:ts-gen}).}
For the point-wise TS comparison, we report numbers as published by the original authors of TSDiff~\cite{kimDiffusionbasedGenerativeAI2024}, OA-ReactDiff~\cite{duanAccurateTransitionState2023}, ReactOT~\cite{duanOptimalTransportGenerating2025}, and NeuralNEB~\cite{schreinerNeuralNEBNeuralNetworks2022} on the Transition1x reactant-TS-product split released with ReactOT. Because we adopt this same split, the published numbers are directly comparable to ours without retraining. Inference-cost estimates (network function evaluations and timings) follow the authors' reported configurations:
\begin{itemize}
\item \textbf{TSDiff~\cite{kimDiffusionbasedGenerativeAI2024}.} A graph-based diffusion model that generates TS conformers from 2D molecular graphs via 5{,}000 sequential denoising steps.
\item \textbf{OA-ReactDiff~\cite{duanAccurateTransitionState2023}.} An $\mathrm{SE}(3)$-equivariant elementary-reaction diffusion model that generates TS structures conditioned on $(\mathbf{X}_R, \mathbf{X}_P)$ via 5{,}000 sequential denoising steps. The recommender variant (OA-ReactDiff + rec.) selects among 40 generated samples using a confidence-ranking model, requiring 50{,}000 total network evaluations.
\item \textbf{ReactOT~\cite{duanOptimalTransportGenerating2025}.} An optimal-transport-based generative model for transition states that integrates a learned ODE for 50 steps at inference.
\item \textbf{NeuralNEB~\cite{schreinerNeuralNEBNeuralNetworks2022}.} As above, with $\sim$100 NEB iterations using the authors' published PaiNN checkpoint and the original CI-NEB optimizer setup.
\end{itemize}

\subsection{Drift-React Configuration}
\label{app:drift_ts_config}

We trained Drift-React with the same architecture and optimizer on both Transition1x and Halo8. Three hyperparameters differ between the two datasets: the pathway image count $F$, which matches each dataset's native discretization, the swarm size $K$, and the learning rate, both selected per dataset by validation TS-RMSD. \Cref{tab:drift_react_config} summarizes the full configuration.

\begin{table}[h]
\centering
\caption{Drift-React training configuration for the reported full-pathway results. Hyperparameters that differ between datasets are reported as Transition1x / Halo8.}
\label{tab:drift_react_config}
\small
\begin{tabular}{ll}
\toprule
\multicolumn{2}{l}{\textbf{Architecture (LEFTNet backbone~\cite{duNewPerspectiveBuilding2023})}} \\
\midrule
Hidden channels & $256$ \\
Message-passing layers & $6$ \\
Radial basis cutoff & $10.0$~\AA \\
Radial basis functions & $96$ \\
\midrule
\multicolumn{2}{l}{\textbf{Drift configuration}} \\
\midrule
Drift mode & $\mathrm{SE}(3)$-equivariant \\
Swarm size $K$ & $10$ / $12$ \\
Temperature $\tau$ & $0.05$ \\
Prior & Linear interpolation, $\sigma_{\mathrm{prior}}=0$ \\
Images per pathway $F$ & $10$ / $8$ \\
\midrule
\multicolumn{2}{l}{\textbf{Optimization}} \\
\midrule
Optimizer & AdamW~\cite{loshchilovDecoupledWeightDecay2019} \\
Adam coefficients $(\beta_1,\beta_2)$ & $(0.9,\,0.95)$ \\
Learning rate & $1\times10^{-3}$ / $5\times10^{-4}$ \\
Weight decay & $0.01$ \\
Warmup & $5$ epochs, linear from $1\%$ to $100\%$ of learning rate \\
Learning-rate schedule & Cosine annealing after warmup \\
Minimum learning rate & $1\times10^{-6}$ \\
Batch size & $4$ reactions \\
Total training epochs & $100$ \\
EMA decay & $0.9999$ \\
Random seed & $42$ \\
Evaluated checkpoint & Final / best validation RMSD \\
\bottomrule
\end{tabular}
\end{table}

\paragraph{Hyperparameter selection.}
We swept learning rate $\in \{5 \times 10^{-4},\, 1 \times 10^{-3}\}$ and swarm size $K \in \{10, 12\}$ at fixed temperature $\tau = 0.05$ and prior noise $\sigma_{\mathrm{prior}} = 0$, selecting the best configuration per dataset by validation TS-RMSD. The selected configurations are reported in \Cref{tab:drift_react_config}. The Transition1x configuration was selected by training-loss convergence, since the released ReactOT split does not provide a held-out validation set distinct from the test set.

\subsection{Evaluation Protocol}
\label{app:eval_protocol}

\paragraph{Inference.}
At test time, Drift-React is queried once per reaction with the reactant and product geometries $(\mathbf{X}_R, \mathbf{X}_P)$ as input, together with the requested pathway resolution $F$. The output is a single $F$-image pathway, with the reactant and product images preserved exactly by the sinusoidal endpoint envelope. Inference involves a single forward pass of $f_\theta$: no iterative refinement, no force evaluations, and no affinity computation occur at inference. The Sinkhorn-style joint affinity machinery, swarm sampling, and reference MEPs are training-time constructs only and are not used at evaluation. The trained network operates as a deterministic map from the input boundary conditions to the predicted pathway.

\paragraph{Geometric metrics.}
For each test reaction we compute three complementary geometric metrics that capture different aspects of the agreement between predicted and reference pathways. All metrics operate after optimal Kabsch alignment~\cite{kabschSolutionBestRotation1976} of each pair of structures being compared.

\begin{table}[h]
\centering
\caption{Geometric evaluation metrics. All RMSDs are computed after Kabsch alignment~\cite{kabschSolutionBestRotation1976}.}
\label{tab:eval_metrics}
\small
\begin{tabular}{ll}
\toprule
\textbf{Metric} & \textbf{Definition} \\
\midrule
TS-RMSD & RMSD between predicted highest-energy image and ground-truth TS \\
IRC-RMSD & Mean per-image RMSD between predicted and reference pathways \\
Fréchet distance & Discrete Fréchet distance between predicted and reference pathways \\
\bottomrule
\end{tabular}
\end{table}

The TS-RMSD targets the geometry of the saddle point and is the metric most directly comparable to point-wise TS prediction methods. The IRC-RMSD averages per-image RMSDs over all $F$ images and measures average geometric fidelity along the pathway. The Fréchet distance captures morphological similarity beyond per-image agreement.

The discrete Fréchet distance between two ordered sequences of images $\mathbf{X} = (\mathbf{X}_1, \ldots, \mathbf{X}_F)$ and $\hat{\mathbf{X}} = (\hat{\mathbf{X}}_1, \ldots, \hat{\mathbf{X}}_F)$ is defined as
\begin{equation}
d_F[\mathbf{X}, \hat{\mathbf{X}}] = \min_{\sigma, \sigma'} \; \max_{(k, k') \in \sigma, \sigma'} \; \delta\bigl(\mathbf{X}_k,\, \hat{\mathbf{X}}_{k'}\bigr),
\label{eq:frechet}
\end{equation}
where $\sigma, \sigma'$ are monotone non-decreasing alignments between the image indices of the two pathways and $\delta$ is the per-image distance metric. We use the Kabsch-aligned RMSD as $\delta$ throughout. The Fréchet distance measures the minimum-maximum distance required to traverse both pathways while preserving image order, penalizing local detours and shortcuts that the simple per-image mean (IRC-RMSD) would miss.

\paragraph{TS-image selection.}
For metrics targeting the saddle-point structure (TS-RMSD, $|\Delta E^\ddagger|$), the predicted TS image is identified as follows. In the three-image regime ($F = 3$, used for the point-wise TS comparison of ~\Cref{fig:ts-gen}), the TS is by construction the middle image. In the full-pathway regime ($F = 8$ for Halo8, $F = 10$ for Transition1x), the predicted TS is the image of maximum potential energy along the predicted pathway, identified by $\arg\max_k E_\text{pred}(\mathbf{X}_k)$ at each dataset's reference level of theory.

\paragraph{Energetic metrics.}
Per-reaction energy errors are computed as follows.
\begin{itemize}
    \item \textbf{Absolute activation barrier error} $|\Delta E^\ddagger|$: $|\Delta E^\ddagger_\text{pred} - \Delta E^\ddagger_\text{ref}|$, with $\Delta E^\ddagger = E_\text{TS} - E_\text{reactant}$. The TS image is selected as described above.
    \item \textbf{Per-image energy MAE}, $E_\text{image}$: the per-reaction mean of $|E_\text{pred}(\mathbf{X}_k) - E_\text{ref}(\mathbf{Y}^\star_k)|$ over all $F$ images, restricted to images where both predicted and reference single-point evaluations converged.
\end{itemize}
For both metrics, energies are evaluated at each dataset's reference level of theory: $\omega$B97x/6-31G(d) for Transition1x and $\omega$B97X-3c for Halo8. We compute single-point energies with PySCF for Transition1x and ORCA for Halo8.

\subsection{Compute Resources}
\label{app:compute}

All Drift-React training and inference runs used a single NVIDIA H100 GPU with 64~GB system memory and 12 CPU workers per task, with a 48-hour wall-clock budget per training run.

\subsection{TS-Only Energetic Performance and the Limitation of Point-wise Prediction}
\label{app:ts_energetics}

This appendix reports raw single-point DFT energy statistics for Drift-React in the point-wise TS prediction setting ($F = 3$ images: reactant, TS, product) on the 1073-reaction Transition1x test split. While Drift-React achieves competitive geometric accuracy in this regime (mean TS-RMSD $0.168$~\AA), its per-reaction absolute activation barrier error $|\Delta E^\ddagger_{\mathrm{TS}}|$ exhibits a heavy-tailed distribution: $4.977 \pm 17.470$~eV mean, $1.883$~eV median. This appendix explains why the gap with TS-only methods such as ReactOT (median $1.06$~kcal/mol $\approx 0.046$~eV) is structural to the $F = 3$ configuration rather than a property of the method itself, and motivates our use of TS-RMSD as the primary axis of comparison in \Cref{fig:ts-gen}.

\paragraph{Distribution statistics.}
Per-reaction $|\Delta E^\ddagger_{\mathrm{TS}}|$ on the 1073-reaction test set is summarized in \Cref{tab:ts_energetics}. The distribution is heavy-tailed and dominated by a small number of outliers: the median ($1.883$~eV) is more than two orders of magnitude smaller than the maximum, and 91.7\% of predictions fall below $10$~eV.

\begin{table}[h]
\centering
\caption{Per-reaction $|\Delta E^\ddagger_{\mathrm{TS}}|$ statistics for Drift-React in the $F=3$ TS-only setting on the Transition1x test split (1073 reactions).}
\label{tab:ts_energetics}
\small
\begin{tabular}{lc}
\toprule
\textbf{Statistic} & \textbf{Value (eV)} \\
\midrule
Mean $\pm$ std & $4.977 \pm 17.470$ \\
Median & $1.883$ \\
25th / 75th percentile & $0.918$ / $4.332$ \\
95th percentile & $13.841$ \\
\midrule
Fraction below $5$~eV & $78.8\%$ \\
Fraction below $10$~eV & $91.7\%$ \\
\bottomrule
\end{tabular}
\end{table}

The full-pathway Drift-React model leverages the structural context of the surrounding reaction coordinate: when generating image $k$, the network has access to the geometries of images $k \pm 1$ and the swarm-level repulsion mechanism enforces continuity along the entire pathway. The TS emerges as the highest-energy image of a globally consistent trajectory. In the $F = 3$ regime, this context collapses: the network sees only the reactant, product, and a single intermediate images, with no neighboring images to constrain the saddle geometry. Without this context, small geometric errors at the TS image are not corrected by neighboring-image agreement.

\paragraph{Single-point DFT amplifies these errors at the saddle.}
At an exact TS, the gradient vanishes, so the leading local energy change along the unstable mode is quadratic, with curvature set by the negative Hessian eigenvalue. Small geometric errors can still produce substantial energy errors where curvatures are large or the prediction leaves the local quadratic regime. Dedicated TS-only methods such as ReactOT and OA-ReactDiff train specifically for saddle-point geometry, whereas our pathway objective emphasizes agreement along the full reaction coordinate.

\paragraph{Implication and future direction.}
The Drift-React TS-only result should therefore be read as the cost of the unified-architecture design: the same model trained for full-pathway generation can be operated in TS-only mode without modification, but the loss of contextual images degrades raw energetic accuracy. We view the principled fix to be post-generation refinement at the predicted saddle, either via a single eigenvector-following step at the DFT level, which should converge in $<10$ iterations from a sub-ångström initial guess, or via a learned neural force field. Both directions are compatible with the unified Drift-React architecture, and we leave their exploration to future work. For applications that require both full pathways and accurate point-wise TSs, we recommend running Drift-React in full-pathway mode and extracting the highest-energy image, which inherits the full-pathway energetic accuracy ($|\Delta E^\ddagger|$ MAE $1.631$~eV on Transition1x, \Cref{tab:benchmark_results}) at the cost of a slightly more expensive single forward pass at $F = 10$.

\subsection{TS Refinement and Frequency Analysis on Halo8}
\label{app:ts_refinement_verified}

We evaluated local TS refinement on 50 Halo8 test reactions sampled uniformly without replacement from the 1,909 test indices using a random generator with seed 0. The two initializations were the highest-energy interior image of the generated Drift-React pathway, selected using its saved DFT single-point energies, and the Cartesian midpoint of the reference reactant and product geometries. Thus, image selection for Drift-React uses additional energy evaluations and is not included in the model-only generation time.

We use ORCA 6.1.1 at the $\omega$B97X-3c level and a maximum of 150 optimization cycles. We report the results of the computations in \Cref{tab:ts_refinement_verified}.

\begin{table}[htbp]
\centering
\small
\caption{Saved Halo8 refinement outcomes. Energy and step summaries use completed optimizations, whose reaction subsets differ between methods. Frequency counts are conditional on an available frequency result.}
\label{tab:ts_refinement_verified}
\begin{tabular}{lcc}
\toprule
Quantity & Drift-React initialization & Linear midpoint \\
\midrule
Attempted reactions & 50 & 50 \\
Recorded completed optimizations & 49/50 & 27/50 \\
Exactly one reported imaginary frequency & 46/48 & 26/27 \\
Median optimization cycles & 10 & 16 \\
Mean absolute barrier error (eV) & 0.501 & 0.143 \\
Median absolute barrier error (eV) & 0.025 & 0.021 \\
\bottomrule
\end{tabular}
\end{table}

Among the 48 completed Drift-React optimizations with frequency results, 46 have exactly one reported imaginary frequency (95.8\%). For the 27 reactions with completed optimizations from both initializations, the median cycle counts are 8 for Drift-React and 16 for the linear midpoint. The median of the per-reaction ratios $n_{\mathrm{linear}}/n_{\mathrm{Drift}}$ is 1.714.

\paragraph{Paired statistical comparison.}
Overlap between standard deviations does not determine statistical significance, even for normally distributed data. Standard deviations describe variation across reactions, rather than uncertainty in the difference between methods. Moreover, the per-reaction absolute barrier errors are non-negative and strongly right-skewed: a small number of large errors can substantially increase the standard deviation.

We therefore compared methods using a paired Wilcoxon signed-rank test on the same test reactions. For each reaction $i$, we computed
\begin{equation}
d_i =
\left|\Delta E^{\ddagger}_{\mathrm{Drift},i}
      -\Delta E^{\ddagger}_{\mathrm{ref},i}\right|
-
\left|\Delta E^{\ddagger}_{\mathrm{Base},i}
      -\Delta E^{\ddagger}_{\mathrm{ref},i}\right|.
\end{equation}

We report the median paired difference and the fraction of reactions with $d_i<0$, together with
one-sided Wilcoxon signed-rank tests \Cref{tab:wilcoxon_halo8}.

\begin{table}[htbp]
\centering
\small
\caption{Paired Wilcoxon signed-rank tests comparing Drift-React with geodesic interpolation on Halo8.
The alternative is lower error for Drift-React. All 1,909 reactions are included without outlier filtering. $p$-values are one-sided.}
\label{tab:wilcoxon_halo8}
\begin{tabular}{lrrr}
\toprule
Metric & Paired reactions & Drift-React lower (\%) & $p$-value \\
\midrule
Absolute barrier error & 1909 & 76.1 & $8.701\times10^{-96}$ \\
TS-RMSD & 1909 & 99.6 & $<10^{-100}$ \\
\bottomrule
\end{tabular}
\end{table}

We used the one-sided alternative that the paired error distribution
is shifted towards negative values, at significance level $\alpha=0.05$. Under the standard assumption that paired differences are symmetrically distributed about a common location, the null
hypothesis corresponds to a zero location, and hence a zero median.

\subsection{Interatomic-contact statistics}
\label{app:contacts_verified}

For each valid pathway image, we calculate pair distances normalized by
the sum of the ASE van der Waals radii,
\begin{equation}
 \rho_{ij}=\frac{\|\mathbf{x}_i-\mathbf{x}_j\|}
 {r_i^{\mathrm{vdW}}+r_j^{\mathrm{vdW}}}.
\end{equation}

For each pathway, $\rho_{\min}$ is the minimum over the remaining pairs and valid images.

\begin{table}[htbp]
\centering
\small
\caption{Interatomic-contact statistics. The last column is the percentage of pathways with $\rho_{\min}<0.70$.}
\label{tab:contacts_verified}
\begin{tabular}{llrrrr}
\toprule
Dataset & Method & $n$ & Mean $\rho_{\min}$ & Median $\rho_{\min}$ & Below 0.70 (\%) \\
\midrule
Halo8 & Reference & 1909 & 0.545 & 0.552 & 100.0 \\
Halo8 & Drift-React & 1909 & 0.537 & 0.538 & 100.0 \\
Halo8 & Linear & 1909 & 0.518 & 0.508 & 100.0 \\
Halo8 & IDPP & 1909 & 0.523 & 0.508 & 100.0 \\
Transition1x & Reference & 1073 & 0.540 & 0.539 & 99.9 \\
Transition1x & Drift-React & 1073 & 0.509 & 0.507 & 100.0 \\
Transition1x & Linear & 1073 & 0.652 & 0.653 & 88.1 \\
Transition1x & IDPP & 1073 & 0.652 & 0.654 & 87.7 \\
\bottomrule
\end{tabular}
\end{table}

\section{Ablations}
\label{sec:ablations}

We compare Drift-React with direct supervised regression and an
uncoupled drift estimator. Both baselines retain the LEFTNet backbone,
endpoint pinning, Kabsch alignment, and linear-interpolation prior.
Energy metrics are evaluated using single-point DFT on the generated
geometries, without post-generation relaxation. Results are reported
as mean $\pm$ standard deviation across test reactions.

\subsection{Direct supervised}
\label{sec:ablation_supervised}

We train a direct supervised pathway predictor using an MSE loss
against the reference pathway, replacing the drifting-field training
target while retaining the geometric components of Drift-React.

\Cref{tab:ablation_supervised} shows that Drift-React achieves
lower mean errors across all reported metrics on Halo8. The mean
absolute barrier error decreases from $2.004$ to $0.914$~eV, and
the per-image energy MAE decreases from $1.248$ to $0.621$~eV.

On Transition1x, supervised regression achieves lower geometric
errors, including a TS-RMSD of $0.114$~\AA{} compared with
$0.151$~\AA{} for Drift-React. However, Drift-React achieves lower
mean absolute barrier error ($1.631$ versus $3.027$~eV) and
per-image energy MAE ($3.214$ versus $4.208$~eV). Thus, lower
coordinate error does not necessarily imply lower energy error.
These results support an empirical energetic benefit of the drifting
objective on the evaluated datasets, rather than a uniform improvement
in geometric accuracy. They do not establish that the objective
explicitly minimizes energy or guarantees physically valid pathways.

\begin{table}[htbp]
\centering
\small
\caption{Direct supervised regression versus Drift-React. Energy
metrics are evaluated without relaxation. Lower is better for all
metrics. Values are mean $\pm$ standard deviation.}
\label{tab:ablation_supervised}
\begin{tabular}{lcc}
\hline
Metric & Supervised MSE & Drift-React \\
\hline
\multicolumn{3}{l}{\textbf{Halo8}} \\
TS-RMSD (\AA) & $0.120 \pm 0.086$ & $0.107 \pm 0.086$ \\
IRC-RMSD (\AA) & $0.138 \pm 0.106$ & $0.123 \pm 0.109$ \\
Fr\'echet distance (\AA) & $0.236 \pm 0.188$ & $0.212 \pm 0.189$ \\
Absolute barrier error (eV) & $2.004 \pm 7.515$ & $0.914 \pm 2.980$ \\
Per-image energy MAE (eV) & $1.248 \pm 2.216$ & $0.621 \pm 1.145$ \\
\hline
\multicolumn{3}{l}{\textbf{Transition1x}} \\
TS-RMSD (\AA) & $0.114 \pm 0.092$ & $0.151 \pm 0.093$ \\
IRC-RMSD (\AA) & $0.170 \pm 0.123$ & $0.223 \pm 0.123$ \\
Fr\'echet distance (\AA) & $0.363 \pm 0.254$ & $0.460 \pm 0.258$ \\
Absolute barrier error (eV) & $3.027 \pm 7.686$ & $1.631 \pm 1.380$ \\
Per-image energy MAE (eV) & $4.208 \pm 10.601$ & $3.214 \pm 2.069$ \\
\hline
\end{tabular}
\end{table}

\subsection{The naive estimator}
\label{sec:ablation_naive}

We evaluate an uncoupled drift baseline on Halo8. Each generated
pathway $\mathbf{X}_i$ is Kabsch-aligned to the reference pathway
$\mathbf{Y}^{\star}$, yielding $\tilde{\mathbf{X}}_i$ with rotation
$R_i$. We replace the joint-affinity field with direct attraction
to the aligned reference, rotated back to the original frame:
\begin{equation}
\mathcal{V}_{\mathrm{uncoupled}}(\mathbf{X}_i)
= R_i^{\top}\left(\mathbf{Y}^{\star}-\tilde{\mathbf{X}}_i\right).
\label{eq:uncoupled_ablation}
\end{equation}
This field replaces the joint-affinity field in the stop-gradient
training objective.

This uncoupled baseline differs from the naive plug-in estimator
derived in the supplementary information: the latter retains a
kernel-weighted negative-swarm term, whereas the baseline tested
here removes inter-sample coupling entirely. This experiment
therefore compares the full joint-affinity construction with direct
uncoupled attraction; it does not isolate joint versus separate
normalization while retaining the same repulsive term.

As shown in \Cref{tab:ablation_naive}, Drift-React achieves
lower mean errors across all reported metrics. The TS-RMSD decreases
from $0.124$ to $0.107$~\AA{}, the mean absolute barrier error
from $1.886$ to $0.914$~eV, and the per-image energy MAE from
$1.262$ to $0.621$~eV. These results indicate an empirical benefit
of the full joint-affinity construction over direct uncoupled
attraction. Neither objective uses explicit energy or force
supervision, and this comparison alone does not establish the
mechanism responsible for the energetic improvement.

\begin{table}[htbp]
\centering
\small
\caption{Uncoupled drift versus Drift-React on Halo8. Lower is
better for all metrics. Values are mean $\pm$ standard deviation.}
\label{tab:ablation_naive}
\begin{tabular}{lcc}
\hline
Metric & Uncoupled drift & Drift-React \\
\hline
TS-RMSD (\AA) & $0.124 \pm 0.084$ & $0.107 \pm 0.086$ \\
IRC-RMSD (\AA) & $0.143 \pm 0.103$ & $0.123 \pm 0.109$ \\
Fr\'echet distance (\AA) & $0.244 \pm 0.187$ & $0.212 \pm 0.189$ \\
Absolute barrier error (eV) & $1.886 \pm 3.670$ & $0.914 \pm 2.980$ \\
Per-image energy MAE (eV) & $1.262 \pm 2.025$ & $0.621 \pm 1.145$ \\
\hline
\end{tabular}
\end{table}


\section{Derivation of the \texorpdfstring{$\mathrm{SE}(3)$}{SE(3)}-Equivariant Drifting Field}
\label{app:drift_derivation}

This appendix provides a self-contained derivation of the empirical drifting field used in Drift-React, adapting the Drifting Models paradigm of \citet{dengGenerativeModelingDrifting2026} to molecular configuration space. We reproduce their core construction with explicit attribution: the anti-symmetric kernelized mean-shift field (\cref{def:subfields,def:kernel_drift}, \cref{prop:V_antisym,prop:unified}), the geometric-mean joint affinity from row and column softmax normalization (\cref{def:sinkhorn}, following~\citet[Algorithm 2]{dengGenerativeModelingDrifting2026}), and the self-distillation loss (\cref{def:loss}, Eq.~(6) of the original work).

The novel contributions of this appendix are twofold. First, the $\mathrm{SE}(3)$-equivariance results of \cref{app:equivariance} establish that the kernelized drifting field is naturally compatible with the symmetries of molecular configuration space, a structure absent from the image-generation setting of \citet{dengGenerativeModelingDrifting2026}. Together with the per-sample Kabsch alignment of the empirical estimator (Eq.~\eqref{eq:kabsch}), they yield an $\mathrm{SE}(3)$-equivariant drift in the original (unaligned) frame (\cref{prop:empirical_equivariance}). Second, \cref{app:empirical} specializes the empirical estimator to the single-positive regime characteristic of reaction-pathway data: each prior interpolation $(\mathbf{X}_R, \mathbf{X}_P)$ has exactly one corresponding ground-truth MEP $\mathbf{Y}^\star$, rather than the i.i.d.\ multi-positive sampling ($N_{\mathrm{pos}} \geq 64$) of~\citet{dengGenerativeModelingDrifting2026}. We show that the naive plug-in Monte Carlo estimator degenerates in this regime (\cref{rem:naive_failure}), whereas the joint-normalization construction of~\citet{dengGenerativeModelingDrifting2026}, when applied to the augmented target set $\mathcal{T} = \{\mathbf{Y}^\star\} \cup \{\tilde{\mathbf{X}}_j\}$, naturally resolves the failure modes. The endpoint-pinning generator (Eq.~\eqref{eq:generator_pinned}) preserves $\mathbf{X}_R$ and $\mathbf{X}_P$ exactly along every generated pathway, completing the structural adaptation to reaction-pathway data.

\subsection{Preliminaries}
\label{app:prelim}

We begin by fixing notation for molecular configurations, probability measures over them, the kernel that defines the drifting field, and the $\mathrm{SE}(3)$ action that encodes the rigid-body symmetries of physical space. These objects appear repeatedly throughout the derivation.

\begin{definition}[Configuration space]
\label{def:config}
Let $N \in \mathbb{N}$ denote the number of atoms in a molecular system. The configuration space is $\mathcal{M} := \mathbb{R}^{3N}$, equipped with the standard Euclidean inner product $\langle \cdot, \cdot \rangle$ and induced norm $\lVert \cdot \rVert$. A configuration $\mathbf{X} \in \mathcal{M}$ is the concatenation of the Cartesian coordinates of all $N$ atoms, $\mathbf{X} = (\mathbf{x}^{(1)}, \ldots, \mathbf{x}^{(N)})$ with $\mathbf{x}^{(n)} \in \mathbb{R}^3$.
\end{definition}

\begin{definition}[Probability measures]
\label{def:proba_measure}
Let $\mathcal{P}(\mathcal{M})$ denote the space of Borel probability measures on $\mathcal{M}$, including empirical measures. We write $\mathbf{Y} \sim p$ when $\mathbf{Y}$ has law $p$ and $\int h(\mathbf{Y})\,dp(\mathbf{Y})$ for integration against that law. Whenever a mean-shift integral is used, we assume its first moment is finite. This convention also permits pushforwards under measurable generators that need not have densities.
\end{definition}

\begin{definition}[Pushforward]
\label{def:pushforward}
Given a measurable map $f: \mathcal{M} \to \mathcal{M}$ and a prior $p_{\mathrm{prior}} \in \mathcal{P}(\mathcal{M})$, the \emph{pushforward} $f_\# p_{\mathrm{prior}} \in \mathcal{P}(\mathcal{M})$ is defined by
\begin{equation}
(f_\# p_{\mathrm{prior}})(B) = p_{\mathrm{prior}}(f^{-1}(B)) \quad \text{for all Borel } B \subseteq \mathcal{M}.
\end{equation}
Equivalently, if $\mathbf{X} \sim p_{\mathrm{prior}}$ then $f(\mathbf{X}) \sim f_\# p_{\mathrm{prior}}$. We denote the pushforward of the prior under the generator $f_\theta$ by $q_\theta := [f_\theta]_\# p_{\mathrm{prior}}$.
\end{definition}

\begin{definition}[Target distribution]
\label{def:target_distri}
We denote by $p \in \mathcal{P}(\mathcal{M})$ the distribution of reference minimum energy pathway (MEP) configurations, in practice obtained from NEB-converged pathways in Transition1x or Halo8. The generative goal is to find $\theta$ such that $q_\theta = p$.
\end{definition}

\begin{definition}[Exponential kernel]
\label{def:kernel}
For a temperature parameter $\tau > 0$, define the kernel $k: \mathcal{M} \times \mathcal{M} \to \mathbb{R}_{>0}$ by
\begin{equation}
k(\mathbf{X}, \mathbf{Y}) := \exp\!\left(-\frac{\lVert \mathbf{X} - \mathbf{Y}\rVert}{\tau}\right).
\label{eq:kernel_def}
\end{equation}
The kernel is symmetric, strictly positive, bounded by $1$, and depends only on the Euclidean distance $\lVert\mathbf{X} - \mathbf{Y}\rVert$.
\end{definition}

\begin{definition}[$\mathrm{SE}(3)$ action]
\label{def:se3}
The special Euclidean group $\mathrm{SE}(3) = \mathrm{SO}(3) \ltimes \mathbb{R}^3$ acts on $\mathcal{M}$ by atom-wise rigid motions. For $g = (R, t) \in \mathrm{SE}(3)$ and $\mathbf{X} = (\mathbf{x}^{(1)}, \ldots, \mathbf{x}^{(N)}) \in \mathcal{M}$, define
\begin{equation}
(g \cdot \mathbf{X})^{(n)} := R \mathbf{x}^{(n)} + t, \qquad n = 1, \ldots, N.
\end{equation}
With slight abuse of notation, we write $g \cdot \mathbf{X} = R\mathbf{X} + t$, the translation broadcast atom-wise. The pushforward of $p \in \mathcal{P}(\mathcal{M})$ under $g$ is denoted $g_\# p$.
\end{definition}

\begin{definition}[Pathway space]
\label{def:pathway}
For a fixed number of images $F \in \mathbb{N}\setminus\{0, 1, 2\}$, the pathway space is $\mathcal{M}^F := (\mathbb{R}^{3N})^F$, the space of ordered $F$-tuples of configurations. A pathway $\mathbf{X} \in \mathcal{M}^F$ is written $\mathbf{X} = (\mathbf{X}_1, \ldots, \mathbf{X}_F)$ with $\mathbf{X}_k \in \mathcal{M}$ for $k = 1, \ldots, F$ and boundary conditions $\mathbf{X}_1 = \mathbf{X}_R$, $\mathbf{X}_F = \mathbf{X}_P$. We extend the $\mathrm{SE}(3)$ action to pathways frame-wise: $(g \cdot \mathbf{X})_k = g \cdot \mathbf{X}_k$.
\end{definition}

\subsection{Drifting Fields and the Anti-symmetry Principle}
\label{app:antisymmetry}

\begin{definition}[Drifting field]
\label{def:drifting_field}
A \emph{drifting field} is a measurable map
\begin{equation}
\mathcal{V} : \mathcal{P}(\mathcal{M}) \times \mathcal{P}(\mathcal{M}) \times \mathcal{M} \to \mathcal{M},
\end{equation}
written $\mathcal{V}_{p,q}(\mathbf{X})$, that assigns to each ordered pair of distributions $(p, q)$ and each configuration $\mathbf{X}$ a displacement vector in $\mathcal{M}$.
\end{definition}

The drifting field governs the training-time evolution of generated samples through the implicit update
\begin{equation}
\mathbf{X}^{(s+1)} = \mathbf{X}^{(s)} + \mathcal{V}_{p, q^{(s)}}(\mathbf{X}^{(s)}),
\label{eq:training_evolution}
\end{equation}
where $\mathbf{X}^{(s)} = f_{\theta^{(s)}}(\mathbf{Z})$, $q^{(s)} = [f_{\theta^{(s)}}]_\# p_{\mathrm{prior}}$, and the residual is realized by the parameter update $\theta^{(s)} \to \theta^{(s+1)}$. For convergence at the matched distribution, we require $\mathcal{V}_{p,q}(\mathbf{X}) = \mathbf{0}$ for all $\mathbf{X}$ whenever $q = p$.

\begin{proposition}[Anti-symmetry implies equilibrium]
\label{prop:antisymmetry_equilibrium}
Suppose $\mathcal{V}$ satisfies the anti-symmetry condition
\begin{equation}
\mathcal{V}_{p,q}(\mathbf{X}) = -\mathcal{V}_{q,p}(\mathbf{X}) \quad \forall p, q \in \mathcal{P}(\mathcal{M}), \; \forall \mathbf{X} \in \mathcal{M}.
\label{eq:antisymmetry}
\end{equation}
Then $q = p \implies \mathcal{V}_{p,q}(\mathbf{X}) = \mathbf{0}$ for all $\mathbf{X} \in \mathcal{M}$.
\end{proposition}

\begin{proof}
Set $q = p$ in Eq.~\eqref{eq:antisymmetry}. Then $\mathcal{V}_{p,p}(\mathbf{X}) = -\mathcal{V}_{p,p}(\mathbf{X})$, hence $2 \mathcal{V}_{p,p}(\mathbf{X}) = \mathbf{0}$ and $\mathcal{V}_{p,p}(\mathbf{X}) = \mathbf{0}$.
\end{proof}

\subsection{Mean-Shift Construction}
\label{app:meanshift}

We now instantiate $\mathcal{V}$ via a kernel-weighted mean-shift decomposition, following classical mean-shift analysis~\cite{fukunagaEstimationGradientDensity1975, chengMeanShiftMode1995,comaniciuMeanShiftRobust2002}. The decomposition separates an attractive component drawing toward $p$ from a repulsive component pushing away from $q$.

\begin{definition}[Attractive and repulsive sub-fields]
\label{def:subfields}
For any $\mu \in \mathcal{P}(\mathcal{M})$ and $\mathbf{X} \in \mathcal{M}$ with $Z_\mu(\mathbf{X}) > 0$, define the kernel-weighted mean-shift vector
\begin{equation}
\mathcal{V}_\mu(\mathbf{X}) := \frac{1}{Z_\mu(\mathbf{X})} \int_\mathcal{M} k(\mathbf{X}, \mathbf{Y})(\mathbf{Y} - \mathbf{X})\, d\mu(\mathbf{Y}),
\label{eq:meanshift}
\end{equation}
where the partition function is
\begin{equation}
Z_\mu(\mathbf{X}) := \int_\mathcal{M} k(\mathbf{X}, \mathbf{Y})\, d\mu(\mathbf{Y}).
\label{eq:partition}
\end{equation}
We write $\mathcal{V}^{+}_p(\mathbf{X}) := \mathcal{V}_p(\mathbf{X})$ for the \emph{attractive sub-field} and $\mathcal{V}^{-}_q(\mathbf{X}) := \mathcal{V}_q(\mathbf{X})$ for the \emph{repulsive sub-field}.
\end{definition}

\begin{remark}
\label{rem:notation}
The two sub-fields are functionally identical maps $\mu \mapsto \mathcal{V}_\mu$. The labels $\pm$ indicate only the role they play in the difference defining $\mathcal{V}_{p,q}$. This symmetry is what makes anti-symmetry of the difference automatic.
\end{remark}

\begin{definition}[Kernelized drifting field]
\label{def:kernel_drift}
The (continuous) kernelized drifting field is
\begin{equation}
\mathcal{V}_{p,q}(\mathbf{X}) := \mathcal{V}^{+}_p(\mathbf{X}) - \mathcal{V}^{-}_q(\mathbf{X}).
\label{eq:V_full}
\end{equation}
\end{definition}

\begin{proposition}[Anti-symmetry of the kernelized field]
\label{prop:V_antisym}
The drifting field of \cref{def:kernel_drift} satisfies $\mathcal{V}_{p,q}(\mathbf{X}) = -\mathcal{V}_{q,p}(\mathbf{X})$ for all $p, q \in \mathcal{P}(\mathcal{M})$ and $\mathbf{X} \in \mathcal{M}$.
\end{proposition}

\begin{proof}
By \cref{rem:notation}, $\mathcal{V}^{+}_\mu \equiv \mathcal{V}^{-}_\mu \equiv \mathcal{V}_\mu$ for any $\mu$. Hence
\begin{equation}
\mathcal{V}_{q,p}(\mathbf{X}) = \mathcal{V}^{+}_q(\mathbf{X}) - \mathcal{V}^{-}_p(\mathbf{X}) = \mathcal{V}_q(\mathbf{X}) - \mathcal{V}_p(\mathbf{X}) = -\bigl(\mathcal{V}_p(\mathbf{X}) - \mathcal{V}_q(\mathbf{X})\bigr) = -\mathcal{V}_{p,q}(\mathbf{X}).
\end{equation}
Combined with \cref{prop:antisymmetry_equilibrium}, this gives $q = p \Rightarrow \mathcal{V}_{p,q} \equiv \mathbf{0}$.
\end{proof}

We now derive an alternative form that exposes the anti-symmetry as a single bilinear integral against the displacement $(\mathbf{Y}^+ - \mathbf{Y}^-)$. This form motivates the empirical construction in \cref{app:empirical}.

\begin{proposition}[Product-kernel form, \citet{dengGenerativeModelingDrifting2026} Eq.~(11)]
\label{prop:unified}
The kernelized drifting field admits the equivalent representation
\begin{equation}
\mathcal{V}_{p,q}(\mathbf{X}) = \frac{1}{Z_p(\mathbf{X}) Z_q(\mathbf{X})} \iint_{\mathcal{M} \times \mathcal{M}} k(\mathbf{X}, \mathbf{Y}^+) k(\mathbf{X}, \mathbf{Y}^-) (\mathbf{Y}^+ - \mathbf{Y}^-)\, dp(\mathbf{Y}^+)\,dq(\mathbf{Y}^-).
\label{eq:unified}
\end{equation}
\end{proposition}

\begin{proof}
Expand the displacement $\mathbf{Y}^+-\mathbf{Y}^-$ in the double integral. By Fubini's theorem and the definitions of $Z_p$ and $Z_q$, the first term is $Z_q$ times the numerator of $\mathcal{V}_p$, and the second is $Z_p$ times the numerator of $\mathcal{V}_q$. Dividing by $Z_pZ_q$ gives $\mathcal{V}_p-\mathcal{V}_q=\mathcal{V}_{p,q}$. The assumed finite first moments make the vector-valued integrals absolutely integrable, and the strictly positive kernel makes both partition functions positive.
\end{proof}

\subsection{\texorpdfstring{$\mathrm{SE}(3)$}{SE(3)}-Equivariance}
\label{app:equivariance}

We now establish that the drifting field respects the natural symmetries of molecular configuration space. This section is novel relative to~\citet{dengGenerativeModelingDrifting2026}, which considers image data without geometric symmetries.

\begin{lemma}[$\mathrm{SE}(3)$-invariance of the kernel]
\label{lem:kernel_invariance}
For all $g \in \mathrm{SE}(3)$ and $\mathbf{X}, \mathbf{Y} \in \mathcal{M}$,
\begin{equation}
k(g \cdot \mathbf{X}, g \cdot \mathbf{Y}) = k(\mathbf{X}, \mathbf{Y}).
\end{equation}
\end{lemma}

\begin{proof}
Write $g = (R, t)$. Then $g \cdot \mathbf{X} - g \cdot \mathbf{Y} = (R\mathbf{X} + t) - (R\mathbf{Y} + t) = R(\mathbf{X} - \mathbf{Y})$. Since $R \in \mathrm{SO}(3)$ is orthogonal, $\lVert R(\mathbf{X} - \mathbf{Y}) \rVert = \lVert \mathbf{X} - \mathbf{Y}\rVert$. Substituting into Eq.~\eqref{eq:kernel_def} gives the result.
\end{proof}

\begin{lemma}[$\mathrm{SE}(3)$-invariance of the partition function]
\label{lem:Z_invariance}
For all $g \in \mathrm{SE}(3)$, $\mathbf{X} \in \mathcal{M}$, and $\mu \in \mathcal{P}(\mathcal{M})$,
\begin{equation}
Z_{g_\# \mu}(g \cdot \mathbf{X}) = Z_\mu(\mathbf{X}).
\end{equation}
\end{lemma}

\begin{proof}
By the change-of-variables formula for pushforwards (\cref{def:pushforward}), if $\mathbf{Y} \sim \mu$ then $g \cdot \mathbf{Y} \sim g_\# \mu$. Therefore
\begin{align}
Z_{g_\# \mu}(g \cdot \mathbf{X}) &= \int_\mathcal{M} k(g \cdot \mathbf{X}, \mathbf{Y}')\, d(g_\#\mu)(\mathbf{Y}') \\
&= \int_\mathcal{M} k(g \cdot \mathbf{X}, g \cdot \mathbf{Y})\, d\mu(\mathbf{Y}) \\
&= \int_\mathcal{M} k(\mathbf{X}, \mathbf{Y})\, d\mu(\mathbf{Y}) = Z_\mu(\mathbf{X}),
\end{align}
where the second equality is the change of variables and the third uses \cref{lem:kernel_invariance}.
\end{proof}

\begin{theorem}[$\mathrm{SE}(3)$-equivariance of the drifting field]
\label{thm:equivariance}
For all $g = (R, t) \in \mathrm{SE}(3)$, $p, q \in \mathcal{P}(\mathcal{M})$, and $\mathbf{X} \in \mathcal{M}$,
\begin{equation}
\mathcal{V}_{g_\# p,\, g_\# q}(g \cdot \mathbf{X}) = R \cdot \mathcal{V}_{p,q}(\mathbf{X}).
\label{eq:equivariance}
\end{equation}
\end{theorem}

\begin{proof}
We compute the transformed attractive sub-field. By \cref{def:subfields} and \cref{lem:Z_invariance},
\begin{align}
\mathcal{V}^{+}_{g_\# p}(g \cdot \mathbf{X}) &= \frac{1}{Z_{g_\# p}(g \cdot \mathbf{X})} \int k(g \cdot \mathbf{X}, \mathbf{Y}')(\mathbf{Y}' - g \cdot \mathbf{X})\, d(g_\# p)(\mathbf{Y}') \\
&= \frac{1}{Z_p(\mathbf{X})} \int k(g \cdot \mathbf{X}, g \cdot \mathbf{Y})(g \cdot \mathbf{Y} - g \cdot \mathbf{X})\, dp(\mathbf{Y}) \label{eq:cov_step}\\
&= \frac{1}{Z_p(\mathbf{X})} \int k(\mathbf{X}, \mathbf{Y})\, R(\mathbf{Y} - \mathbf{X})\, dp(\mathbf{Y}) \label{eq:translation_cancel}\\
&= R \cdot \frac{1}{Z_p(\mathbf{X})} \int k(\mathbf{X}, \mathbf{Y})(\mathbf{Y} - \mathbf{X})\, dp(\mathbf{Y}) \label{eq:linearity}\\
&= R \cdot \mathcal{V}^{+}_p(\mathbf{X}).
\end{align}
The change of variables in Eq.~\eqref{eq:cov_step} uses the pushforward relation. Eq.~\eqref{eq:translation_cancel} uses \cref{lem:kernel_invariance} on the kernel and the identity $g \cdot \mathbf{Y} - g \cdot \mathbf{X} = R(\mathbf{Y} - \mathbf{X})$ on the displacement: the translation $t$ cancels in the difference. Eq.~\eqref{eq:linearity} factors $R$ out of the linear integral.

The same argument applied to $\mathcal{V}^{-}_q$ yields $\mathcal{V}^{-}_{g_\# q}(g \cdot \mathbf{X}) = R \cdot \mathcal{V}^{-}_q(\mathbf{X})$. Linearity of the difference gives Eq.~\eqref{eq:equivariance}.
\end{proof}

\begin{corollary}[Equivariance of the update]
\label{cor:update_equivariance}
The update map $\Phi_{p,q}: \mathbf{X} \mapsto \mathbf{X} + \mathcal{V}_{p,q}(\mathbf{X})$ commutes with the $\mathrm{SE}(3)$ action:
\begin{equation}
\Phi_{g_\# p, g_\# q}(g \cdot \mathbf{X}) = g \cdot \Phi_{p,q}(\mathbf{X}).
\end{equation}
\end{corollary}

\begin{proof}
$\Phi_{g_\# p, g_\# q}(g \cdot \mathbf{X}) = (g \cdot \mathbf{X}) + R \cdot \mathcal{V}_{p,q}(\mathbf{X}) = R\mathbf{X} + t + R \mathcal{V}_{p,q}(\mathbf{X}) = R(\mathbf{X} + \mathcal{V}_{p,q}(\mathbf{X})) + t = g \cdot \Phi_{p,q}(\mathbf{X})$.
\end{proof}

\subsection{Empirical Estimator with Finite Samples}
\label{app:empirical}

We now specialize the continuous drifting field of \cref{def:kernel_drift} to the empirical regime used at training time, taking care to track the structure of reaction-pathway data: each training example consists of one ground-truth pathway and a swarm of generated pathways, and the kernel is taken over full pathways rather than single configurations.

\paragraph{Pathway-level data structure.}
Each training example provides:
\begin{itemize}
    \item A single reference MEP pathway $\mathbf{Y}^\star \in \mathcal{M}^F$ drawn from a precomputed NEB pathway in Transition1x or Halo8, with $\mathbf{Y}^\star_1 = \mathbf{X}_R$ and $\mathbf{Y}^\star_F = \mathbf{X}_P$.
    \item A swarm of $K$ generated pathways $\{\mathbf{X}_1, \ldots, \mathbf{X}_K\} \subset \mathcal{M}^F$, computed by passing $K$ independently sampled prior interpolations through the current generator $f_\theta$. Following~\citet{dengGenerativeModelingDrifting2026}, the swarm itself serves as the empirical sample for both the points at which $\mathcal{V}$ is evaluated and the negatives used to estimate $\mathcal{V}^{-}_q$.
\end{itemize}
The use of a single positive sample $\mathbf{Y}^\star$ paired with the specific endpoints $(\mathbf{X}_R, \mathbf{X}_P)$ of each training example is a structural feature of reaction-pathway data that distinguishes our setting from the i.i.d.\ multi-positive regime ($N_{\mathrm{pos}} \geq 64$) studied in~\citet{dengGenerativeModelingDrifting2026}: each prior interpolation has exactly one corresponding ground-truth pathway, rather than a pool of unpaired positives drawn from a shared dataset.

\paragraph{Endpoint-pinning generator.}
The generator $f_\theta$ produces an $F$-images pathway from a prior interpolation while preserving the reactant and product geometries exactly. Given an $\mathrm{SE}(3)$-equivariant displacement network $\boldsymbol{\delta}_\theta : \mathcal{M}^F \to \mathcal{M}^F$ (LEFTNet in our experiments) and the endpoint-pinning envelope
\begin{equation}
w(\alpha) := \sin(\pi \alpha), \qquad w(0) = w(1) = 0,
\label{eq:envelope}
\end{equation}
the generator is defined image-wise by
\begin{equation}
\bigl[f_\theta(\mathbf{Z})\bigr]_k := \mathbf{Z}_k + w(\alpha_k) \cdot \boldsymbol{\delta}_\theta(\mathbf{Z})_k, \qquad \alpha_k = \frac{k - 1}{F - 1},
\label{eq:generator_pinned}
\end{equation}
where $\mathbf{Z}$ is the linear-interpolation prior and $\boldsymbol{\delta}_\theta(\mathbf{Z})_k$ is the predicted displacement at image $k$. Because $w(\alpha_1) = w(\alpha_F) = 0$, the generator preserves $\mathbf{X}_R$ and $\mathbf{X}_P$ exactly: every swarm pathway shares these boundary images with $\mathbf{Y}^\star$. Consequently, image $k = 1$ and $k = F$ contribute zero to all pairwise pathway distances, and the affinity computation below is effectively driven by the $F - 2$ intermediate images where the pathways differ.

\paragraph{Pathway distance and kernel.}
We treat each pathway as a flat point cloud of $M = F \cdot N$ atoms in $\mathbb{R}^3$. The pathway distance is the per-coordinate root-mean-square deviation:
\begin{equation}
d(\mathbf{X}, \mathbf{Y}) := \frac{1}{\sqrt{3M}} \sqrt{\sum_{k=1}^F \sum_{n=1}^N \lVert \mathbf{X}^{(n)}_k - \mathbf{Y}^{(n)}_k \rVert^2},
\label{eq:pathway_distance}
\end{equation}
and the pathway kernel is
\begin{equation}
k(\mathbf{X}, \mathbf{Y}) := \exp\!\left(-\frac{d(\mathbf{X}, \mathbf{Y})}{\tau}\right).
\label{eq:pathway_kernel}
\end{equation}
The $1/\sqrt{3M}$ normalization places the temperature $\tau$ on a per-coordinate RMSD scale, decoupling its physical meaning from the molecule size $N$ or the pathway resolution $F$. In practice, we usually set $\tau = 0.05$~\AA{}.

\paragraph{Per-sample Kabsch alignment.}
Because the generator's outputs and the reference $\mathbf{Y}^\star$ live in different reference frames (rotations and translations are not constrained), we align each swarm pathway to the reference before computing affinities. For each $i$, let $R_i \in \mathrm{SO}(3)$ and $\mathbf{c}_i \in \mathbb{R}^3$ be the optimal rotation and translation determined by the Kabsch algorithm~\cite{kabschSolutionBestRotation1976}. The aligned pathway is
\begin{equation}
\tilde{\mathbf{X}}_i := R_i (\mathbf{X}_i - \bar{\mathbf{X}}_i) + \bar{\mathbf{Y}}^\star,
\label{eq:kabsch}
\end{equation}
where $\bar{\mathbf{X}}_i$ and $\bar{\mathbf{Y}}^\star$ denote the centroids over all $M$ atoms. All pairwise distances and affinity computations below are evaluated on the aligned coordinates $\tilde{\mathbf{X}}_i$. After the drift is computed, it is rotated back into the original frame by $R_i^\top$ (\cref{prop:empirical_equivariance}).

\subsubsection{Naive Plug-In Estimator}

\begin{definition}[Plug-in estimator]
\label{def:plugin}
The naive plug-in estimator of Eq.~\eqref{eq:V_full} at a generated pathway $\mathbf{X}_i$, with one positive and $K$ negatives, is
\begin{equation}
\hat{\mathcal{V}}_{p,q}^{\mathrm{naive}}(\mathbf{X}_i) := \underbrace{\frac{k(\mathbf{X}_i, \mathbf{Y}^\star)(\mathbf{Y}^\star - \mathbf{X}_i)}{k(\mathbf{X}_i, \mathbf{Y}^\star)}}_{=\, \mathbf{Y}^\star - \mathbf{X}_i} \;-\; \frac{\sum_{j=1}^K k(\mathbf{X}_i, \mathbf{X}_j)(\mathbf{X}_j - \mathbf{X}_i)}{\sum_{j=1}^K k(\mathbf{X}_i, \mathbf{X}_j)}.
\label{eq:naive}
\end{equation}
\end{definition}

\begin{remark}[Failure modes of the naive estimator]
\label{rem:naive_failure}
The naive estimator has two structural deficiencies in the single-positive regime:
\begin{enumerate}
    \item \textbf{Unconditional attraction.} With a single positive sample, the empirical attractive partition $\hat Z_p(\mathbf{X}_i) = k(\mathbf{X}_i, \mathbf{Y}^\star)$ cancels the numerator entirely, reducing the attractive term to $\mathbf{Y}^\star - \mathbf{X}_i$. The kernel, and thus the temperature $\tau$, has no effect on the attractive term: the estimator cannot modulate the strength of attraction by distance to $\mathbf{Y}^\star$.
    \item \textbf{Decoupled normalization.} The attractive and repulsive terms are normalized by independent partition functions, $\hat Z_p$ and $\hat Z_q$ respectively. When the swarm is far from $\mathbf{Y}^\star$, $\hat Z_q$ may be much larger than $\hat Z_p$, leading to repulsion dominating attraction and the swarm drifting away from the target rather than toward it.
\end{enumerate}
The multi-positive regime of~\citet{dengGenerativeModelingDrifting2026} mitigates the first deficiency through positive-sample diversity, but the second persists. Both are addressed simultaneously by the joint normalization introduced below.
\end{remark}

\subsubsection{Joint normalization via row-and-column softmax}

To address both deficiencies, we replace independent row-wise softmax normalization with a jointly normalized affinity matrix that treats $\mathbf{Y}^\star$ and the aligned swarm members $\{\tilde{\mathbf{X}}_j\}$ on equal footing. The construction below follows~\citet[Algorithm 2]{dengGenerativeModelingDrifting2026} closely: we make the empirical correspondence with the continuous drift formula explicit and adapt it to the single-positive pathway-level setting.

\begin{definition}[Augmented target set and logit matrix]
\label{def:logits}
Define the augmented target set $\mathcal{T} = (\mathcal{T}_0, \mathcal{T}_1, \ldots, \mathcal{T}_K)$ with $\mathcal{T}_0 := \mathbf{Y}^\star$ and $\mathcal{T}_j := \tilde{\mathbf{X}}_j$ for $j \in \{1, \ldots, K\}$. Before masking, the pairwise logit matrix $L \in \mathbb{R}^{K \times (K+1)}$ is
\begin{equation}
L_{ij} := -\frac{d(\tilde{\mathbf{X}}_i, \mathcal{T}_j)}{\tau}, \quad i \in \{1, \ldots, K\}, \; j \in \{0, 1, \ldots, K\},
\label{eq:logits_def}
\end{equation}
with self-distance masking applied to the entry corresponding to a swarm pathway being its own negative:
\begin{equation}
L_{i, i} := -\infty.
\label{eq:self_mask}
\end{equation}
Eq.~\eqref{eq:self_mask} ensures that each swarm pathway does not appear as its own negative sample, which would otherwise dominate the repulsive sum and induce trivial self-cancellation.
\end{definition}

\begin{definition}[Row and column softmaxes]
\label{def:softmax}
For $K\geq 2$ and $L \in (\mathbb{R}\cup\{-\infty\})^{K \times (K+1)}$ with the mask above, define
\begin{equation}
\sigma_{\mathrm{row}}(L)_{ij} := \frac{\exp(L_{ij})}{\sum_{j'=0}^K \exp(L_{ij'})}, \qquad \sigma_{\mathrm{col}}(L)_{ij} := \frac{\exp(L_{ij})}{\sum_{i'=1}^K \exp(L_{i'j})}.
\label{eq:softmax}
\end{equation}
\end{definition}

\begin{definition}[Geometric-mean joint affinity, \citet{dengGenerativeModelingDrifting2026} Algorithm 2]
\label{def:sinkhorn}
Following~\citet[Algorithm 2]{dengGenerativeModelingDrifting2026}, the joint affinity matrix $\mathbf{A} \in [0,1]^{K \times (K+1)}$ is the entrywise geometric mean of the row and column softmaxes:
\begin{equation}
\mathbf{A}_{ij} := \sqrt{\sigma_{\mathrm{row}}(L)_{ij} \cdot \sigma_{\mathrm{col}}(L)_{ij}}.
\label{eq:sinkhorn_A}
\end{equation}
\end{definition}

\begin{remark}[Connection to optimal transport]
\label{rem:sinkhorn_comparison}
The geometric-mean form of Eq.~\eqref{eq:sinkhorn_A} combines row and column softmax weights in one step. It is a balancing heuristic, not a doubly stochastic matrix: $\mathbf{A}$ is rectangular, and its row and column sums are not prescribed. The construction follows~\citet{dengGenerativeModelingDrifting2026} for computational simplicity.
\end{remark}

\begin{definition}[Attractive and repulsive affinities]
\label{def:affinity_split}
We partition $\mathbf{A}$ as
\begin{equation}
A^{+}_i := \mathbf{A}_{i,0}, \qquad A^{-}_{ij} := \mathbf{A}_{i,j} \quad \text{for } j \in \{1, \ldots, K\}.
\end{equation}
\end{definition}

\begin{proposition}[Row softmax uses an augmented empirical normalizer]
\label{prop:Z_absorbed}
Define the masked kernel weights $\bar k_{ij}:=\exp(L_{ij})$. For $j\ne i$, these equal $k(\tilde{\mathbf{X}}_i,\mathcal{T}_j)$; for the self-negative entry $j=i$, $\bar k_{ii}=0$. Then row softmax is exactly
\begin{equation}
\sigma_{\mathrm{row}}(L)_{ij} = \frac{\bar k_{ij}}{\hat Z_{\mathrm{aug}}(\tilde{\mathbf{X}}_i)}, \qquad \hat Z_{\mathrm{aug}}(\tilde{\mathbf{X}}_i) := \sum_{j'=0}^K \bar k_{ij'}.
\label{eq:Z_absorption}
\end{equation}
The denominator is one augmented row sum, not the product $\hat Z_p\hat Z_q$ in Eq.~\eqref{eq:empirical_unified}. Column softmax and the geometric mean alter the weights further.
\end{proposition}

\begin{proof}
Substitute $\bar k_{ij}=\exp(L_{ij})$ into the definition of row softmax in Eq.~\eqref{eq:softmax}. The mask gives $\exp(-\infty)=0$ at $j=i$; every other entry agrees with the pathway kernel by Eqs.~\eqref{eq:pathway_kernel} and~\eqref{eq:logits_def}.
\end{proof}

\subsubsection{Empirical Drifting Field}

We can now assemble the empirical drifting field used in Drift-React.

\begin{definition}[Empirical drifting field, aligned frame]
\label{def:empirical_drift}
Let $\tilde{\mathbf{X}}_i, \mathbf{Y}^\star, \{\tilde{\mathbf{X}}_j\}_{j=1}^K \in \mathcal{M}^F$ be the Kabsch-aligned swarm pathway, the reference pathway, and the aligned swarm of negatives. With affinities $A^+_i, A^-_{ij}$ defined by Eqs.~\eqref{eq:logits_def}~through~\eqref{eq:sinkhorn_A} and \cref{def:affinity_split}, the empirical drifting field at $\tilde{\mathbf{X}}_i$, expressed in the aligned frame, is
\begin{equation}
\boxed{\quad \tilde{\mathcal{V}}(\tilde{\mathbf{X}}_i) = A^+_i \sum_{j=1}^{K} A^-_{ij}\, (\mathbf{Y}^\star - \tilde{\mathbf{X}}_j). \quad}
\label{eq:final_drift_aligned}
\end{equation}
The drift in the original (unaligned) frame is recovered by the inverse rotation
\begin{equation}
\mathcal{V}(\mathbf{X}_i) = R_i^\top \cdot \tilde{\mathcal{V}}(\tilde{\mathbf{X}}_i),
\label{eq:final_drift_unaligned}
\end{equation}
where $R_i$ is the Kabsch rotation of Eq.~\eqref{eq:kabsch}.
\end{definition}

The product-kernel form of \cref{prop:unified} motivates the displacement $(\mathbf{Y}^+ - \mathbf{Y}^-)$ in this construction. With one positive ($\mathbf{Y}^+ = \mathbf{Y}^\star$) and $K$ aligned negatives ($\mathbf{Y}^- \in \{\tilde{\mathbf{X}}_j\}$), a plug-in estimate of Eq.~\eqref{eq:unified} would instead be
\begin{equation}
\hat{\mathcal{V}}(\tilde{\mathbf{X}}_i) = \frac{1}{\hat Z_p(\tilde{\mathbf{X}}_i)\, \hat Z_q(\tilde{\mathbf{X}}_i)} \sum_{j=1}^K k(\tilde{\mathbf{X}}_i, \mathbf{Y}^\star)\, k(\tilde{\mathbf{X}}_i, \tilde{\mathbf{X}}_j)\, (\mathbf{Y}^\star - \tilde{\mathbf{X}}_j).
\label{eq:empirical_unified}
\end{equation}
The masked row softmax, column softmax, and geometric mean define different weights from the separate-normalizer estimator in Eq.~\eqref{eq:empirical_unified}; Eq.~\eqref{eq:final_drift_aligned} is therefore a chosen empirical field, not an algebraic consequence of Eq.~\eqref{eq:unified}. Equation~\eqref{eq:final_drift_unaligned} rotates its displacement vectors back to the original frame using $R_i^\top$. This frame conversion treats the fitted Kabsch rotation as fixed for the displacement; it is not a claim about the derivative of the Kabsch optimizer with respect to its input.

\begin{remark}[Implementation note]
\label{rem:implementation}
The boxed Eq.~\eqref{eq:final_drift_aligned} is implemented in matrix form as
\begin{equation}
\tilde{\mathcal{V}}(\tilde{\mathbf{X}}_i) = \Bigl(A^+_i \sum_j A^-_{ij}\Bigr) \mathbf{Y}^\star - A^+_i \sum_j A^-_{ij}\, \tilde{\mathbf{X}}_j,
\end{equation}
which is algebraically identical to Eq.~\eqref{eq:final_drift_aligned} (since $A^+_i$ does not depend on $j$) but expresses the computation as two batched matrix-vector products amenable to GPU acceleration.
\end{remark}

\begin{proposition}[$\mathrm{SE}(3)$-equivariance of the empirical drift]
\label{prop:empirical_equivariance}
Let $g = (R, t) \in \mathrm{SE}(3)$ act jointly on all configurations. Assume the optimal Kabsch rotation for each aligned pair is unique (or the optimizer selects rotations equivariantly). The empirical drift $\mathcal{V}(\mathbf{X}_i)$ defined by \cref{def:empirical_drift} (in the unaligned frame) is $\mathrm{SE}(3)$-equivariant:
\begin{equation}
\mathcal{V}(g \cdot \mathbf{X}_i)\bigr|_{g \cdot \mathbf{Y}^\star,\, \{g \cdot \mathbf{X}_j\}} = R \cdot \mathcal{V}(\mathbf{X}_i)\bigr|_{\mathbf{Y}^\star,\, \{\mathbf{X}_j\}}.
\end{equation}
\end{proposition}

\begin{proof}
Equivariance has two parts: (i) the aligned-frame drift transforms as a displacement under a joint rigid motion, and (ii) the inverse alignment rotation converts it back to the original frame.

\textbf{(i) Covariance in the aligned frame.} Under a joint transformation of both the swarm and reference, their centroids transform by $\bar{\mathbf{X}}_i \mapsto R\bar{\mathbf{X}}_i+t$ and $\bar{\mathbf{Y}}^\star \mapsto R\bar{\mathbf{Y}}^\star+t$. Uniqueness of the Kabsch optimum gives $R_i \mapsto R R_i R^\top$, so Eq.~\eqref{eq:kabsch} gives $\widetilde{(g\cdot\mathbf{X}_j)}=g\cdot\tilde{\mathbf{X}}_j$ for every swarm member, while the reference becomes $g\cdot\mathbf{Y}^\star$. Pairwise distances, logits, and affinities are unchanged because rigid motions preserve distance. Each displacement in Eq.~\eqref{eq:final_drift_aligned} transforms as $g\cdot\mathbf{Y}^\star-g\cdot\tilde{\mathbf{X}}_j=R(\mathbf{Y}^\star-\tilde{\mathbf{X}}_j)$. Hence the aligned-frame drift transforms as $\tilde{\mathcal{V}}' = R\tilde{\mathcal{V}}$.

\textbf{(ii) Equivariance of the inverse rotation.} Applying Eq.~\eqref{eq:final_drift_unaligned} to the transformed inputs gives
\begin{equation}
\mathcal{V}(g \cdot \mathbf{X}_i) = (R R_i R^\top)^\top (R\tilde{\mathcal{V}}) = R R_i^\top\tilde{\mathcal{V}} = R \mathcal{V}(\mathbf{X}_i),
\end{equation}
which is the $\mathrm{SE}(3)$-equivariance condition. The translation $t$ does not appear because the drift is a displacement vector, not a position.
\end{proof}

\begin{remark}[Role of Kabsch alignment]
\label{rem:kabsch_essential}
Per-sample Kabsch alignment makes affinity comparisons insensitive to an independently rotated or translated swarm pathway relative to the reference. It is not required for equivariance under a rigid motion applied jointly to all inputs: raw pairwise Euclidean distances are already invariant under such a motion, and raw displacement vectors already rotate equivariantly. The alignment-then-inverse-rotation protocol preserves joint equivariance under the uniqueness condition of \cref{prop:empirical_equivariance}.
\end{remark}

\subsection{Fixed-Point Training Objective}
\label{app:loss}

For the continuous field, \cref{prop:antisymmetry_equilibrium} gives a fixed-point relation at any generator with $q_{\hat\theta}=p$. Then for $\mathbf{Z} \sim p_{\mathrm{prior}}$,
\begin{equation}
f_{\hat\theta}(\mathbf{Z}) = f_{\hat\theta}(\mathbf{Z}) + \mathcal{V}_{p, q_{\hat\theta}}\bigl(f_{\hat\theta}(\mathbf{Z})\bigr).
\label{eq:fixed_point}
\end{equation}

Following \citet{dengGenerativeModelingDrifting2026}, we use this relation to motivate a self-distillation loss with the stop-gradient operator $\mathrm{sg}(\cdot)$, which acts as the identity on its argument's value but blocks gradient flow. The finite-sample field in \cref{def:empirical_drift} is a different estimator; it need not vanish for every sample even when its underlying generated and target distributions match.

\begin{definition}[Self-distillation loss, \citet{dengGenerativeModelingDrifting2026} Eq.~(6)]
\label{def:loss}
The Drift-React training loss is
\begin{equation}
\mathcal{L}(\theta) := \mathbb{E}_{\mathbf{Z} \sim p_{\mathrm{prior}}} \left[\bigl\lVert f_\theta(\mathbf{Z}) - \mathrm{sg}\bigl(f_\theta(\mathbf{Z}) + \mathcal{V}_{p, q_\theta}\bigl(f_\theta(\mathbf{Z})\bigr)\bigr)\bigr\rVert^2\right].
\label{eq:loss_def}
\end{equation}
The drift $\mathcal{V}_{p, q_\theta}$ is computed using the empirical field of \cref{def:empirical_drift} in the original (unaligned) frame, so that the regression target $\mathrm{sg}(f_\theta(\mathbf{Z}) + \mathcal{V})$ lives in the same frame as the generator output $f_\theta(\mathbf{Z})$.
\end{definition}

%% file: checklist.tex
\section*{NeurIPS Paper Checklist}

\begin{enumerate}

\item {\bf Claims}
    \item[] Question: Do the main claims made in the abstract and introduction accurately reflect the paper's contributions and scope?
    \item[] Answer: \answerYes{}
    \item[] Justification: The abstract and introduction state the paper's contributions (one-step generation of complete reaction pathways via an SE(3)-equivariant drifting field) and scope (full-pathway generation as the primary task, point-wise TS prediction as a secondary evaluation). Quantitative claims in the abstract and introduction match the experimental results in Section~\ref{sec:results} and the corresponding table.

\item {\bf Limitations}
    \item[] Question: Does the paper discuss the limitations of the work performed by the authors?
    Answer: \answerYes{}
    \item[] Justification: We discuss the dependence on a single positive sample per training reaction, the higher per-image energy MAE on Transition1x relative to classical baselines (Section~\ref{sec:results}), the sensitivity of saddle-point energetics to small coordinate perturbations limiting head-to-head TS-only energy comparison with single-structure generators (Section~\ref{sec:results}), and the lack of a head-to-head comparison with the concurrent MEPIN method due to differing levels of theory and data splits.

\item {\bf Theory assumptions and proofs}
    \item[] Question: For each theoretical result, does the paper provide the full set of assumptions and a complete (and correct) proof?
    \item[] Answer: \answerYes{}
    \item[] Justification: The paper's theoretical content (continuous drifting field formulation, anti-symmetry, and the failure-mode analysis of single-positive Monte Carlo estimation) is presented with assumptions stated inline; full derivations and proofs are deferred to Appendix~\ref{app:empirical}.

    \item {\bf Experimental result reproducibility}
    \item[] Question: Does the paper fully disclose all the information needed to reproduce the main experimental results of the paper to the extent that it affects the main claims and/or conclusions of the paper (regardless of whether the code and data are provided or not)?
    \item[] Answer: \answerYes{}
    \item[] Justification: The architecture (LEFTNet backbone with sinusoidal endpoint pinning), training algorithm (Algorithm~\ref{alg:drifting}), loss formulation, hyperparameters, dataset splits, and evaluation protocol are described in Section~\ref{sec:results} and Appendix~\ref{app:baselines}. The Halo8 reactant-grouped split procedure is documented in Appendix~\ref{app:splits}.

\item {\bf Open access to data and code}
    \item[] Question: Does the paper provide open access to the data and code, with sufficient instructions to faithfully reproduce the main experimental results, as described in supplemental material?
    \item[] Answer: \answerYes{}
    \item[] Justification: The full training and evaluation codebase, including data preprocessing, baseline implementations, hyperparameter configurations, and evaluation scripts, is released under a permissive open-source license alongside the paper. Datasets used in the work are publicly available: Transition1x is released by~\citet{schreinerTransition1xDatasetBuilding2022} under CC-BY, and Halo8 is released by~\citet{leeDatasetChemicalReaction2025}. The Halo8 reactant-grouped split is regenerated deterministically from the released dataset using the procedure described in Appendix~\ref{app:splits} (fixed random seed).

\item {\bf Experimental setting/details}
    \item[] Question: Does the paper specify all the training and test details (e.g., data splits, hyperparameters, how they were chosen, type of optimizer) necessary to understand the results?
    \item[] Answer: \answerYes{}
    \item[] Justification: Training hyperparameters (swarm size $K$, kernel temperature $\tau$, learning rate, optimizer, EMA decay, prior noise level, $F$), dataset splits, and baseline configurations are reported in Appendix~\ref{app:baselines}. Transition1x uses the standardized split released with~\citet{duanOptimalTransportGenerating2025}, the Halo8 reactant-grouped split is described in Appendix~\ref{app:splits}.

\item {\bf Experiment statistical significance}
    \item[] Question: Does the paper report error bars suitably and correctly defined or other appropriate information about the statistical significance of the experiments?
    \item[] Answer: \answerYes{}
    \item[] Justification: All metrics in Table~\ref{tab:benchmark_results} are reported as mean $\pm$ standard deviation across the test set. The standard deviations capture per-reaction variability in the held-out test partition.

\item {\bf Experiments compute resources}
    \item[] Question: For each experiment, does the paper provide sufficient information on the computer resources (type of compute workers, memory, time of execution) needed to reproduce the experiments?
    \item[] Answer: \answerYes{}
    \item[] Justification: All Drift-React training and inference runs used a single NVIDIA H100 GPU (80~GB GPU memory) with 64~GB system memory and 12 CPU workers per task, with a 48-hour wall-clock budget per training run (Appendix~\ref{app:compute}). NeuralNEB retraining and inference used the same compute setup. Single-point DFT energy evaluations for the energetic metrics were run on the local CPU cluster (PySCF for Transition1x, ORCA for Halo8). The total compute budget for the project, including preliminary hyperparameter sweeps, NeuralNEB retraining on both datasets, and full-test-set DFT evaluation, is approximately 1{,}000 GPU-hours and 20{,}000 CPU-hours.

\item {\bf Code of ethics}
    \item[] Question: Does the research conducted in the paper conform, in every respect, with the NeurIPS Code of Ethics \url{https://neurips.cc/public/EthicsGuidelines}?
    \item[] Answer: \answerYes{}
    \item[] Justification: The work uses only publicly released computational chemistry datasets (Transition1x, Halo8) and standard machine-learning methodology, and conforms to the NeurIPS Code of Ethics.

\item {\bf Broader impacts}
    \item[] Question: Does the paper discuss both potential positive societal impacts and negative societal impacts of the work performed?
    \item[] Answer: \answerYes{}
    \item[] Justification: Faster transition state and pathway prediction can accelerate the design of catalysts, enzymes, and pharmaceuticals, with attendant benefits for chemical research and materials discovery. The method does not introduce new dual-use risks beyond those of existing computational chemistry tools, and its outputs require subsequent quantum-chemical validation before downstream use, limiting the risk of unintended deployment.

\item {\bf Safeguards}
    \item[] Question: Does the paper describe safeguards that have been put in place for responsible release of data or models that have a high risk for misuse (e.g., pre-trained language models, image generators, or scraped datasets)?
    \item[] Answer: \answerNA{}
    \item[] Justification: The paper does not release a model or dataset with high misuse risk. The trained Drift-React weights are a small-scale geometric prediction model for reaction pathways without dual-use potential beyond standard computational chemistry tools.

\item {\bf Licenses for existing assets}
    \item[] Question: Are the creators or original owners of assets (e.g., code, data, models), used in the paper, properly credited and are the license and terms of use explicitly mentioned and properly respected?
    \item[] Answer: \answerYes{}
    \item[] Justification: All used datasets (Transition1x~\cite{schreinerTransition1xDatasetBuilding2022}, Halo8~\cite{leeDatasetChemicalReaction2025}) and code components (LEFTNet~\cite{duNewPerspectiveBuilding2023}, NeuralNEB~\cite{schreinerNeuralNEBNeuralNetworks2022}, geodesic interpolation~\cite{zhuGeodesicInterpolationReaction2019}) are properly cited and used in accordance with their respective licenses. Transition1x is released under CC-BY, LEFTNet and NeuralNEB are released under permissive licenses on the authors' GitHub repositories. Our use is consistent with the terms of these licenses.

\item {\bf New assets}
    \item[] Question: Are new assets introduced in the paper well documented and is the documentation provided alongside the assets?
    \item[] Answer: \answerYes{}
    \item[] Justification: The Drift-React codebase, trained model weights for Transition1x and Halo8, and the Halo8 reactant-grouped split (released as a deterministic seed and split-generation script) are documented in the released repository, including: a README with installation and quickstart instructions, training scripts with the exact hyperparameter configurations from Appendix~\ref{app:drift_ts_config}, evaluation scripts that reproduce the metrics in Table~\ref{tab:benchmark_results} and Figure~\ref{fig:ts-gen}, and a model card detailing intended use, limitations, and the reference levels of theory used during evaluation. The released code is licensed under MIT.

\item {\bf Crowdsourcing and research with human subjects}
    \item[] Question: For crowdsourcing experiments and research with human subjects, does the paper include the full text of instructions given to participants and screenshots, if applicable, as well as details about compensation (if any)? 
    \item[] Answer: \answerNA{}
    \item[] Justification: The paper does not involve crowdsourcing or research with human subjects.

\item {\bf Institutional review board (IRB) approvals or equivalent for research with human subjects}
    \item[] Question: Does the paper describe potential risks incurred by study participants, whether such risks were disclosed to the subjects, and whether Institutional Review Board (IRB) approvals (or an equivalent approval/review based on the requirements of your country or institution) were obtained?
    \item[] Answer: \answerNA{}
    \item[] Justification: The paper does not involve research with human subjects.

\item {\bf Declaration of LLM usage}
    \item[] Question: Does the paper describe the usage of LLMs if it is an important, original, or non-standard component of the core methods in this research? Note that if the LLM is used only for writing, editing, or formatting purposes and does \emph{not} impact the core methodology, scientific rigor, or originality of the research, declaration is not required.
    \item[] Answer: \answerNA{}
    \item[] Justification: LLMs were used only for writing assistance (grammar, phrasing, and formatting) and did not contribute to the core methodology, scientific results, or originality of the research.

\end{enumerate}

%% file: references.bib
@article{asgeirssonNudgedElasticBand2021,
  title = {Nudged {{Elastic Band Method}} for {{Molecular Reactions Using Energy-Weighted Springs Combined}} with {{Eigenvector Following}}},
  author = {{\'A}sgeirsson, Vilhj{\'a}lmur and Birgisson, Benedikt Orri and Bjornsson, Ragnar and Becker, Ute and Neese, Frank and Riplinger, Christoph and J{\'o}nsson, Hannes},
  year = 2021,
  month = aug,
  journal = {Journal of Chemical Theory and Computation},
  volume = {17},
  number = {8},
  pages = {4929--4945},
  publisher = {American Chemical Society},
  issn = {1549-9618},
  doi = {10.1021/acs.jctc.1c00462},
  url = {https://doi.org/10.1021/acs.jctc.1c00462},
  urldate = {2026-05-06}
}

@misc{beagleholeMachineLearningTransition2025,
  title = {Machine Learning Transition State Geometries and Applications in Reaction Property Prediction},
  author = {Beaglehole, Isaac W. and Pemberton, Miles J. and Farrar, Elliot H. E. and Grayson, Matthew N.},
  year = 2025,
  month = mar,
  publisher = {ChemRxiv},
  doi = {10.26434/chemrxiv-2024-wdtz9-v3},
  url = {https://chemrxiv.org/engage/chemrxiv/article-details/67cec10c6dde43c908420904},
  urldate = {2025-12-04},
  archiveprefix = {ChemRxiv},
  langid = {english}
}

@article{chengMeanShiftMode1995,
  title = {Mean Shift, Mode Seeking, and Clustering},
  author = {Cheng, Yizong},
  year = 1995,
  month = aug,
  journal = {IEEE Transactions on Pattern Analysis and Machine Intelligence},
  volume = {17},
  number = {8},
  pages = {790--799},
  issn = {1939-3539},
  doi = {10.1109/34.400568},
  url = {https://ieeexplore.ieee.org/document/400568},
  urldate = {2026-04-29}
}

@article{comaniciuMeanShiftRobust2002,
  title = {Mean Shift: A Robust Approach toward Feature Space Analysis},
  shorttitle = {Mean Shift},
  author = {Comaniciu, D. and Meer, P.},
  year = 2002,
  month = may,
  journal = {IEEE Transactions on Pattern Analysis and Machine Intelligence},
  volume = {24},
  number = {5},
  pages = {603--619},
  issn = {1939-3539},
  doi = {10.1109/34.1000236},
  url = {https://ieeexplore.ieee.org/document/1000236},
  urldate = {2026-04-29}
}

@misc{darouichTrainingDomainRobust2026,
  title = {Beyond the {{Training Domain}}: {{Robust Generative Transition State Models}} for {{Unseen Chemistry}}},
  shorttitle = {Beyond the {{Training Domain}}},
  author = {Darouich, Samir and Toney, Jacob W. and Luo, Weiliang and K{\"a}stner, Johannes and Niepert, Mathias and Kulik, Heather J.},
  year = 2026,
  month = jan,
  number = {arXiv:2601.16469},
  eprint = {2601.16469},
  primaryclass = {physics},
  publisher = {arXiv},
  doi = {10.48550/arXiv.2601.16469},
  url = {http://arxiv.org/abs/2601.16469},
  urldate = {2026-01-29},
  archiveprefix = {arXiv}
}

@misc{dengGenerativeModelingDrifting2026,
  title = {Generative {{Modeling}} via {{Drifting}}},
  author = {Deng, Mingyang and Li, He and Li, Tianhong and Du, Yilun and He, Kaiming},
  year = 2026,
  month = feb,
  number = {arXiv:2602.04770},
  eprint = {2602.04770},
  primaryclass = {cs},
  publisher = {arXiv},
  doi = {10.48550/arXiv.2602.04770},
  url = {http://arxiv.org/abs/2602.04770},
  urldate = {2026-02-18},
  archiveprefix = {arXiv}
}

@article{duanAccurateTransitionState2023,
  title = {Accurate Transition State Generation with an Object-Aware Equivariant Elementary Reaction Diffusion Model},
  author = {Duan, Chenru and Du, Yuanqi and Jia, Haojun and Kulik, Heather J.},
  year = 2023,
  month = dec,
  journal = {Nature Computational Science},
  volume = {3},
  number = {12},
  pages = {1045--1055},
  issn = {2662-8457},
  doi = {10.1038/s43588-023-00563-7},
  url = {https://www.nature.com/articles/s43588-023-00563-7},
  urldate = {2024-12-26},
  langid = {english}
}

@article{duanOptimalTransportGenerating2025,
  title = {Optimal Transport for Generating Transition States in Chemical Reactions},
  author = {Duan, Chenru and Liu, Guan-Horng and Du, Yuanqi and Chen, Tianrong and Zhao, Qiyuan and Jia, Haojun and Gomes, Carla P. and Theodorou, Evangelos A. and Kulik, Heather J.},
  year = 2025,
  month = apr,
  journal = {Nature Machine Intelligence},
  volume = {7},
  number = {4},
  pages = {615--626},
  publisher = {Nature Publishing Group},
  issn = {2522-5839},
  doi = {10.1038/s42256-025-01010-0},
  url = {https://www.nature.com/articles/s42256-025-01010-0},
  urldate = {2025-12-12},
  copyright = {2025 The Author(s)},
  langid = {english}
}

@misc{duNewPerspectiveBuilding2023,
  title = {A New Perspective on Building Efficient and Expressive {{3D}} Equivariant Graph Neural Networks},
  author = {Du, Weitao and Du, Yuanqi and Wang, Limei and Feng, Dieqiao and Wang, Guifeng and Ji, Shuiwang and Gomes, Carla and Ma, Zhi-Ming},
  year = 2023,
  month = apr,
  number = {arXiv:2304.04757},
  eprint = {2304.04757},
  primaryclass = {cs},
  publisher = {arXiv},
  doi = {10.48550/arXiv.2304.04757},
  url = {http://arxiv.org/abs/2304.04757},
  urldate = {2026-02-21},
  archiveprefix = {arXiv}
}

@article{eStringMethodStudy2002,
  title = {String {{Method}} for the {{Study}} of {{Rare Events}}},
  author = {E, Weinan and Ren, Weiqing and {Vanden-Eijnden}, Eric},
  year = 2002,
  month = aug,
  journal = {Physical Review B},
  volume = {66},
  number = {5},
  eprint = {cond-mat/0205527},
  pages = {052301},
  issn = {0163-1829, 1095-3795},
  doi = {10.1103/PhysRevB.66.052301},
  url = {http://arxiv.org/abs/cond-mat/0205527},
  urldate = {2026-03-18},
  archiveprefix = {arXiv}
}

@article{eyringActivatedComplexChemical1935,
  title = {The {{Activated Complex}} in {{Chemical Reactions}}},
  author = {Eyring, Henry},
  year = 1935,
  month = feb,
  journal = {The Journal of Chemical Physics},
  volume = {3},
  number = {2},
  pages = {107--115},
  issn = {0021-9606},
  doi = {10.1063/1.1749604},
  url = {https://doi.org/10.1063/1.1749604},
  urldate = {2026-03-18}
}

@article{fukuiPathChemicalReactions1981,
  title = {The Path of Chemical Reactions - the {{IRC}} Approach},
  author = {Fukui, Kenichi},
  year = 1981,
  month = dec,
  journal = {Accounts of Chemical Research},
  volume = {14},
  number = {12},
  pages = {363--368},
  publisher = {American Chemical Society},
  issn = {0001-4842},
  doi = {10.1021/ar00072a001},
  url = {https://doi.org/10.1021/ar00072a001},
  urldate = {2025-12-12}
}

@article{fukuiVariationalPrinciplesChemical1981,
  title = {Variational Principles in a Chemical Reaction},
  author = {Fukui, Kenichi},
  year = 1981,
  journal = {International Journal of Quantum Chemistry},
  volume = {20},
  number = {S15},
  pages = {633--642},
  issn = {1097-461X},
  doi = {10.1002/qua.560200866},
  url = {https://onlinelibrary.wiley.com/doi/abs/10.1002/qua.560200866},
  urldate = {2026-02-20},
  copyright = {Copyright \copyright{} 1981 John Wiley \& Sons, Inc.},
  langid = {english}
}

@article{fukunagaEstimationGradientDensity1975,
  title = {The Estimation of the Gradient of a Density Function, with Applications in Pattern Recognition},
  author = {Fukunaga, K. and Hostetler, L.},
  year = 1975,
  month = jan,
  journal = {IEEE Transactions on Information Theory},
  volume = {21},
  number = {1},
  pages = {32--40},
  issn = {1557-9654},
  doi = {10.1109/TIT.1975.1055330},
  url = {https://ieeexplore.ieee.org/document/1055330},
  urldate = {2026-04-29}
}

@article{goedeckerLinearScalingElectronic1999,
  title = {Linear Scaling Electronic Structure Methods},
  author = {Goedecker, Stefan},
  year = 1999,
  month = jul,
  journal = {Reviews of Modern Physics},
  volume = {71},
  number = {4},
  pages = {1085--1123},
  publisher = {American Physical Society},
  doi = {10.1103/RevModPhys.71.1085},
  url = {https://link.aps.org/doi/10.1103/RevModPhys.71.1085},
  urldate = {2026-03-18}
}

@article{gomez-bombarelliDesignEfficientMolecular2016,
  title = {Design of Efficient Molecular Organic Light-Emitting Diodes by a High-Throughput Virtual Screening and Experimental Approach},
  author = {{G{\'o}mez-Bombarelli}, Rafael and {Aguilera-Iparraguirre}, Jorge and Hirzel, Timothy D. and Duvenaud, David and Maclaurin, Dougal and {Blood-Forsythe}, Martin A. and Chae, Hyun Sik and Einzinger, Markus and Ha, Dong-Gwang and Wu, Tony and Markopoulos, Georgios and Jeon, Soonok and Kang, Hosuk and Miyazaki, Hiroshi and Numata, Masaki and Kim, Sunghan and Huang, Wenliang and Hong, Seong Ik and Baldo, Marc and Adams, Ryan P. and {Aspuru-Guzik}, Al{\'a}n},
  year = 2016,
  month = oct,
  journal = {Nature Materials},
  volume = {15},
  number = {10},
  pages = {1120--1127},
  publisher = {Nature Publishing Group},
  issn = {1476-4660},
  doi = {10.1038/nmat4717},
  url = {https://www.nature.com/articles/nmat4717},
  urldate = {2026-03-18},
  copyright = {2016 Springer Nature Limited},
  langid = {english}
}

@article{henkelmanClimbingImageNudged2000,
  title = {A {{Climbing Image Nudged Elastic Band Method}} for {{Finding Saddle Points}} and {{Minimum Energy Paths}}},
  author = {{Graeme Henkelman} and {Blas P. Uberuaga} and {Hannes J\'onsson}},
  year = 2000,
  journal = {ResearchGate},
  pages = {9901--9904},
  doi = {10.1063/1.1329672},
  url = {https://www.researchgate.net/publication/224877263_A_Climbing_Image_Nudged_Elastic_Band_Method_for_Finding_Saddle_Points_and_Minimum_Energy_Paths},
  urldate = {2026-03-17},
  langid = {english}
}

@article{hjorthlarsenAtomicSimulationEnvironment2017,
  title = {The Atomic Simulation Environment---a {{Python}} Library for Working with Atoms},
  author = {Hjorth Larsen, Ask and J{\o}rgen Mortensen, Jens and Blomqvist, Jakob and Castelli, Ivano E and Christensen, Rune and Du{\l}ak, Marcin and Friis, Jesper and Groves, Michael N and Hammer, Bj{\o}rk and Hargus, Cory and Hermes, Eric D and Jennings, Paul C and Bjerre Jensen, Peter and Kermode, James and Kitchin, John R and Leonhard Kolsbjerg, Esben and Kubal, Joseph and Kaasbjerg, Kristen and Lysgaard, Steen and Bergmann Maronsson, J{\'o}n and Maxson, Tristan and Olsen, Thomas and Pastewka, Lars and Peterson, Andrew and Rostgaard, Carsten and Schi{\o}tz, Jakob and Sch{\"u}tt, Ole and Strange, Mikkel and Thygesen, Kristian S and Vegge, Tejs and Vilhelmsen, Lasse and Walter, Michael and Zeng, Zhenhua and Jacobsen, Karsten W},
  year = 2017,
  month = jun,
  journal = {Journal of Physics: Condensed Matter},
  volume = {29},
  number = {27},
  pages = {273002},
  publisher = {IOP Publishing},
  issn = {0953-8984},
  doi = {10.1088/1361-648X/aa680e},
  url = {https://doi.org/10.1088/1361-648X/aa680e},
  urldate = {2026-04-28},
  langid = {english}
}

@article{kabschSolutionBestRotation1976,
  title = {A Solution for the Best Rotation to Relate Two Sets of Vectors},
  author = {Kabsch, W.},
  year = 1976,
  journal = {Acta Crystallographica Section A},
  volume = {32},
  number = {5},
  pages = {922--923},
  issn = {1600-5724},
  doi = {10.1107/S0567739476001873},
  url = {https://onlinelibrary.wiley.com/doi/abs/10.1107/S0567739476001873},
  urldate = {2026-04-17},
  langid = {english}
}

@article{kimDiffusionbasedGenerativeAI2024,
  title = {Diffusion-Based Generative {{AI}} for Exploring Transition States from {{2D}} Molecular Graphs},
  author = {Kim, Seonghwan and Woo, Jeheon and Kim, Woo Youn},
  year = 2024,
  month = jan,
  journal = {Nature Communications},
  volume = {15},
  number = {1},
  pages = {341},
  publisher = {Nature Publishing Group},
  issn = {2041-1723},
  doi = {10.1038/s41467-023-44629-6},
  url = {https://www.nature.com/articles/s41467-023-44629-6},
  urldate = {2025-12-12},
  copyright = {2024 The Author(s)},
  langid = {english}
}

@article{kissComputationalEnzymeDesign2013,
  title = {Computational {{Enzyme Design}}},
  author = {Kiss, Gert and {\c C}elebi-{\"O}l{\c c}{\"u}m, Nihan and Moretti, Rocco and Baker, David and Houk, K. N.},
  year = 2013,
  journal = {Angewandte Chemie International Edition},
  volume = {52},
  number = {22},
  pages = {5700--5725},
  issn = {1521-3773},
  doi = {10.1002/anie.201204077},
  url = {https://onlinelibrary.wiley.com/doi/abs/10.1002/anie.201204077},
  urldate = {2026-03-18},
  copyright = {Copyright \copyright{} 2013 WILEY-VCH Verlag GmbH \& Co. KGaA, Weinheim},
  langid = {english}
}

@article{kozuchHowConceptualizeCatalytic2011,
  title = {How to {{Conceptualize Catalytic Cycles}}? {{The Energetic Span Model}}},
  shorttitle = {How to {{Conceptualize Catalytic Cycles}}?},
  author = {Kozuch, Sebastian and Shaik, Sason},
  year = 2011,
  month = feb,
  journal = {Accounts of Chemical Research},
  volume = {44},
  number = {2},
  pages = {101--110},
  publisher = {American Chemical Society},
  issn = {0001-4842},
  doi = {10.1021/ar1000956},
  url = {https://doi.org/10.1021/ar1000956},
  urldate = {2026-03-18}
}

@article{leeDatasetChemicalReaction2025,
  title = {A Dataset of Chemical Reaction Pathways Incorporating Halogen Chemistry},
  author = {Lee, Minhyeok and Jeong, Jinyoung and Ashyrmamatov, Islambek and Ucaxk, Umit V. and Kim, Sunwoo and Lee, Juyong and Sim, Eunji},
  year = 2025,
  month = oct,
  journal = {Scientific Data},
  volume = {12},
  number = {1},
  pages = {1655},
  publisher = {Nature Publishing Group},
  issn = {2052-4463},
  doi = {10.1038/s41597-025-05944-3},
  url = {https://www.nature.com/articles/s41597-025-05944-3},
  urldate = {2025-12-04},
  copyright = {2025 The Author(s)},
  langid = {english}
}

@article{lindgrenScaledDynamicOptimizations2019,
  title = {Scaled and {{Dynamic Optimizations}} of {{Nudged Elastic Bands}}},
  author = {Lindgren, Per and Kastlunger, Georg and Peterson, Andrew A.},
  year = 2019,
  month = nov,
  journal = {Journal of Chemical Theory and Computation},
  volume = {15},
  number = {11},
  eprint = {1906.10257},
  primaryclass = {physics},
  pages = {5787--5793},
  issn = {1549-9618, 1549-9626},
  doi = {10.1021/acs.jctc.9b00633},
  url = {http://arxiv.org/abs/1906.10257},
  urldate = {2026-05-06},
  archiveprefix = {arXiv}
}

@misc{loshchilovDecoupledWeightDecay2019,
  title = {Decoupled {{Weight Decay Regularization}}},
  author = {Loshchilov, Ilya and Hutter, Frank},
  year = 2019,
  month = jan,
  number = {arXiv:1711.05101},
  eprint = {1711.05101},
  primaryclass = {cs},
  publisher = {arXiv},
  doi = {10.48550/arXiv.1711.05101},
  url = {http://arxiv.org/abs/1711.05101},
  urldate = {2026-05-04},
  archiveprefix = {arXiv}
}

@article{namTransferableLearningReaction2025,
  title = {Transferable {{Learning}} of {{Reaction Pathways}} from {{Geometric Priors}}},
  author = {Nam, Juno and Steiner, Miguel and Misterka, Max and Yang, Soojung and Singhal, Avni and {G{\'o}mez-Bombarelli}, Rafael},
  year = 2025,
  month = nov,
  journal = {The Journal of Physical Chemistry Letters},
  volume = {16},
  number = {45},
  pages = {11690--11699},
  publisher = {American Chemical Society},
  doi = {10.1021/acs.jpclett.5c02620},
  url = {https://doi.org/10.1021/acs.jpclett.5c02620},
  urldate = {2026-05-03}
}

@misc{nikitinRightSaddleStereochemistryAware2026,
  title = {Right into the {{Saddle}}: {{Stereochemistry-Aware Generation}} of {{Molecular Transition States}}},
  shorttitle = {Right into the {{Saddle}}},
  author = {Nikitin, Filipp and Anstine, Dylan M. and Isayev, Olexandr},
  year = 2026,
  month = apr,
  doi = {10.26434/chemrxiv.15001681/v1},
  url = {https://chemrxiv.org/doi/full/10.26434/chemrxiv.15001681/v1},
  urldate = {2026-04-09},
  langid = {english}
}

@article{schlegelGeometryOptimization2011,
  title = {Geometry Optimization},
  author = {Schlegel, H. Bernhard},
  year = 2011,
  journal = {WIREs Computational Molecular Science},
  volume = {1},
  number = {5},
  pages = {790--809},
  issn = {1759-0884},
  doi = {10.1002/wcms.34},
  url = {https://onlinelibrary.wiley.com/doi/abs/10.1002/wcms.34},
  urldate = {2026-04-27},
  copyright = {Copyright \copyright{} 2011 John Wiley \& Sons, Ltd.},
  langid = {english}
}

@article{schrammEnzymaticTransitionStates2018,
  title = {Enzymatic {{Transition States}} and {{Drug Design}}},
  author = {Schramm, Vern L.},
  year = 2018,
  month = nov,
  journal = {Chemical Reviews},
  volume = {118},
  number = {22},
  pages = {11194--11258},
  publisher = {American Chemical Society},
  issn = {0009-2665},
  doi = {10.1021/acs.chemrev.8b00369},
  url = {https://doi.org/10.1021/acs.chemrev.8b00369},
  urldate = {2026-03-18}
}

@article{schreinerNeuralNEBNeuralNetworks2022,
  title = {{{NeuralNEB}}---Neural Networks Can Find Reaction Paths Fast},
  author = {Schreiner, Mathias and Bhowmik, Arghya and Vegge, Tejs and J{\o}rgensen, Peter Bj{\o}rn and Winther, Ole},
  year = 2022,
  month = dec,
  journal = {Machine Learning: Science and Technology},
  volume = {3},
  number = {4},
  pages = {045022},
  publisher = {IOP Publishing},
  issn = {2632-2153},
  doi = {10.1088/2632-2153/aca23e},
  url = {https://doi.org/10.1088/2632-2153/aca23e},
  urldate = {2026-01-19},
  langid = {english}
}

@article{schreinerTransition1xDatasetBuilding2022,
  title = {Transition1x - a Dataset for Building Generalizable Reactive Machine Learning Potentials},
  author = {Schreiner, Mathias and Bhowmik, Arghya and Vegge, Tejs and Busk, Jonas and Winther, Ole},
  year = 2022,
  month = dec,
  journal = {Scientific Data},
  volume = {9},
  number = {1},
  pages = {779},
  publisher = {Nature Publishing Group},
  issn = {2052-4463},
  doi = {10.1038/s41597-022-01870-w},
  url = {https://www.nature.com/articles/s41597-022-01870-w},
  urldate = {2025-12-04},
  copyright = {2022 The Author(s)},
  langid = {english}
}

@misc{schuttEquivariantMessagePassing2021,
  title = {Equivariant Message Passing for the Prediction of Tensorial Properties and Molecular Spectra},
  author = {Sch{\"u}tt, Kristof T. and Unke, Oliver T. and Gastegger, Michael},
  year = 2021,
  month = jun,
  number = {arXiv:2102.03150},
  eprint = {2102.03150},
  primaryclass = {cs},
  publisher = {arXiv},
  doi = {10.48550/arXiv.2102.03150},
  url = {http://arxiv.org/abs/2102.03150},
  urldate = {2026-04-28},
  archiveprefix = {arXiv}
}

@article{smidstrupImprovedInitialGuess2014,
  title = {Improved Initial Guess for Minimum Energy Path Calculations},
  author = {Smidstrup, S{\o}ren and Pedersen, Andreas and Stokbro, Kurt and J{\'o}nsson, Hannes},
  year = 2014,
  month = jun,
  journal = {The Journal of Chemical Physics},
  volume = {140},
  number = {21},
  publisher = {AIP Publishing},
  issn = {0021-9606},
  doi = {10.1063/1.4878664},
  url = {https://pubs.aip.org/aip/jcp/article/140/21/214106/566679/Improved-initial-guess-for-minimum-energy-path},
  urldate = {2026-02-24},
  langid = {english}
}

@article{tachibanaDifferentialGeometryChemically1978,
  title = {Differential Geometry of Chemically Reacting Systems},
  author = {Tachibana, Akitomo and Fukui, Kenichi},
  year = 1978,
  month = dec,
  journal = {Theoretica chimica acta},
  volume = {49},
  number = {4},
  pages = {321--347},
  issn = {1432-2234},
  doi = {10.1007/BF00552483},
  url = {https://doi.org/10.1007/BF00552483},
  urldate = {2026-04-27},
  langid = {english}
}

@article{truhlarVariationalTransitionState,
  title = {Variational {{Transition State Theory}}},
  author = {Truhlar, D G and Garrett, B C},
  year = 1984,
  journal = {Ann. Rev. Phys. Chem.},
  doi = {10.1146/annurev.pc.35.100184.001111},
  url = {https://www.annualreviews.org/content/journals/10.1146/annurev.pc.35.100184.001111},
  langid = {english}
}

@misc{tuoFlowMatchingReaction2025,
  title = {Flow Matching for Reaction Pathway Generation},
  author = {Tuo, Ping and Chen, Jiale and Li, Ju},
  year = 2025,
  month = nov,
  number = {arXiv:2507.10530},
  eprint = {2507.10530},
  primaryclass = {physics},
  publisher = {arXiv},
  doi = {10.48550/arXiv.2507.10530},
  url = {http://arxiv.org/abs/2507.10530},
  urldate = {2025-12-04},
  archiveprefix = {arXiv}
}

@article{unkeMachineLearningForce2021,
  title = {Machine {{Learning Force Fields}}},
  author = {Unke, Oliver T. and Chmiela, Stefan and Sauceda, Huziel E. and Gastegger, Michael and Poltavsky, Igor and Sch{\"u}tt, Kristof T. and Tkatchenko, Alexandre and M{\"u}ller, Klaus-Robert},
  year = 2021,
  month = aug,
  journal = {Chemical Reviews},
  volume = {121},
  number = {16},
  pages = {10142--10186},
  publisher = {American Chemical Society},
  issn = {0009-2665},
  doi = {10.1021/acs.chemrev.0c01111},
  url = {https://doi.org/10.1021/acs.chemrev.0c01111},
  urldate = {2026-03-18}
}

@misc{wuMachineLearnedLeftmostHessian2026,
  title = {Machine-{{Learned Leftmost Hessian Eigenvectors}} for {{Robust Transition State Finding}}},
  author = {Wu, Guanchen and Yuan, Chung-Yueh and Hegazy, Kareem and Blau, Samuel M. and {Head-Gordon}, Teresa},
  year = 2026,
  month = mar,
  number = {arXiv:2603.21323},
  eprint = {2603.21323},
  primaryclass = {physics},
  publisher = {arXiv},
  doi = {10.48550/arXiv.2603.21323},
  url = {http://arxiv.org/abs/2603.21323},
  urldate = {2026-03-25},
  archiveprefix = {arXiv}
}

@misc{xuECTSUltrafastDiffusion2025,
  title = {{{ECTS}}: {{An}} Ultra-Fast Diffusion Model for Exploring Chemical Reactions with Equivariant Consistency},
  shorttitle = {{{ECTS}}},
  author = {Xu, Mingyuan and Li, Bowen and Dong, Zhaojia and Dral, Pavlo and Zhu, Tong and Chen, Hongming},
  year = 2025,
  month = mar,
  doi = {10.26434/chemrxiv-2025-f9vdp},
  url = {https://chemrxiv.org/doi/full/10.26434/chemrxiv-2025-f9vdp},
  urldate = {2026-03-01},
  copyright = {https://creativecommons.org/licenses/by-nc-nd/4.0/},
  langid = {english}
}

@article{zhuGeodesicInterpolationReaction2019,
  title = {Geodesic Interpolation for Reaction Pathways},
  author = {Zhu, Xiaolei and Thompson, Keiran C. and Mart{\'i}nez, Todd J.},
  year = 2019,
  month = apr,
  journal = {The Journal of Chemical Physics},
  volume = {150},
  number = {16},
  pages = {164103},
  issn = {0021-9606, 1089-7690},
  doi = {10.1063/1.5090303},
  url = {https://pubs.aip.org/jcp/article/150/16/164103/198363/Geodesic-interpolation-for-reaction-pathways},
  urldate = {2026-04-17},
  langid = {english}
}
